\documentclass[11pt]{article}
\usepackage{tikz,graphicx,xcolor}
\usetikzlibrary{arrows.meta,positioning,calc,shapes.geometric}
\usepackage{placeins}
\usepackage[T1]{fontenc}
\usepackage{lmodern}
\usepackage{authblk}
\usepackage[margin=1in]{geometry}
\usepackage{amsmath,amssymb,amsthm,mathtools,mathrsfs}
\usepackage{booktabs,enumitem,etoolbox,needspace,microtype}
\usepackage{cite}
\usepackage[colorlinks=true,linkcolor=blue,citecolor=blue,urlcolor=blue]{hyperref}
\newtheorem{theorem}{Theorem}[section]
\newtheorem{lemma}[theorem]{Lemma}
\newtheorem{proposition}[theorem]{Proposition}
\newtheorem{corollary}[theorem]{Corollary}
\theoremstyle{definition}
\newtheorem{definition}[theorem]{Definition}
\theoremstyle{remark}

\newtheorem*{remark*}{Remark}
\BeforeBeginEnvironment{theorem}{\Needspace{6\baselineskip}}
\BeforeBeginEnvironment{lemma}{\Needspace{6\baselineskip}}
\BeforeBeginEnvironment{proposition}{\Needspace{6\baselineskip}}
\BeforeBeginEnvironment{corollary}{\Needspace{5\baselineskip}}
\newcommand{\PM}{\operatorname{pm}}
\newcommand{\per}{\operatorname{per}}
\newcommand{\diag}{\operatorname{diag}}
\newcommand{\rank}{\operatorname{rank}}
\newcommand{\tr}{\operatorname{tr}}
\newcommand{\poly}{\operatorname{poly}}
\newcommand{\sharpP}{\mathord{\#}\mathsf P}
\newcommand{\IS}{\mathord{\#}\mathrm{IS}}
\newcommand{\UIPM}{\mathord{\#}\mathrm{UnitIntervalPM}}
\newcommand{\Word}{\mathrm{WordEval}}
\newcommand{\cS}{\mathcal S}
\newcommand{\R}{\mathsf R}
\newcommand{\D}{\mathsf D}
\newcommand{\B}{\mathsf B}
\newcommand{\E}{\mathsf E}
\newcommand{\cA}{\mathcal A}
\newcommand{\cF}{\mathcal F}
\newcommand{\cE}{\mathcal E}
\newcommand{\cL}{\mathcal L}
\newcommand{\ind}[1]{\mathbf 1_{\{#1\}}}
\newcommand{\enc}[1]{\widehat{#1}}
\DeclareMathOperator{\wt}{wt}
\allowdisplaybreaks[1]
\setlist[enumerate]{itemsep=3pt,topsep=5pt}

\newcommand{\historycell}[2]{\raggedright\strut#1\newline\strut#2\strut}
\newenvironment{historytable}{%
  \small
  \begin{tabular}{@{}p{.34\linewidth}@{\hspace{.04\linewidth}}p{.62\linewidth}@{}}
  \toprule
  \textbf{Work} & \textbf{Progress} \tabularnewline
  \midrule
}{%
  \bottomrule
  \end{tabular}
}

\title{Hidden Circuits and Exact Counting in Ordered Graphs}
\author[1]{Chenghua Liu}
\author[2]{Boning Meng}
\affil[1]{Institute of Software, Chinese Academy of Sciences, Beijing, China}
\affil[2]{University of Regensburg, Regensburg, Germany}
\affil[ ]{\texttt{liuch.russell@gmail.com},
  \texttt{mengboning2013@gmail.com}}
\date{}
\begin{document}
\maketitle
\begin{abstract}
%We prove that exact perfect-matching counting is $\sharpP$-complete
%under polynomial-time Turing reductions on simple, unweighted monotone,
%unit interval, and chordal permutation graphs.
%For distance-hereditary graphs on $n$ vertices, we give an exact
%perfect-matching counting algorithm using $O(n^2)$ arithmetic operations,
%improving the $O(n^4)$ bound obtained from the algorithm of
%Curticapean and Marx (SODA, 2016).
%Together with prior results and our inclusion of QChains in the
%distance-hereditary graphs, these results complete the exact-counting
%classification of all classes in the graph-class diagram of
%Dyer and M\"uller (SIAM Journal on Discrete Mathematics, 2019).
We show that exact perfect-matching counting remains hard under strong ordering restrictions: it is $\sharpP$-complete on each of three classes of simple, unweighted graphs defined by such restrictions---monotone graphs, unit interval graphs, and chordal permutation graphs. The monotone result settles the exact-counting complexity left open by Dyer, Jerrum, and M\"uller (JACM 2017), complementing their rapid-mixing theorem. Inspired by quantum circuits, our reductions implement a circuit simulation using globally coupled matching-transfer operators. The key construction is an exact projection, implemented by a polynomial-length sequence of normalized transfers, that restores tensor-product locality and makes encoded gates composable. Interpolation-based cancellation then reduces circuit evaluation to unweighted perfect-matching counts in all three classes. We also place Dyer and M\"uller's class QChains within the distance-hereditary graphs and give an $O(n^2)$-arithmetic-operation counting algorithm for the latter, improving the $O(n^4)$ bound obtainable from Curticapean and Marx (SODA 2016). Together with prior results, these advances complete the exact-counting classification of the graph classes in Dyer and M\"uller's diagram (SIDMA 2019).
\end{abstract}

\section*{Acknowledgments}

Boning Meng is deeply grateful to Yixin Cao for their early discussions
of this problem. These exchanges greatly deepened his understanding of
the problem and provided invaluable inspiration for its solution.
We invited him to join us as a coauthor, but he modestly declined.

The authors used OpenAI's ChatGPT in preparing
this manuscript, including for language editing and \LaTeX{} preparation.
The authors also used ChatGPT to explore proof ideas and develop arguments.
All claims and proofs were
independently checked and finalized by the authors, who take full
responsibility for the content.

\section{Introduction}\label{sec:intro}

A perfect matching partitions a set of objects into compatible pairs.
Finding one such partition is a classical polynomial-time problem.
To study all such partitions, we may count them exactly, approximate
their number, or sample one from a nearly uniform distribution.
Approximate counting and nearly uniform sampling are closely related
through self-reduction, provided the required smaller instances remain
accessible to the algorithms~\cite{JVV86,SJ89}.
Samples estimate the fraction of perfect matchings containing a chosen
edge, and suitable estimates on successively smaller graphs can be
combined to approximate the total count.
Conversely, approximate counts of the possible completions guide the
probabilities for choosing edges when generating a random matching.

For unrestricted bipartite graphs, exact counting is
$\sharpP$-complete~\cite{Valiant79}, whereas the permanent algorithm of
Jerrum, Sinclair, and Vigoda gives a fully polynomial randomized
approximation scheme (FPRAS)~\cite{JSV04}.
For general nonbipartite graphs, the existence of an FPRAS remains
open~\cite{SVW18}, although positive results cover several restricted
families. These include graphs of even order $n$ and minimum degree
at least $n/2$~\cite{JS89}, certain classes with restricted odd-cycle
structure~\cite{SVW18}, and regular expander families satisfying suitable
degree and spectral conditions~\cite{ENO22}.
Many such algorithms use sampling, but approximation need not proceed
this way. For example, Barvinok gives deterministic quasipolynomial-time
relative approximations for weighted perfect-matching counts when all
edge weights of a complete graph lie between two fixed positive
constants~\cite{Barvinok17}.

Our focus is on which structural restrictions make \emph{exact} counting
tractable. Planarity permits exact counting through the FKT
method~\cite{Kasteleyn61,TemperleyFisher61}, and nested neighborhoods give
exact formulas on chain graphs~\cite{OUU10}.
Thilikos and Wiederrecht~\cite{ThilikosWiederrecht24} recently obtained a
complete exact-counting classification for minor-closed graph classes.
Here we investigate structure arising from the order of vertices.

Such \emph{ordering restrictions} require that vertices admit one or
more orderings in which adjacency follows specified rules.
For example, when assigning objects to positions, the allowed positions
for each object may form a consecutive interval, with both interval
boundaries moving only to the right as the objects are considered in
order. The resulting compatibility graphs are monotone graphs, and
their perfect matchings are restricted permutations, a model used in
statistical permutation tests~\cite{DGH01}.
A nonbipartite example comes from pairing equal-length intervals on a
line, with overlap determining compatibility. Ordering the intervals by
their left endpoints constrains which pairs can overlap, but does not
require pairs to join two prescribed sides of a bipartition.
These rules organize adjacency while still allowing dense graphs.

For these ordered classes, existing sampling results also yield
approximate counting algorithms.
Building on the restricted-permutation work of Diaconis, Graham, and
Holmes~\cite{DGH01}, Dyer, Jerrum, and M\"uller~\cite{DJM17} proved rapid
mixing of the switch chain on monotone graphs. This chain samples
perfect matchings by repeatedly exchanging two matched pairs.
Dyer and M\"uller~\cite{DM19} extended the approach to quasimonotone
graphs, a nonbipartite class containing all unit interval graphs.
These classes are closed under vertex deletion, so the sampling
algorithms remain available on the smaller graphs needed for
self-reduction. Despite these sampling and approximation guarantees,
the exact-counting questions raised in~\cite{OUU10,DJM17} remained
unresolved.

We prove that exact perfect-matching counting is $\sharpP$-complete on
monotone graphs and unit interval graphs, and establish hardness even
on chordal permutation graphs, which admit both interval and
permutation representations. We also obtain an exact algorithm for
Dyer and M\"uller's QChains class~\cite[Section~3.4]{DM19} by placing it within the
distance-hereditary graphs.
The three hardness proofs share a quantum-circuit-inspired construction.

\subsection{Main results}\label{sec:intro-results}

On the hardness side, we prove the following.

\begin{theorem}\label{thm:main}
Exact perfect-matching counting is $\sharpP$-complete under $\le_T$
on each of the following classes:
\begin{enumerate}[label=(\roman*)]
\item monotone graphs, equivalently bipartite permutation graphs;
\item unit interval graphs;
\item chordal permutation graphs.
\end{enumerate}
\end{theorem}

% The complete TikZ landscape is inserted here by the manuscript builder.
% Exact-counting overlay on Dyer--Mueller (2019), Appendix Fig. A.1.
% Topology and Bezier control points transcribed from the authors' public
% NonbipartiteArxivV3.tex, arXiv:1705.05790v3, lines 2247--2319.
% Public source: https://arxiv.org/src/1705.05790
% Journal accepted PDF: https://eprints.whiterose.ac.uk/id/eprint/137388/7/Dyer%20NonbipartiteSIDMA-final.pdf
% Requirements in the parent preamble: \usepackage{tikz,graphicx,xcolor}.
% Citation key: DM19.  No original sampling, mixing, or P-stability symbols
% are retained. All 31 nodes remain; the erroneous Cograph--OddChordal
% containment is omitted, as explained in the caption's footnote.
\begingroup
\definecolor{pmOldFP}{HTML}{0072B2}
\definecolor{pmOldHard}{HTML}{D55E00}
\definecolor{pmNewFP}{HTML}{009E73}
\definecolor{pmNewHard}{HTML}{AA4499}
\tikzset{
  pmclass/.style={rectangle,rounded corners=2pt,fill=white,
    inner xsep=4pt,inner ysep=3pt,line width=1.0pt},
  pmpriorfp/.style={draw=pmOldFP},
  pmpriorhard/.style={draw=pmOldHard},
  pmnewfp/.style={draw=pmNewFP,line width=1.3pt},
  pmnewhard/.style={draw=pmNewHard,line width=1.3pt}
}
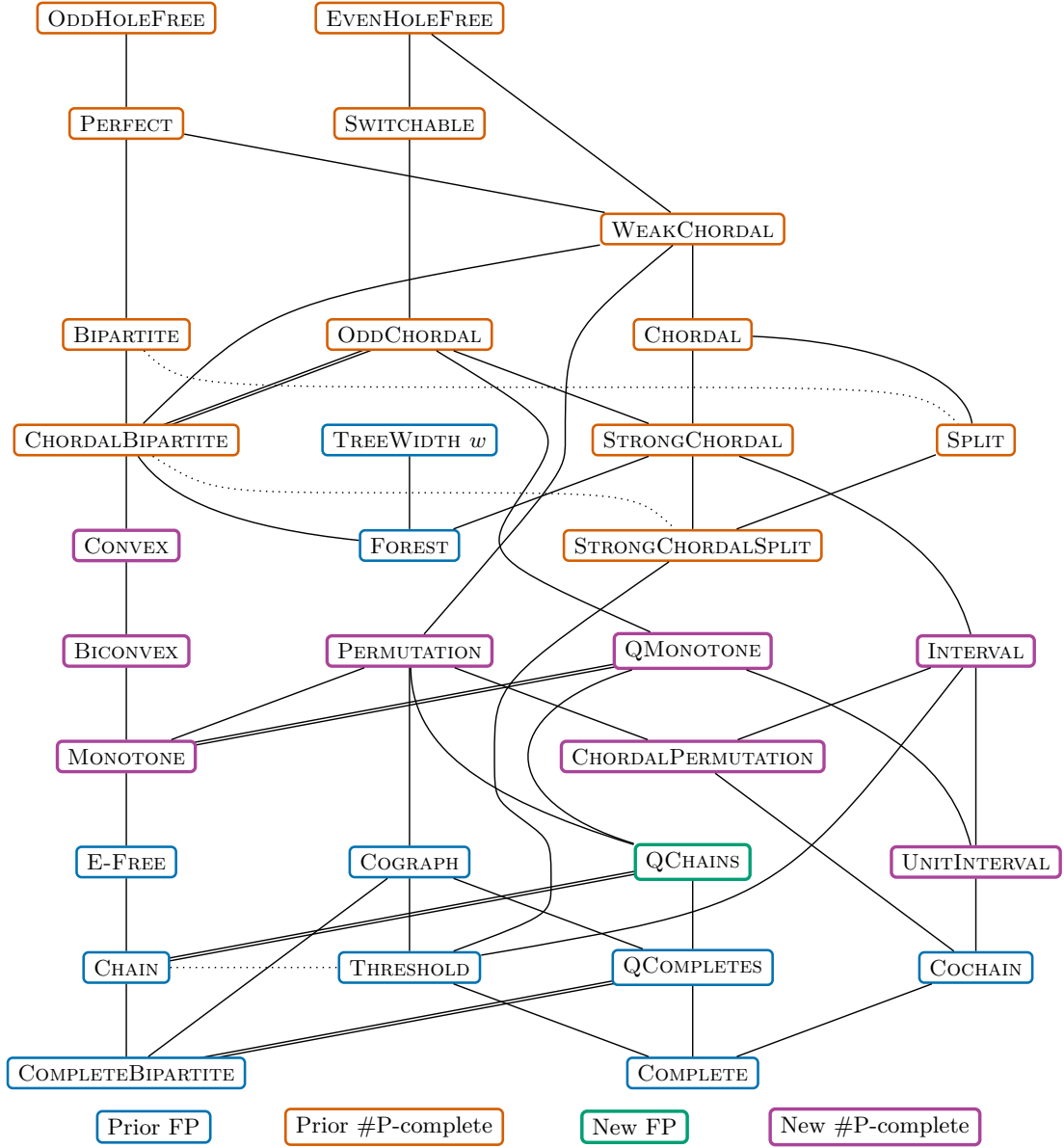
\begin{figure}[htbp]
  \begin{minipage}{\textwidth}
  \renewcommand{\thempfootnote}{\arabic{mpfootnote}}
  \centering
  \resizebox{.885\textwidth}{!}{%
  \begin{tikzpicture}[xscale=4.0,yscale=1.49,font=\footnotesize\scshape,line width=.55pt]
    \draw[double] (1,6)--(2,7)  (1,3)--(3,4)  (1,1)--(3,2)  (1,0)--(3,1);
    \node[pmclass,pmpriorhard] at (1,10) (oh-f) {OddHoleFree};
    \node[pmclass,pmpriorhard] at (2,10) (eh-f) {EvenHoleFree};
    \node[pmclass,pmpriorhard] at (1,9) (perf) {Perfect};
    \node[pmclass,pmpriorhard] at (2,9) (noch) {Switchable}; \draw (eh-f)--(noch);
    \node[pmclass,pmpriorhard] at (3,8) (wkch) {WeakChordal};        \draw (eh-f)--(wkch)--(perf);
    \node[pmclass,pmpriorhard] at (1,7) (bip) {Bipartite};  \draw (oh-f)--(perf)--(bip);
    \node[pmclass,pmpriorhard] at (2,7) (oddch) {OddChordal};         \draw (noch)--(oddch);
    \node[pmclass,pmpriorhard] at (3,7) (chord) {Chordal};       \draw (wkch)--(chord);
    \node[pmclass,pmpriorhard] at (1,6) (chbip) {ChordalBipartite};  \draw (bip)--(chbip);
    \draw (wkch)  .. controls (1.5,7.3) .. (chbip);
    \node[pmclass,pmpriorfp] at (2,6) (tw<w) {TreeWidth $w$};
    \node[pmclass,pmpriorhard] at (3,6) (strch) {StrongChordal};       \draw (oddch)--(strch)--(chord);
    \node[pmclass,pmpriorhard] at (4,6) (split) {Split};
    \draw (chord) to[bend left=40] (split);
    \node[pmclass,pmnewhard] at (1,5) (conv) {Convex};        \draw (chbip)--(conv);
    \node[pmclass,pmpriorfp] at (2,5) (tree) {Forest};        \draw (tw<w)--(tree)--(strch);
    \draw (chbip) to[bend right] (tree);
    \node[pmclass,pmpriorhard] at (3,5) (scs) {StrongChordalSplit};
                                                  \draw (strch)--(scs)--(split);
    \node[pmclass,pmnewhard] at (1,4) (bic) {Biconvex};  \draw (conv)--(bic);
    \node[pmclass,pmnewhard] at (2,4) (perm) {Permutation};
    \node[pmclass,pmnewhard] at (3,4) (qmon) {QMonotone};
    \node[pmclass,pmnewhard] at (4,4) (int) {Interval};
    \draw (strch) to[bend left=20] (int);
    \node[pmclass,pmnewhard] at (1,3) (mono) {Monotone};        \draw (bic)--(mono)--(perm);
    \node[pmclass,pmnewhard] at (3,3) (chp) {ChordalPermutation};  \draw (perm)--(chp)--(int);
    \node[pmclass,pmpriorfp] at (1,2) (Efcb) {E-Free}; \draw (mono)--(Efcb);
    \node[pmclass,pmpriorfp] at (2,2) (cogr) {Cograph};\draw (perm)--(cogr);
    \node[pmclass,pmnewfp] at (3,2) (qchs) {QChains};
    \draw (perm) to[bend right=25] (qchs);
    \draw (qmon) to[bend right=50] (qchs);
    \node[pmclass,pmnewhard] at (4,2) (unii) {UnitInterval};        \draw (int)--(unii);
    \draw (unii) to[bend right=20] (qmon);
    \node[pmclass,pmpriorfp] at (1,1) (chain) {Chain};\draw (Efcb)--(chain);
    \node[pmclass,pmpriorfp] at (2,1) (thre) {Threshold}; \draw (cogr)--(thre);
    \draw (int)   .. controls (3.3,1.6) .. (thre);
    \draw (scs)   .. controls (2.3,3.7) .. (2.3,3.0)
                  .. controls (2.3,2.6) .. (2.4,2.4)
                  .. controls (2.5,2.2) .. (2.5,1.8)
                  .. controls (2.5,1.5) .. (thre);
    \node[pmclass,pmpriorfp] at (3,1) (qcbs) {QCompletes};  \draw (qchs)--(qcbs)--(cogr);
    \node[pmclass,pmpriorfp] at (4,1) (c-chn) {Cochain}; \draw (unii)--(c-chn)--(chp);
    \node[pmclass,pmpriorfp] at (1,0) (cb) {CompleteBipartite};          \draw (chain)--(cb)--(cogr);
    \node[pmclass,pmpriorfp] at (3,0) (compl) {Complete};       \draw (qcbs)--(compl);
    \draw (thre)--(compl)--(c-chn);
    \draw (wkch)  .. controls (2.55,7.0) .. (2.55,6.0)
                  .. controls (2.55,5.7) .. (perm);
    \draw (oddch) .. controls (2.55,6.1) .. (2.4,5.5)
                  .. controls (2.25,4.9) .. (qmon);
    \draw[dotted]   (bip) .. controls (1.2,6.5) .. (2.5,6.5)
                          .. controls (3.8,6.5) .. (split);
    \draw[dotted] (chbip) .. controls (1.3,5.5) .. (2.0,5.5)
                          .. controls (2.8,5.5) .. (scs);
    \draw[dotted] (chain)--(thre);
  \end{tikzpicture}
  }
  \par\medskip
  \begin{tikzpicture}[font=\footnotesize]
    \node[pmclass,pmpriorfp] at (0,0) {Prior FP};
    \node[pmclass,pmpriorhard] at (3.3,0) {Prior $\#\mathrm P$-complete};
    \node[pmclass,pmnewfp] at (6.6,0) {New FP};
    \node[pmclass,pmnewhard] at (9.9,0) {New $\#\mathrm P$-complete};
  \end{tikzpicture}
  \caption[Exact perfect-matching counting in the graph-class diagram.]{\label{fig:exact-landscape}
  Exact perfect-matching counting on the classes in
  Dyer and M\"uller~\cite[Appendix, Fig.~A.1]{DM19}.
  The original node positions and all relations except the erroneous
  Cograph--OddChordal inclusion\protect\hyperlink{fn:exact-landscape-correction}{\protect\footnotemark[1]}
  are retained. Frame colors
  distinguish prior results from the classifications established here,
  including their consequences by inclusion.
  Ordinary and double lines run from a smaller class below to a larger
  class above.  Double lines specifically denote inclusion in the
  quasi-class after closure under disjoint union.
  Dotted lines denote linked bipartite/split classes, obtained by completing
  one bipartition side to a clique. These dotted links do not denote
  inclusions.
  Only exact counting is annotated.
  We retain the source's convention that \textsc{EvenHoleFree} forbids
  induced even cycles of length at least six, and $w\ge1$ is fixed in
  \textsc{TreeWidth}~$w$.  All completeness assertions use polynomial-time
  Turing reductions.}
  \footnotetext[1]{\hypertarget{fn:exact-landscape-correction}{}%
  We omit one erroneous inclusion from the source diagram,
  $\mathrm{Cograph}\subseteq\mathrm{OddChordal}$.
  The same paper gives a counterexample~\cite[Section~2, Fig.~2.7]{DM19}:
  join every vertex of two disjoint edges to two independent vertices.
  This is a cograph, but it has a spanning six-cycle whose chords
  all join positions of the same parity, so it is not odd chordal.
  This correction does not change any of the counting classifications.}
  \end{minipage}
\end{figure}
\endgroup

Here $\le_T$ denotes polynomial-time Turing reductions. All query
graphs are simple and unweighted, and the reductions supply valid
interval or permutation representations. Thus the hardness holds even
when these representations are part of the input. By inclusion, the theorem
also settles interval, permutation, convex, biconvex, and
quasimonotone graphs. The chordal permutation result is stronger
than either superclass conclusion alone: it establishes hardness
on the intersection of interval and permutation graphs.

These results resolve the open exact-counting problems for interval
and unit interval graphs in~\cite{OUU10,KOU11}, and for convex,
biconvex, and monotone graphs in~\cite[Section~2.1]{DJM17}.
Their significance for sampling is developed in
Section~\ref{sec:intro-impact}, followed by the common circuit-based
mechanism in Section~\ref{sec:intro-methods}.

On the tractable side, we prove
$\mathsf{QChains}\subsetneq\mathsf{DH}$
(Lemma~\ref{lem:qchains-dh}), where $\mathsf{DH}$ denotes the
distance-hereditary graphs and QChains is the quasi-class of
disjoint unions of chain graphs defined by Dyer and
M\"uller~\cite{DM19}. This inclusion places QChains within the
scope of existing bounded-clique-width counting
algorithms~\cite{GolumbicRotics00,MRAG06,CurticapeanMarx16}.
We also give a direct algorithm on all distance-hereditary graphs:
after $O(n+m)$ graph preprocessing, it uses $O(n^2)$ integer
arithmetic operations on $O(n\log n)$-bit integers
(Theorem~\ref{thm:qchains-fp}), improving the arithmetic-operation
count over the generic $O(n^4)$ bound~\cite[Theorem~1.3]{CurticapeanMarx16}
obtained from a $3$-expression~\cite{GolumbicRotics00}.

Together with prior results, these theorems complete the
polynomial-time versus $\sharpP$-complete classification of the
classes in Dyer and M\"uller's graph-class diagram~\cite{DM19}.
Figure~\ref{fig:exact-landscape} records this classification,
distinguishing direct results and their consequences as explained
in the text. The new hard cases follow from the theorem, and the
remaining tractable case follows from the QChains inclusion.
The prior classifications include consequences of graph-class inclusions
and the perfect-matching-preserving transformations between linked
classes described in~\cite[Appendix]{DM19}. The underlying hardness
results and exact algorithms appear in~\cite{OUU10,DM15,DM19}.

\subsection{Impact on exact counting, sampling, and approximation}
\label{sec:intro-impact}

Our results answer three different structural questions about perfect
matchings: whether consecutive allowed positions simplify counting
permutations, whether one-dimensional geometry simplifies counting
pairings, and whether two simultaneous ordered representations suffice
for an exact algorithm. Each question has its own history, and each of
our direct hardness theorems supplies a missing answer. We describe these
connections separately, before turning to the tractable boundary.

All completeness statements below use polynomial-time Turing reductions
and concern simple, unweighted inputs. A \emph{fully polynomial
almost-uniform sampler} (FPAUS) samples within total variation distance
$\varepsilon$ of uniform in time polynomial in the input length and
$\log(1/\varepsilon)$. An FPRAS estimates a count within relative error
$\varepsilon$ with probability at least $3/4$, in time polynomial in the
input length and $1/\varepsilon$~\cite{JVV86,SJ89}.
The approximation and sampling guarantees below are prior results or
their standard self-reduction consequences. Our new conclusions concern
exact complexity. For sampling, we first decide in polynomial time
whether a perfect matching exists. If none exists, the count is zero.
Otherwise, the sampling guarantees apply. The switch chains used
below have the uniform distribution on perfect matchings as their
stationary distribution. Their rapid-mixing bounds imply that, starting
from any perfect matching found in polynomial time, a number of steps
polynomial in the input length and $\log(1/\varepsilon)$ suffices to
reach total variation distance at most $\varepsilon$ from uniform.
This gives the FPAUS guarantees~\cite{DJM17,DM19}.

\subsubsection{Restricted permutations and monotone constraints}
\label{sec:impact-permutations}

Permutation tests compare observed data with random reassignments, but
observation constraints can rule out most reassignments. To calibrate such
a test, one can sample uniformly from the admissible permutations.
Their number is the normalization constant of this distribution. Diaconis,
Graham, and Holmes~\cite{DGH01} studied this problem for truncated data,
where the restrictions naturally take the form of intervals of allowed
positions.

Concretely, object $i$ may occupy a position in
$[a_i,b_i]\subseteq[n]$, and every position must be used exactly once.
The counting problem is
\[
 Z(a,b)=\#\{\pi\in S_n:a_i\leq\pi(i)\leq b_i\text{ for every }i\}.
\]
If both endpoint sequences are nondecreasing, the allowed-position
matrix defines a \emph{monotone graph}, equivalently a bipartite
permutation graph~\cite{DM15,DJM17}. The corresponding $0$--$1$ matrix has consecutive
ones in every row, with both row boundaries nondecreasing. It is natural to ask whether these restrictions permit
an exact formula or a polynomial-time dynamic program, as happens for
smaller families with nested allowed sets~\cite{OUU10,DM15}.

The general permanent FPRAS~\cite{JSV04,BSVV08} already applies to
these inputs, but exploiting their order allows a much simpler sampler:
exchange two assigned
positions if the resulting permutation remains admissible. This is the
\emph{switch chain} on perfect matchings. Dyer, Jerrum, and
M\"uller~\cite{DJM17} proved that it mixes in polynomial time on monotone
graphs. Exact counting on this class was left open
in~\cite{OUU10,DJM17}. Table~\ref{tab:history-monotone} summarizes the
algorithmic results alongside our hardness theorem. Combining the
prior sampling and approximation guarantees with our exact-counting
result gives the following contrast for restricted permutations.

\begin{corollary}[Monotone restricted permutations]
\label{cor:monotone-sampling}
For nondecreasing endpoint sequences $a$ and $b$, computing $Z(a,b)$
exactly is $\sharpP$-complete under polynomial-time Turing reductions.
In contrast, prior results give an FPAUS for admissible permutations
and an FPRAS for $Z(a,b)$.
\end{corollary}

The FPAUS follows from the rapid-mixing bound of
Dyer, Jerrum, and M\"uller~\cite[Section~3]{DJM17}. The FPRAS follows either
from the general nonnegative-permanent algorithm of Jerrum, Sinclair,
and Vigoda~\cite{JSV04}, subsequently improved by Bez\'akov\'a et
al.~\cite{BSVV08}, or from sampling and self-reducibility~\cite{JVV86}.
For the latter route, fixing one assignment deletes its object and
position, leaving another monotone instance. Thus the same ordered
structure supports a rapidly mixing local sampler but retains the full
difficulty of exact counting. A polynomial-time exact algorithm for
these instances would imply $\mathsf{FP}=\sharpP$.

Relaxing the monotonicity restriction gives \emph{convex bipartite
graphs}: each object still has consecutive allowed positions. If the
objects can also be ordered so that the objects allowed at each position
are consecutive, the graph is \emph{biconvex}. Hence
\[
 \mathrm{Monotone}\subseteq\mathrm{Biconvex}\subseteq\mathrm{Convex}.
\]
Our theorem settles exact counting for both larger families as well:
even one-sided or two-sided consecutiveness does not suffice for exact
tractability. Both retain the JSV approximation guarantee, although the
basic switch chain can mix exponentially slowly on biconvex
instances~\cite{DM15,DJM17}. Thus slow mixing of the switch chain is
compatible with efficient approximation by other algorithms.

\begin{table}[htbp]
\centering
\begin{historytable}
\multicolumn{2}{@{}l@{}}{\emph{Prior exact-counting results}} \tabularnewline
\addlinespace[.35em]
\historycell{Okamoto et al.}{\footnotesize (WG proceedings, 2010)~\cite{OUU10}}
& \historycell{$\sharpP$-complete on chordal bipartite graphs.}{Applies to a superclass of monotone graphs.} \tabularnewline
\addlinespace[.55em]
\historycell{Dyer--M\"uller}{\footnotesize (Ann. Fac. Sci. Toulouse, 2015)~\cite{DM15}}
& \historycell{Polynomial-time exact counting on E-Free graphs.}{Applies to the E-Free subclass of monotone graphs.} \tabularnewline
\midrule
\multicolumn{2}{@{}l@{}}{\emph{Prior sampling and approximation results}} \tabularnewline
\addlinespace[.35em]
\historycell{Jerrum et al.}{\footnotesize (JACM, 2004)~\cite{JSV04}}
& \historycell{FPRAS for all bipartite graphs.}{Covers monotone, biconvex, and convex inputs.} \tabularnewline
\addlinespace[.55em]
\historycell{Bez\'akov\'a et al.}{\footnotesize (SICOMP, 2008)~\cite{BSVV08}}
& \historycell{Faster approximation of the permanent.}{Uses improved simulated-annealing schedules.} \tabularnewline
\addlinespace[.55em]
\historycell{Dyer et al.}{\footnotesize (JACM, 2017)~\cite{DJM17}}
& \historycell{Rapid mixing on monotone graphs.}{Gives an FPAUS using local switches.} \tabularnewline
\midrule
\historycell{\textbf{This work}}{\footnotesize\emph{Exact classification}}
& \historycell{$\sharpP$-complete on monotone graphs.}{Biconvex and convex hardness follow.} \tabularnewline
\end{historytable}
\caption{Exact-counting results around monotone graphs, together with
prior sampling and approximation guarantees. E-Free is the subclass
shown in Figure~\ref{fig:exact-landscape}.}
\label{tab:history-monotone}
\end{table}

\subsubsection{Equal-length intervals and nonbipartite sampling}
\label{sec:impact-intervals}

A basic test of geometric tractability is whether a pairing problem
becomes easy when compatibility depends only on positions along a line.
Equal-length intervals impose a particularly strong restriction: all
objects have the same geometric extent, and two may be paired precisely
when their intervals overlap. Does this structure, which permits a
left-to-right description of the entire input, also permit efficient
exact counting?

The input is an even number of equal-length intervals, and the goal is
to count their partitions into intersecting pairs. Its graph is a
\emph{unit interval graph}, also called a proper interval graph~\cite{DM19}.
Unlike a restricted permutation, a pairing here need not match objects
from two predetermined sides: three intervals can overlap pairwise.
Consequently the general bipartite-permanent FPRAS does not by itself
answer the sampling or approximation question for this class.
The exact-counting questions for unit interval and interval graphs
were explicitly left open in~\cite{OUU10,KOU11}.

A graph is \emph{quasimonotone} if, for every partition of its vertices
into two sides, retaining only crossing edges produces a monotone graph.
This condition supplies ordered structure on every cut without
requiring the original graph to be bipartite. The class contains all
unit interval graphs~\cite[Section~3.2.1]{DM19} and can be recognized
in polynomial time~\cite{DMQ19}.
The following two corollaries contrast our exact-counting results with
the prior sampling and approximation guarantees for these classes.

\begin{corollary}[Pairings of equal-length intervals]
\label{cor:unitinterval-sampling}
Given an equal-length interval representation, counting partitions
into intersecting pairs is $\sharpP$-complete under polynomial-time Turing
reductions. By comparison, these pairings admit an FPAUS, and their
number admits an FPRAS.
\end{corollary}

\begin{corollary}[The quasimonotone extension]
\label{cor:qmonotone-sampling}
Exact perfect-matching counting is $\sharpP$-complete on quasimonotone
graphs under polynomial-time Turing reductions. Nevertheless, this
class admits an FPAUS for perfect matchings and an FPRAS for their number.
\end{corollary}

Both corollaries use the same prior algorithmic results. Dyer and
M\"uller~\cite[Section~3.3]{DM19} proved rapid switch-chain mixing on
quasimonotone graphs, giving an FPAUS that also applies to unit interval
graphs. Both classes are closed under vertex deletion, so fixing a
matched edge and deleting its endpoints preserves the input class.
Self-reduction therefore gives the FPRAS guarantees~\cite{JVV86}.
Exact hardness follows in the other direction: our unit interval
theorem implies hardness for the containing quasimonotone class.
It also implies hardness for general interval graphs, without extending
the sampling guarantee to that larger class.

Table~\ref{tab:history-intervals} places the result alongside prior exact
counting on cochain graphs, a subclass of unit interval graphs, and
prior hardness on strongly chordal split graphs, a related branch of
Figure~\ref{fig:exact-landscape}.

\begin{table}[htbp]
\centering
\begin{historytable}
\multicolumn{2}{@{}l@{}}{\emph{Prior exact-counting results}} \tabularnewline
\addlinespace[.35em]
\historycell{Okamoto et al.}{\footnotesize (WG proceedings, 2010)~\cite{OUU10}}
& \historycell{$\sharpP$-complete on strongly chordal split graphs.}{Follows via linked classes~\cite[Appendix]{DM19}.} \tabularnewline
\addlinespace[.55em]
\historycell{Okamoto et al.}{\footnotesize (WG proceedings, 2010)~\cite{OUU10}}
& \historycell{Polynomial-time exact counting on cochain graphs.}{Applies to a subclass of unit interval graphs.} \tabularnewline
\addlinespace[.55em]
\historycell{Dyer--M\"uller}{\footnotesize (SIDMA, 2019)~\cite{DM19}}
& \historycell{Exact counting on cochain graphs in $O(n^2)$ arithmetic operations.}{Uses matching counts in chain graphs.} \tabularnewline
\midrule
\multicolumn{2}{@{}l@{}}{\emph{Prior sampling and approximation results}} \tabularnewline
\addlinespace[.35em]
\historycell{Dyer--M\"uller}{\footnotesize (SIDMA, 2019)~\cite{DM19}}
& \historycell{Rapid mixing on quasimonotone graphs.}{Includes all unit interval graphs.} \tabularnewline
\midrule
\historycell{\textbf{This work}}{\footnotesize\emph{Exact classification}}
& \historycell{$\sharpP$-complete on unit interval graphs.}{Hardness holds with an equal-length representation supplied.} \tabularnewline
\end{historytable}
\caption{Prior exact-counting results near unit interval graphs and
the sampling guarantees on quasimonotone graphs. The new unit interval
hardness also settles exact counting on quasimonotone graphs.}
\label{tab:history-intervals}
\end{table}

\subsubsection{Simultaneous interval and permutation representations}
\label{sec:impact-intersection}

When one structural restriction is insufficient for an exact algorithm,
a natural next question is whether imposing a second restriction changes
the answer. Hardness on two separate graph classes does not imply
hardness on their intersection: the instances witnessing each result may
violate the other restriction. Our chordal permutation result answers
this question for two basic geometric descriptions of a graph.

An \emph{interval graph} represents vertices by intervals on a line,
with adjacency meaning overlap. A \emph{permutation graph} represents
vertices by segments between two parallel lines, with adjacency meaning
crossing. The two endpoint orders specify all adjacencies. Graphs
admitting both descriptions are precisely \emph{chordal permutation
graphs}~\cite{DM19}. Their perfect matchings can be viewed either as
partitions into overlapping interval pairs or as partitions into
crossing segment pairs, with both representations describing the same
compatibilities.

Prior hardness on chordal graphs~\cite{OUU10}, including the strongly
chordal split branch in Figure~\ref{fig:exact-landscape}, did not settle
this intersection. On the tractable side, cochain graphs lie within the
intersection, while cographs form a polynomial-time solvable subclass
of permutation graphs~\cite{OUU10,DM19}. Dyer and
M\"uller~\cite[Section~4]{DM19} also showed
that the switch chain can mix exponentially slowly even here.
Their slow-mixing family $G_k$ nevertheless has the explicit count $\PM(G_k)=2\cdot3^k-1$~\cite[Section~4.1.2]{DM19}: a bottleneck for one sampling process need
not be an obstacle to exact counting. Table~\ref{tab:history-intersection}
separates the prior exact-counting results from the switch-chain result.

\begin{table}[htbp]
\centering
\begin{historytable}
\multicolumn{2}{@{}l@{}}{\emph{Prior exact-counting results}} \tabularnewline
\addlinespace[.35em]
\historycell{Okamoto et al.}{\footnotesize (WG proceedings, 2010)~\cite{OUU10}}
& \historycell{$\sharpP$-complete on chordal graphs.}{Applies to a superclass of the intersection.} \tabularnewline
\addlinespace[.55em]
\historycell{Okamoto et al.}{\footnotesize (WG proceedings, 2010)~\cite{OUU10}}
& \historycell{$\sharpP$-complete on strongly chordal split graphs.}{Follows via linked classes~\cite[Appendix]{DM19}.} \tabularnewline
\addlinespace[.55em]
\historycell{Dyer--M\"uller}{\footnotesize (SIDMA, 2019)~\cite{DM19}}
& \historycell{Exact counting on cochain graphs in $O(n^2)$ arithmetic operations.}{Applies to a subclass of the intersection.} \tabularnewline
\addlinespace[.55em]
\historycell{Dyer--M\"uller}{\footnotesize (SIDMA, 2019)~\cite{DM19}}
& \historycell{Polynomial-time exact counting on cographs.}{Uses a cotree for this permutation subclass.} \tabularnewline
\midrule
\multicolumn{2}{@{}l@{}}{\emph{Prior sampling and approximation results}} \tabularnewline
\addlinespace[.35em]
\historycell{Dyer--M\"uller}{\footnotesize (SIDMA, 2019)~\cite{DM19}}
& \historycell{Exponential switch-chain mixing time.}{Occurs on a family of chordal permutation graphs.} \tabularnewline
\midrule
\historycell{\textbf{This work}}{\footnotesize\emph{Exact classification}}
& \historycell{$\sharpP$-complete on chordal permutation graphs.}{Hardness holds with both representations supplied.} \tabularnewline
\end{historytable}
\caption{Prior exact-counting results on related chordal and permutation
classes, the slow-mixing result, and hardness on the interval/permutation
intersection. Cochain graphs lie in this intersection. Cographs need not.}
\label{tab:history-intersection}
\end{table}

Theorem~\ref{thm:main}(iii) establishes $\sharpP$-completeness within the
intersection itself. Thus even the simultaneous availability of the
two ordered descriptions does not make exact counting tractable unless
$\mathsf{FP}=\sharpP$. This is additional to our superclass
classifications: interval hardness already follows from the unit
interval theorem, and permutation hardness from the monotone theorem.
The third result locates hard instances satisfying \emph{both}
restrictions. This exact-counting result is independent of the earlier
slow-mixing construction. Neither result establishes approximation
hardness for the entire class.

\subsubsection{Nested choices, exact algorithms, and sampling oracles}
\label{sec:impact-qchains}

Stronger ordering restrictions can still permit exact counting.
In a \emph{chain graph}, the allowed partner sets on each bipartition
side are ordered by inclusion, enabling simple exact counting~\cite{OUU10}.
Does extending this model to nonbipartite graphs preserve efficient
exact counting, or only efficient sampling?

Dyer and M\"uller~\cite[Section~3.4]{DM19} introduced the relevant class,
\emph{QChains}: every vertex bipartition, after retaining crossing edges,
must be a disjoint union of chain graphs. They gave a polynomial-time
recognition algorithm, but did not obtain a better counting algorithm
than the approximate one available for QMonotone. Our inclusion
$\mathsf{QChains}\subsetneq\mathsf{DH}$ answers the exact question
positively. Here a graph is distance-hereditary if every connected
induced subgraph preserves the distances between its vertices~\cite{BandeltMulder86}.
Known bounded-clique-width algorithms already count its perfect
matchings in polynomial time~\cite{GolumbicRotics00,MRAG06,CurticapeanMarx16}.
Our direct pendant/twin algorithm further reduces the arithmetic-operation
count on this larger class (Theorem~\ref{thm:qchains-fp}).
Exact uniform sampling then follows by self-reduction~\cite{JVV86}, in expected
polynomial time in the random-bit model.

Table~\ref{tab:history-qchains} relates our results to earlier counting
algorithms and to the sampling-oracle framework of
Alimohammadi et al.~\cite{AASV21}.

\begin{table}[htbp]
\centering
\begin{historytable}
\historycell{Curticapean--Marx}{\footnotesize (SODA, 2016)~\cite{CurticapeanMarx16}}
& \historycell{Exact counting on DH in $O(n^4)$ arithmetic operations.}{Uses a $3$-expression~\cite{GolumbicRotics00}.} \tabularnewline
\addlinespace[.55em]
\historycell{Dyer--M\"uller}{\footnotesize (SIDMA, 2019)~\cite{DM19}}
& \historycell{Efficient sampling on QChains.}{Follows from rapid mixing on QMonotone.} \tabularnewline
\addlinespace[.55em]
\historycell{Alimohammadi et al.}{\footnotesize (STOC, 2021)~\cite{AASV21}}
& \historycell{FPRAS for fixed-size matchings in planar graphs.}{Uses exact perfect-matching counts on induced subgraphs.} \tabularnewline
\midrule
\historycell{\textbf{This work}}{\footnotesize\emph{Structural classification}}
& \historycell{Polynomial-time exact counting on QChains.}{Follows from $\mathsf{QChains}\subsetneq\mathsf{DH}$.} \tabularnewline
\addlinespace[.55em]
\historycell{\textbf{This work}}{\footnotesize\emph{Algorithmic improvement}}
& \historycell{Exact counting on DH in $O(n^2)$ arithmetic operations.}{Uses pendant and twin reductions.} \tabularnewline
\end{historytable}
\caption{Exact tractability and sampling oracles: a new QChains
classification and a faster algorithm on distance-hereditary graphs
($\mathsf{DH}$).}
\label{tab:history-qchains}
\end{table}

Exact counting can also serve as a subroutine for sampling a larger
space of matchings. Alimohammadi, Anari, Shiragur, and
Vuong~\cite{AASV21} give an FPRAS for fixed-size matchings in planar
graphs using a rapidly mixing walk on sets of unmatched vertices.
In the unweighted setting, the weight of a set $S$ is
$\mu(S)=\PM(G[V\setminus S])$. Their implementation extends to
vertex-deletion-closed graph families with an efficient perfect-matching
counting oracle. Our QChains inclusion supplies such an oracle, whereas
our hardness results rule out polynomial-time exact oracles on the
three hard classes unless $\mathsf{FP}=\sharpP$. These conclusions
concern exact oracles for unweighted inputs and leave open the
possibility of efficient approximate oracles or other samplers.

The sampling results alone do not determine exact-counting complexity:
monotone and unit interval graphs admit efficient sampling despite exact
hardness, while QChains also admits efficient exact counting.
We next explain how the ordered representations in our hard classes
can encode a counting computation.

\subsection{A circuit-based approach to ordered counting}\label{sec:intro-methods}

The difficulty in proving hardness is that an ordered representation
constrains many edges at once. Choosing endpoints to create a desired
interaction can also create unwanted ones, so assembling independent
matching gadgets does not automatically preserve the graph class.
Yet the order need not give a small boundary for dynamic programming:
arbitrarily many intervals can overlap, and a layer can have arbitrarily
many vertices. We start from the information that a matching carries
across such a boundary and ask whether it can retain enough structure
to encode a counting computation.

Consider consecutive layers of $N=2p$ vertices, with edges only between
successive layers and boundary conditions forcing $p$ edges across each
cut. A state is the subset of vertices in the current layer still
awaiting partners to the right. The vertices in the next state $T$
must remain available for a later layer, so $S$ must be paired
bijectively with the complement of $T$ in the next layer. If $M$ is
the cut matrix, the number of these pairings is the transfer entry
\[
 \mathcal A_M[S,T]=\operatorname{per}M[S,[N]\setminus T].
\]
Multiplying these transfers sums over intermediate states and counts
perfect matchings. Even simple cut matrices act on a state space of
dimension $\binom{N}{p}$, which is exponential in $N$.
We take a quantum-circuit viewpoint on this space by encoding logical
bits in subset states and composing linear operations on them.
Our circuits use exact rational matrices and need not be unitary.
For this approach to work, an operation on one encoded bit must
continue to act correctly in the presence of all the others.

To obtain a common interface for the three geometric reductions, we
restrict the algebraic construction to the upper-triangular matrix
$U_{ij}=\mathbf 1\{i\le j\}$ and single-entry perturbations of it.
Normalizing by the inverse of the transfer without the perturbation
isolates its effect and leads to the problem $\Word$
(Section~\ref{sec:wordeval}). A tempting interpretation is that changing
one entry supplies a local gate. It does not: in the subset-state
space, normalization can move occupied positions to later tracks,
including tracks outside the intended encoding block. Thus a correct
four-track calculation alone gives no guarantee that the same gate
can be composed inside a large instance. The main task is to recover
locality despite this leakage.

A word whose letter indices lie inside a block can move particles
out only through its right boundary. The number of particles at or to
the left of that block can therefore only decrease, so a particle
that leaves cannot return (Section~\ref{sec:flow}). We encode each bit using two occupied tracks
among four and construct a constant-length filter whose isolated
two-particle action is a rank-two projection. The same filter kills
states with fewer than two particles. With $2k$ particles on $k$
four-track blocks, every occupation pattern other than two particles
per block has an underfilled block. Such a pattern cannot persist
through a full pass of the filters without moving particles between
blocks. Each strict change increases a potential whose range is
$O(k^2)$, so there can be only $O(k^2)$ strict changes in the block
occupation pattern, despite the exponential number of states.

We prove that repeated passes therefore stabilize \emph{exactly}:
if $\mathcal F_k$ is one pass of the filters, then
\[
 P_k=\mathcal F_k^{\,2k(k-1)+2},\qquad
 P_k^2=P_k,\qquad \operatorname{rank}P_k=2^k.
\]
The projection is a word of length $O(k^3)$ multiplied by a known
rational normalization, acting on a space of dimension
$\binom{4k}{2k}$. We use row vectors, so the encoded space is the
row space of $P_k$, equivalently the image of $v\mapsto vP_k$.
The projection need not be orthogonal, and this row space need not
be the span of the raw code states. Encoding and decoding through
this space, with $P_k$ inserted between operations, gives the crucial
composition property: an operation confined to a few blocks acts
on the corresponding logical bits and as the identity on all others
(Section~\ref{sec:projection}). A polynomial-length sequence of the
allowed transfers thus supplies the exact local behavior that was
absent from the individual operations.

With composition established, calculations on four or eight tracks
produce one-bit operations and a two-bit interaction
(Section~\ref{sec:gates}). The latter
does not immediately impose the constraint we need. We recover its
diagonal part by conjugation and interpolation, then use powers of
that diagonal operation to recover a constraint forbidding $11$ and
a controlled-sign operation. The interpolation parameters are shared
across all occurrences of each target operation. For $g$ occurrences
of a fixed diagonal operation, we group contributions by how often
each of its constantly many eigenvalues is used. The number of
resulting multiplicity vectors is polynomial in $g$. The reduction
therefore needs only polynomially many evaluations, rather than
expanding the choices independently at every gate position
(Section~\ref{sec:interpolation}). Simultaneous interpolation is a
well-established tool in counting reductions~\cite{HuangLu16}.
Starting with the sum of all Boolean assignments, applying the
forbidden-$11$ constraint for every edge of a source graph, and
summing the surviving assignments counts its independent sets.
Hadamard and controlled-sign operations allow the necessary routing
between nonadjacent bits (Section~\ref{sec:reduction}).

Connections between circuits, matching
signatures, and tensor contractions are well
established~\cite{Valiant02,MarkovShi08,CaiLuXia18}, as are encoded
matchgate constructions and quantum-information approaches to
counting~\cite{JozsaMiyake08,Backens21}. Here, the new step is to
recover a composable encoding by an exact polynomial-length
projection despite the global coupling of the ordered transfers.

It remains to realize this algebraic computation in the promised
graph classes. We first remove inverse transfers by interpolating
a shared repetition parameter, leaving the unweighted paired-layer
problem $\mathrm{PairEval}$. The three geometric constructions then
realize this common interface, but their valid representations
introduce clique edges or complemented cuts. We compensate by
adding paired bipartite probes in the monotone construction
(Section~\ref{sec:monotone}), and clique probes in the unit interval
(Section~\ref{sec:unitinterval}) and chordal permutation
(Section~\ref{sec:permutation}) constructions.
After normalization, the matching count is a polynomial
in a common probe multiplicity. Evaluating it at $-1$ adds weight
$-1$ on prescribed cliques or complete bipartite subgraphs, canceling
the unwanted edges or correcting the complemented cuts, up to a
known uniform sign. All queried multiplicities are nonnegative,
and all queried graphs are simple, unweighted, and explicitly
represented in their target class. Thus the geometric restrictions
are respected throughout the reduction. Cancellation takes place
only in exact postprocessing.

The tractable side uses a different way of controlling boundary
information. The distance-hereditary algorithm merges bags through
pendant and twin operations. Within each bag, all active vertices
have the same neighborhood outside that bag. The state therefore
records only how many active vertices remain unmatched.
Our recurrence handles a merge of bags of sizes $a,b$ in $O(ab)$
integer arithmetic operations. Each pair of original vertices is
charged at most once, when their bags first merge, so these costs
sum to $O(n^2)$
(Section~\ref{sec:qchains}).

\subsection{Further context}

The ordered classes studied here complement the topological
approach to perfect-matching counting. The FKT method gives
polynomial-time algorithms on planar
graphs~\cite{Kasteleyn61,TemperleyFisher61}. Extensions to fixed
surfaces and excluded-minor
classes~\cite{GalluccioLoebl99,Tesler00,Vazirani89,StraubThieraufWagner16}
culminated in the exact-counting dichotomy for minor-closed graph
classes of Thilikos and Wiederrecht~\cite{ThilikosWiederrecht24}.
Every fixed proper minor-closed class has only linearly many
edges~\cite{Mader67}, whereas the ordering constraints considered
here allow dense instances. Density alone does not determine the
complexity of exact counting: chain, cochain, threshold, E-free,
and cograph classes
admit exact algorithms~\cite{OUU10,DM15,DM19}, and fixed
clique-width supports matching-polynomial
algorithms~\cite{MRAG06,CurticapeanMarx16}.

Other work on dense graphs studies different structural restrictions.
El Maalouly and Wang~\cite{ElMaaloulyWang22} prove hardness for graphs
partitionable into two cliques, while matching-polynomial methods give
exact algorithms on cographs~\cite{Jumadildayev25}.
For arbitrary $2r$-vertex graphs, Li~\cite{Li26} gives a
$2^{r-\Omega(\sqrt r)}$-time algorithm, improving the general
exponential-time bound. Our results complement these developments
by settling exact counting under the ordering restrictions in
Theorem~\ref{thm:main}.

Sequential importance sampling provides another approach to
approximate counting. Diaconis and Kolesnik~\cite{DK21} analyze this
method for structured restricted permutations. Alimohammadi,
Diaconis, Roghani, and Saberi~\cite{ADRS23} prove polynomial-time
guarantees for bipartite graphs with $n$ vertices on each side and
minimum degree greater than $(1/2+\lambda)n$, for fixed
$0<\lambda<1/2$.

\subsection{Organization}

Section~\ref{sec:prelims} collects notation and prior results.
Section~\ref{sec:wordeval} introduces $\Word$ and the paired-transfer
interface. Sections~\ref{sec:flow}--\ref{sec:reduction} establish
the common algebraic hardness result.
Sections~\ref{sec:monotone}, \ref{sec:unitinterval}, and
\ref{sec:permutation} give the geometric realizations in the order
of Theorem~\ref{thm:main}.
Section~\ref{sec:qchains} gives the structural inclusion and
counting algorithm.

\section{Preliminaries}\label{sec:prelims}

\subsection{Graphs, representations, and the counting problem}

Write $G=(V,E)$ for a graph with vertex set $V$ and edge set $E$;
each edge is an unordered pair of distinct vertices. All graphs
are finite, simple, undirected, and unweighted unless
explicitly used as auxiliary weighted expressions. Vertices are
labeled. A \emph{matching} is a set of edges with pairwise disjoint
endpoints; it is \emph{perfect} if every vertex is incident with
exactly one of its edges. We write $\PM(G)$ for the number of perfect
matchings of $G$. The graph with no vertices has one perfect matching,
the empty edge set; a graph with an odd number of vertices has none.
A \emph{clique} is a set of pairwise adjacent vertices, and an
\emph{independent set} is a set containing no adjacent pair.
A graph is \emph{bipartite} if its vertex set has a partition
$A\mathbin{\dot\cup}B$ into two independent sets, where
$\dot\cup$ denotes disjoint union.

An \emph{interval representation} assigns a closed real interval
$I_v=[\ell_v,r_v]$ to each vertex, with $uv$ an edge exactly when
$u\ne v$ and $I_u\cap I_v\ne\varnothing$. An \emph{interval graph}
is a graph admitting such a representation. It is a \emph{unit
interval graph} if it admits one with all intervals of length one,
or equivalently of any one common positive length $\Lambda$.
Rescaling the line proves the equivalence. Distinct vertices may
have identical intervals. A \emph{proper interval representation}
has no interval properly containing another; graphs admitting such
a representation are exactly the unit interval graphs
\cite{DM19}.

A \emph{chordal graph} has no induced cycle of length at least four.
An induced cycle has no extra edge among its cycle vertices.
Write $N_G(v)=\{u\in V:uv\in E\}$ for the neighborhood of $v$.
An ordering $(v_1,\ldots,v_n)$ of $V$ is a \emph{perfect elimination
ordering} if $N_G(v_i)\cap\{v_{i+1},\ldots,v_n\}$ is a clique for
every $i$. We use the characterization
\cite[Section~7, pp.~851--852]{FulkersonGross65}
\begin{equation}\label{eq:chordal-peo}
 G\text{ is chordal}
 \quad\Longleftrightarrow\quad
 G\text{ admits a perfect elimination ordering}.
\end{equation}

For comparison with the literature, a \emph{chordal bipartite graph}
is bipartite and has no induced cycle of length at least six.
A \emph{split graph} has a vertex partition into a clique
and an independent set. A \emph{chain graph} is bipartite and its
neighborhoods on each side can be ordered by inclusion. A
\emph{cochain graph} is the complement of a chain graph; graph
complementation exchanges edges and nonedges between distinct
vertices. A \emph{threshold graph} can be constructed from the
empty graph by repeatedly adding either an isolated vertex or
a vertex adjacent to all existing vertices. These comparison classes follow the conventions in~\cite{DM19,OUU10}.

The problem $\UIPM$ takes a unit interval graph $G$, encoded by
its vertices and edges, and outputs the nonnegative integer
$\PM(G)$ in binary. We may equivalently supply a rational
equal-length interval representation. Corneil et al.\ \cite{CorneilEtAl95}
give a linear-time recognition algorithm and, for a nonempty graph
on $n$ vertices, a construction of a unit interval representation
whose endpoints are integer multiples of $1/n$. The coordinates
can be kept polynomially bounded: the union of the intervals in
a connected component of $a$ vertices spans length at most $a$;
translate the components to start at zero and then place them
successively with gaps of one. All endpoints then lie in $[0,2n]$
on the same $1/n$ grid. Their numerators are $O(n^2)$, so their
binary encoding length is polynomial. Thus a representation can be recovered
in polynomial time, and a representation determines all edges by
pairwise endpoint comparisons. Scaling clears the common denominator
and gives equal-length intervals with integer endpoints. These
input models are polynomial-time equivalent. Formally, the counting
function on unrestricted graph encodings can be defined to be zero
when recognition rejects; this gives the same restricted problem.

In the integer model $I_v=[x_v,x_v+\Lambda]$, adjacency is
$|x_u-x_v|\le\Lambda$. For integer $x_v$ and $\Lambda$, this is
equivalent to $|x_u-x_v|<r$ with $r=\Lambda+1$. Both the coordinates
and the common length (or radius) are input data. This equivalence
does not fix the integer radius to a constant.

\subsection{Permutation diagrams and quasi-classes}

A \emph{permutation diagram} consists of one labeled endpoint on
each of two parallel lines for every vertex, with pairwise distinct
endpoints on each line. Vertices are adjacent
when the corresponding straight segments intersect. Equivalently,
two vertices are adjacent exactly when their relative orders in
the two endpoint lists disagree. A graph admitting such a diagram
is a \emph{permutation graph}; it is a \emph{chordal permutation
graph} if it is also chordal. Every endpoint list in our reductions
is explicit, with integer ranks serving as coordinates.

A bipartite graph is \emph{monotone} if its two parts can be ordered
so that each row neighborhood of the biadjacency matrix is an
interval and both interval endpoints are nondecreasing with the
row index; we allow empty intervals $[t+1,t]$.
Dyer and M\"uller~\cite[Lemma~2.13]{DM15} give the characterization
\begin{equation}\label{eq:monotone-bipartite-permutation}
 G\text{ is monotone}
 \quad\Longleftrightarrow\quad
 G\text{ is a bipartite permutation graph}.
\end{equation}
Their convention excludes isolated vertices, but the equivalence
extends to them: isolated rows can be placed first with interval
$[1,0]$ and isolated columns last, while isolated segments can be
placed first in the same order on both lines.

For an explicit monotone ordering, sort each independent part of
a bipartite permutation diagram by its upper endpoints. The order
within each part is the same on both lines. If $a_x,b_x$ count the
opposite-part endpoints preceding $x$ on the upper and lower lines,
respectively, its neighborhood in this order is
\[
 [\min(a_x,b_x)+1,\max(a_x,b_x)].
\]
Both $a_x,b_x$ are nondecreasing along its own part, so both
neighborhood endpoints are nondecreasing.

For a bipartite graph class $\mathcal C$, define its \emph{quasi-class} by
\[
 G\in\operatorname{quasi}(\mathcal C)
 \quad\Longleftrightarrow\quad
 \forall L\subseteq V(G),\quad G[L:V(G)\setminus L]\in\mathcal C,
\]
where $G[L:R]$ retains only edges crossing the partition
$V(G)=L\mathbin{\dot\cup}R$. Thus QMonotone is quasi-Monotone,
whereas QChains is quasi-Chains, with Chains denoting disjoint
unions of chain graphs~\cite{DM19}.

A graph $G$ is \emph{distance-hereditary} if
\begin{equation}\label{eq:distance-hereditary-definition}
 d_{G[W]}(u,v)=d_G(u,v)
 \qquad\bigl(W\subseteq V(G),\ G[W]\text{ connected},\ u,v\in W\bigr),
\end{equation}
where $d_H$ denotes graph distance in $H$. Let $\mathsf{DH}$ denote
this class. The definition applies componentwise to disconnected
graphs. The needed structural characterizations and counting
states are given in Section~\ref{sec:qchains}.

For any class $\mathcal C$ considered here, $\#\mathrm{PM}(\mathcal C)$
means the exact binary output $\PM(G)$ on input $G\in\mathcal C$.
The hardness constructions output both the graph and its indicated
representation. Membership in
$\sharpP$ follows from the usual canonical matching certificates.

\subsection{Counting complexity and the source problem}

We use binary encodings and measure running time in bit operations,
unless an arithmetic-operation bound is stated explicitly.
For nonnegative integer-valued functions, $\mathsf{FP}$ denotes
those computable by deterministic polynomial-time algorithms, and
$\sharpP$ denotes those that count the accepting computation paths
of a nondeterministic polynomial-time machine. Equivalently, a
$\sharpP$ function counts polynomial-length certificates satisfying
a polynomial-time predicate, with each certificate counted once.

A \emph{polynomial-time Turing reduction} from a function $f$ to a
function $g$, written $f\le_T g$, is a deterministic polynomial-time
algorithm for $f$ allowed to query an oracle returning exact values
of $g$. Its query lengths, number of queries, and computation on
oracle answers must all be polynomially bounded. Throughout, $\le_T$
always has this polynomial-time meaning, and all hardness and
completeness claims are with respect to $\le_T$. A counting problem $g$ satisfies
\[
 \begin{aligned}
 g\text{ is }\sharpP\text{-hard}
 &\ \Longleftrightarrow\ (\forall f\in\sharpP)\ f\le_T g,\\
 g\text{ is }\sharpP\text{-complete}
 &\ \Longleftrightarrow\ g\in\sharpP\text{ and }g\text{ is }\sharpP\text{-hard}.
 \end{aligned}
\]
Unlike a \emph{parsimonious reduction}, which maps each
instance to one with the same answer, a Turing reduction can combine
several oracle answers by exact arithmetic. Our reduction uses
rational coefficients and cancellations in this postprocessing;
every queried graph remains unweighted.

The exact problem $\IS$ asks for the total number of independent
sets, of all sizes, in a given simple undirected graph. A \emph{vertex
cover} is a vertex set meeting every edge; taking complements gives
a bijection between vertex covers and independent sets. Provan and
Ball~\cite{PB83} proved that counting vertex covers in bipartite graphs
is $\sharpP$-complete. Viewing these graphs as general graphs therefore
gives $\sharpP$-hardness of $\IS$. Membership in $\sharpP$ follows by
choosing a vertex subset and checking that it contains no edge.
Thus $\IS$ is $\sharpP$-complete.

\subsection{Permanents and matrix notation}

For an integer $n\ge0$, put $[n]=\{0,\ldots,n-1\}$ and
$\cS_{n,q}=\{S\subseteq[n]:|S|=q\}$ for $0\le q\le n$.
The complement $\overline S$ is $[n]\setminus S$ when the ambient
set is $[n]$. We write $\ind{\mathcal P}$ for the indicator of a
statement $\mathcal P$, equal to one when it holds and zero otherwise.
For a matrix $M$, $M[S,T]$ denotes its submatrix on row indices
$S$ and column indices $T$, each in increasing order; when $S,T$
themselves are row and column labels of a matrix, the same bracket
notation denotes the corresponding entry.

For a square matrix $M\in\mathbb Q^{n\times n}$, its
\emph{permanent} is
\[
 \per M=\sum_{\sigma\in\mathfrak S_n}
                   \prod_{i=0}^{n-1}M_{i,\sigma(i)},
\]
where $\mathfrak S_n$ is the set of permutations of $[n]$.
The empty permanent is one. The formula has no permutation signs:
for example, $\per\bigl(\begin{smallmatrix}a&b\\c&d\end{smallmatrix}\bigr)
=ad+bc$. Determinant identities therefore do not automatically
hold for permanents.

For a bipartite graph with ordered parts $A=\{a_0,\ldots,a_{n-1}\}$
and $B=\{b_0,\ldots,b_{n-1}\}$, its \emph{biadjacency matrix}
has entry $M_{ij}=\ind{a_ib_j\in E}$. Each nonzero product in the
permanent selects exactly one neighbor of each vertex, so
$\per M=\PM(G)$. More generally, assigning the rational edge weight
$M_{ij}$ gives the weighted sum
\[
 \sum_{\text{perfect matchings }\mathcal M}
                    \prod_{a_ib_j\in\mathcal M}M_{ij}
 =\per M.
\]
This identity uses the bipartite biadjacency matrix; it is not a
formula for perfect matchings from the adjacency matrix of an
arbitrary graph. For any graph $G$ and rational edge weights $w$, write
\[
 \PM(G;w)=\sum_{\mathcal M\text{ perfect matching of }G}
                         \prod_{e\in\mathcal M}w(e).
\]
This convention also applies to nonbipartite graphs. A weight-zero
edge makes no contribution and can be deleted.

All linear algebra is over $\mathbb Q$. Matrices act on row vectors:
rows index inputs, columns index outputs, and products execute from
left to right. Thus an entry of $M_1\cdots M_a$ is the sum over
intermediate indices of the products of the corresponding entries;
we call each such sequence a \emph{path} in the product. Write
$I_d$ for the $d\times d$ identity, or $I$ when its dimension is
clear, $M^{\mathsf T}$ for transpose, and
$M^{-\mathsf T}=(M^{-1})^{\mathsf T}$ for an invertible $M$.
The notation $\diag(d_1,\ldots,d_s)$ denotes a diagonal matrix.
The \emph{tensor product} (Kronecker product) is indexed by ordered
pairs and satisfies
\[
 (A\otimes B)[(i,j),(i',j')]=A[i,i']B[j,j'];
\]
$A^{\otimes k}$ denotes its $k$-fold repetition. Tensor factors
are placed in increasing coordinate order. A \emph{projection}
means an idempotent matrix $P$, satisfying $P^2=P$, with no
orthogonality requirement. Its rank equals its trace because its
only eigenvalues are zero and one and it is diagonalizable over
$\mathbb Q$. Here rank is the dimension of the image and trace
is the sum of diagonal entries.

\subsection{Exact interpolation}

We repeatedly use elementary polynomial interpolation. A polynomial
$f\in\mathbb Q[z]$ of degree at most $d$ is uniquely determined by
its values at distinct rational nodes $a_0,\ldots,a_d$, through
\begin{equation}\label{eq:lagrange}
 f(z)=\sum_{j=0}^d f(a_j)
              \prod_{\substack{0\le i\le d\\i\ne j}}
                       \frac{z-a_i}{a_j-a_i}.
\end{equation}
This follows since the right-hand side has the prescribed values
and the difference has more roots than its degree. If $d$ and the
binary lengths of the nodes and values are polynomially bounded,
the coefficients and any evaluation at a point of polynomial bit
length can be computed in polynomial time by exact rational
arithmetic. Indeed, each numerator and denominator is a product
of polynomially many numbers of polynomial bit length, and only
polynomially many such terms are added. A constant number of
successive interpolations retains polynomial complexity; a separate
parameter for each of polynomially many positions need not do so.

For $j\ge0$ we use the polynomial convention
$\binom tj=t(t-1)\cdots(t-j+1)/j!$, with $\binom t0=1$;
it applies also to negative or nonintegral $t$. A matrix $K$ is
\emph{nilpotent} if $K^{d+1}=0$ for some $d\ge0$, and $I+K$ is
then called \emph{unipotent}. The binomial polynomial
$\sum_{j=0}^d\binom tjK^j$ equals $(I+K)^t$ at nonnegative
integers and equals $(I+K)^{-1}=\sum_{j=0}^d(-K)^j$ at $t=-1$.
These facts justify the inverse-transfer interpolation below.

\section{WordEval and layered perfect matchings}\label{sec:wordeval}\label{sec:oracle}

We introduce the common algebraic interface before any geometric
realization. The states are subsets of a layer, not explicit rows
of a matrix stored by the reduction. A transfer records how the
vertices still awaiting partners can be matched across the next cut.

\subsection{Subset states and permanental transfers}

Fix $n\ge1$ and $0\le q\le n$. A \emph{track} is one of the
ordered positions in $[n]$, and a \emph{state} is a subset
$S\in\cS_{n,q}$. Identify $S$ with its indicator bit string, listed
in increasing track order. The \emph{$q$-particle space} is the
$\binom nq$-dimensional rational row-vector space with one coordinate
for each state. In linear expressions a state string denotes its
standard coordinate row; for example, $-1010$ denotes the negative
of that row, not a signed binary integer. A track is \emph{occupied}
when its indicator is one; a \emph{particle} means only membership
in the subset.

For an $n\times n$ matrix $M$, define its
\emph{$q$th permanental compound}, a $\binom nq\times\binom nq$
matrix, by
\[
 F_q(M)[S,T]=\per M[S,T],\qquad S,T\in\cS_{n,q}.
\]
Thus an entry counts weighted bijections from the selected rows
to the selected columns, using the permanent defined in
Section~\ref{sec:prelims}. 
At half filling, $N=2p$, let $J[S,T]=\ind{T=[N]\setminus S}$
on $\cS_{N,p}$. Thus $J^2=I$. For a cut matrix $M$, define
\[
 \cA_M=F_p(M)J,\qquad
 \cA_M[X,Y]=\per M[X,[N]\setminus Y].
\]

\subsection{Reading a perfect matching one layer at a time}

For $\ell\ge1$, consider $\ell+1$ successive layers, each initially
an independent set of $N$ vertices labeled by tracks. The
\emph{cut} between two consecutive layers is the set of edges
joining them. Its biadjacency matrix has rows in the earlier layer
and columns in the later layer. Give the cuts matrices
$M_1,\ldots,M_\ell$, with no other edges. In the first layer retain only $S$;
in the last retain only $[N]\setminus T$. Every perfect matching
then has exactly $p$ edges at each cut. The state at an internal
layer is the set of its $p$ vertices not matched from the preceding
layer. It follows that the perfect matching count is
\begin{equation}\label{eq:transfer}
 (\cA_{M_1}\cdots\cA_{M_\ell})[S,T].
\end{equation}
This identity also holds with rational edge weights, interpreting
the weighted perfect matching count as the sum of products of edge weights.

The complement has a concrete role: $Y$ describes the vertices
reserved for the next cut, so the current cut must cover its
complement. The first boundary forces $p$ edges at the first cut;
each full internal layer then forces $p$ edges at the next cut by
induction. Multiplying the transfers sums over all intermediate
states, and every perfect matching has a unique such state sequence.
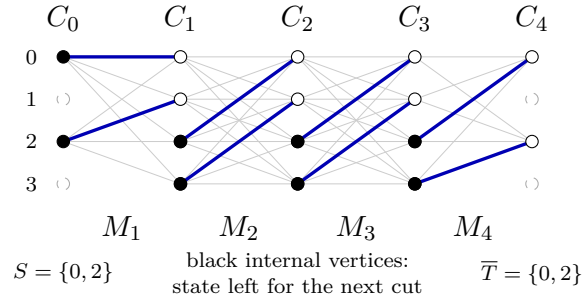
\begin{figure}[htbp]
\centering
\begin{tikzpicture}[x=1.55cm,y=.56cm,
 vertex/.style={circle,draw,fill=white,inner sep=1.65pt},
 retained/.style={circle,draw,fill=black,inner sep=1.65pt},
 selected/.style={line width=1.25pt,draw=blue!70!black},
 unused/.style={gray!40,line width=.35pt}]
% N=4, p=2, four cuts. Every displayed solid blue edge belongs
% to one perfect matching. Dashed circles are deleted boundary vertices.
\foreach \j in {0,...,4}{
  \foreach \i in {0,...,3}{\coordinate (v\j\i) at (\j,-\i);}
  \node[above=6pt] at (v\j0) {$C_{\j}$};
}
\foreach \j/\k in {0/1,1/2,2/3,3/4}{
  \foreach \i in {0,...,3}{\foreach \r in {0,...,3}{
    \ifnum\j=0
      \ifodd\i\else\draw[unused] (v\j\i)--(v\k\r);\fi
    \else
      \ifnum\k=4
        \ifodd\r\else\draw[unused] (v\j\i)--(v\k\r);\fi
      \else\draw[unused] (v\j\i)--(v\k\r);\fi
    \fi
  }}
}
\foreach \a/\b in {00/10,02/11,12/20,13/21,22/30,23/31,32/40,33/42}
  \draw[selected] (v\a)--(v\b);
\foreach \j in {1,2,3}{
 \foreach \i in {0,...,3}{\node[vertex] at (v\j\i) {};}
 \foreach \i in {2,3}{\node[retained] at (v\j\i) {};}
}
\foreach \j in {0,4}{
 \foreach \i in {0,2}{
   \ifnum\j=0\node[retained] at (v\j\i) {};
   \else\node[vertex] at (v\j\i) {};\fi
 }
 \foreach \i in {1,3}{\node[circle,draw=gray!65,dashed,inner sep=1.65pt] at (v\j\i) {};}
}
\foreach \i in {0,...,3}{\node[left=6pt] at (v0\i) {\scriptsize $\i$};}
\foreach \j/\lab in {.5/1,1.5/2,2.5/3,3.5/4}
 \node at (\j,-4.05) {$M_{\lab}$};
\node[align=center,font=\scriptsize] at (0,-5.1) {$S=\{0,2\}$};
\node[align=center,font=\scriptsize] at (2,-5.1) {black internal vertices:\\state left for the next cut};
\node[align=center,font=\scriptsize] at (4,-5.1) {$\overline T=\{0,2\}$};
\end{tikzpicture}
\caption{Layer transfer at $N=4$, $p=2$, with four complete cuts for
illustration. Dashed boundary vertices are deleted; thin gray lines
are available edges and thick blue lines form one perfect matching.
Exactly two edges cross each cut. At an internal layer, black vertices
are not matched from the preceding layer and must be matched to the
next one. If that next layer has outgoing state $Y$, the current cut
must use its complementary set $\overline Y$, explaining the factor
$J$ in $\cA_M[X,Y]=\per M[X,\overline Y]$. The last retained set is
$\overline T$ because its vertices are all matched from the left;
these last-layer vertices are therefore white.}
\label{fig:layered-transfer}
\end{figure}

We use this exponentially large state space implicitly to encode a
counting computation. The relevant operations are specified by
small perturbations of a particularly simple cut matrix.

\subsection{Four normalized operations and WordEval}

Set
\[
 U_{ij}=\ind{i\le j},\qquad L=U^{\mathsf T}.
\]
For $S=\{s_1<\cdots<s_q\}$ and $T=\{t_1<\cdots<t_q\}$,
\[
 F_q(U)[S,T]\ne0\ \Longrightarrow\
 s_j\le t_j\quad(1\le j\le q),\qquad F_q(U)[S,S]=1.
\]
Thus $F_q(U)$ is unit upper triangular
(upper triangular with every diagonal entry equal to one)
when states are ordered by increasing element sum, with arbitrary
tie-breaking, and is invertible.

Let $e_{ab}$ be the elementary matrix having a one in position $(a,b)$
and zeros elsewhere.
For $0\le i\le n-2$, put
\[
 V_i=U+e_{i+1,i},\qquad Z_i=U-e_{i,i},
\]
and define operations on the $q$-particle space by
\begin{equation}\label{eq:RD}
 \R_i^{(n,q)}=F_q(V_i)F_q(U)^{-1},\qquad
 \D_i^{(n,q)}=F_q(Z_i)F_q(U)^{-1}.
\end{equation}
Let $J_q$ have rows indexed by $\cS_{n,q}$ and
columns by $\cS_{n,n-q}$, with
$J_q[S,T]=\ind{T=[n]\setminus S}$. Define
\begin{equation}\label{eq:BE}
 \begin{split}
 \B_i^{(n,q)}&=J_q(\R_i^{(n,n-q)})^{\mathsf T}J_{n-q},\\
 \E_i^{(n,q)}&=J_q(\D_i^{(n,n-q)})^{\mathsf T}J_{n-q}.
 \end{split}
\end{equation}
In particular, complementing a local $q$-particle state uses the
$(n-q)$-particle space. Superscripts are omitted when determined
by context.

\begin{definition}
An instance of $\Word$ specifies an even number of tracks $N=2p\ge2$,
two states $S,T\in\cS_{N,p}$, and an explicit word $w$ in the
operations $\R_i,\D_i,\B_i,\E_i$, with $0\le i\le N-2$.
A \emph{word} is a finite ordered sequence of these letters,
evaluated on the specified $(N,p)$ space; its length is its number
of letters. For $w=X_1\cdots X_g$, write
\[
 W_w=X_1\cdots X_g,\qquad
 \Word(N,w,S,T)=W_w[S,T],\qquad W_{\varnothing}=I.
\]
The boundary states are given as $N$-bit vectors,
so the input length is at least $N$.
\end{definition}

The matrices may have exponentially many rows and columns, but are
specified implicitly by $N,p$ and their letter indices. A known
rational scalar multiplying a word is handled separately by exact
arithmetic and is not an additional oracle input letter.

\paragraph{Why normalize the perturbations?}
The upper triangular matrix $U$ permits motion only to later
tracks. The factors $F_q(U)^{-1}$ remove this background transfer,
leaving the effect of adding one subdiagonal edge or removing one
diagonal edge. Already at $n=2$, $q=1$, one obtains
\[
 U=\begin{pmatrix}1&1\\0&1\end{pmatrix},\qquad
 \R_0=\begin{pmatrix}1&0\\1&0\end{pmatrix},\qquad
 \D_0=\begin{pmatrix}0&1\\0&1\end{pmatrix}.
\]
These two operations send either one-particle state to a prescribed
track. On larger spaces the same normalization can affect later
tracks, so the operations cannot simply be assumed local. The flow
and projection arguments in Sections~\ref{sec:flow}--\ref{sec:projection}
are what make their Boolean-encoded actions composable.

\begin{lemma}\label{lem:word-bitlength}
For every WordEval instance with $N$ tracks and word $w$,
\[
 W_w[S,T]\in\mathbb Z,\qquad
 \operatorname{bitlength}(W_w[S,T])=\poly(N+|w|).
\]
\end{lemma}
\begin{proof}
Let $N=2p$, $F=F_p(U)$, and $d=\binom Np\le2^N$.
Every matrix $M\in\{U,V_i,Z_i\}$ has entries in $\{0,1\}$.
By the definition of the permanent, an entry
$F_p(M)[S,T]=\per M[S,T]$ is a sum over the $p!$ bijections
from $S$ to $T$, each contributing zero or one. Hence
\[
 F_p(M)[S,T]=\sum_{\pi:S\overset{\sim}{\to}T}
       \prod_{u\in S}M_{u,\pi(u)}\in\{0,1,\ldots,p!\}.
\]

We next justify a polynomial bound on the nilpotence index of
$F-I$, rather than using only its exponentially large dimension.
For a subset state $S$, put $\omega(S)=\sum_{u\in S}u$.
A nonzero permanent term in $F[S,T]$ is a bijection
$\pi:S\to T$ satisfying $\pi(u)\ge u$ for every $u\in S$.
Hence $\omega(T)\ge\omega(S)$; equality forces $\pi(u)=u$
for every $u$, and therefore $S=T$. This also proves $F[S,S]=1$:
the identity is the only contributing bijection on $S$.
Consequently, for $B=F-I$,
\[
 B[S,T]\ne0\ \Longrightarrow\ \omega(T)\ge\omega(S)+1.
\]
On $\cS_{2p,p}$ its minimum and maximum are
\[
 \frac{p(p-1)}2\quad\hbox{and}\quad\frac{p(3p-1)}2,
\]
whose difference is $p^2$. A nonzero term in an entry of $B^a$
would require $a$ successive strict increases, so
$B^{p^2+1}=0$. The finite geometric-series identity now gives
\[
 F^{-1}=\sum_{a=0}^{p^2}(-B)^a,
\]
which has integer entries.

Each entry of $B$ has absolute value at most $p!$.
For $a\ge1$, matrix multiplication expresses an entry of $B^a$
as a sum of $d^{a-1}$ products of $a$ such entries, giving the
bound $d^{a-1}(p!)^a$. Including the identity term, every entry
of $F^{-1}$ therefore has absolute value at most
\[
 C_0=(p^2+1)(dp!)^{p^2}.
\]
By~\eqref{eq:RD}, every entry of $\R_i$ or $\D_i$ is bounded
in absolute value by $C_1=dp!C_0$. At half filling, the
complements and transposes in~\eqref{eq:BE} only permute these
entries, so the same integer bound holds for $\B_i$ and $\E_i$.

Finally, an entry of a product of $g\ge1$ letters is an integer
of absolute value at most $d^{g-1}C_1^g$. Since
$\log_2 d\le N$ and $\log_2(p!)\le p\log_2(p+1)$,
\[
 \log_2 C_1
 \le \log_2(p^2+1)+(p^2+1)\bigl(N+p\log_2(p+1)\bigr).
\]
The binary length of the answer, including a sign bit, is thus
polynomial in $N+g$. For the empty word the answer is an entry
of the identity and has constant bit length.
\end{proof}

\subsection{A common paired-transfer interface}

For $N=2p$, let $\mathcal P_N$ consist of the pairs
\[
 (U,L),\quad(V_i,L),\quad(Z_i,L),\quad
 (U,V_i^{\mathsf T}),\quad(U,Z_i^{\mathsf T})
 \qquad(0\le i\le N-2).
\]
The \emph{paired layered evaluation problem}, denoted by
$\mathrm{PairEval}$, takes boundary states
$S,T\in\cS_{N,p}$ and an explicit sequence of $h\ge1$ pairs
$(M_j,H_j)\in\mathcal P_N$, and asks for
\begin{equation}\label{eq:paired-value}
 Z=(\cA_{M_1}\cA_{H_1}\cdots\cA_{M_h}\cA_{H_h})[S,T].
\end{equation}
This is an ordinary perfect-matching count on the independent-layer
graph above, with $2h+1$ layers and $n_0=4ph$ retained vertices.
Its two color classes each have $2ph$ vertices.

\begin{lemma}\label{lem:paired-transfer-interface}
One has $\Word\le_T\mathrm{PairEval}$.
For a word of length $g\ge1$, the reduction uses $gp^2+1$ instances,
each with $h=g(t+1)$ pairs for an integer $0\le t\le gp^2$.
\end{lemma}
\begin{proof}
For the rest of this proof write $F=F_p(U)$ and $J=J_p$.
Set
\[
 W=\cA_U\cA_L=FJF^{\mathsf T}J.
\]
Both $F$ and $JF^{\mathsf T}J$ are unit upper triangular in the
element-sum order. Every strict transition increases that sum,
whose range on $\cS_{2p,p}$ has width $p^2$. Consequently
\begin{equation}\label{eq:unipotent}
 W=I+K_0,\qquad K_0^{p^2+1}=0,
 \qquad W^t=\sum_{j=0}^{p^2}\binom tjK_0^j.
\end{equation}
The last expression is a matrix polynomial in $t$, agrees with
ordinary powers for all nonnegative integers $t$, and has value
$W^{-1}$ at $t=-1$.

Replace each letter of a word by the indicated sequence of two-cut
pairs, all using the same nonnegative integer $t$:
\begin{equation}\label{eq:letter-pairs}
\begin{array}{c|l}
\text{letter}&\text{sequence of pairs}\\\hline
\R_i&(V_i,L)(U,L)^t\\
\D_i&(Z_i,L)(U,L)^t\\
\B_i&(U,L)^t(U,V_i^{\mathsf T})\\
\E_i&(U,L)^t(U,Z_i^{\mathsf T}).
\end{array}
\end{equation}
Here $(M,H)$ means two consecutive cuts with transfer
$\cA_M\cA_H$; juxtaposition concatenates pairs, and $(U,L)^t$
means $t$ repetitions, with no pair when $t=0$.
Their transfer products at $t=-1$ are precisely the four letters.
For example,
\[
 \cA_{V_i}\cA_LW^{-1}=F_p(V_i)F^{-1}=\R_i,
\]
whereas
\[
 W^{-1}\cA_U\cA_{V_i^{\mathsf T}}
 =\cA_L^{-1}\cA_{V_i^{\mathsf T}}
 =JF^{-\mathsf T}F_p(V_i)^{\mathsf T}J=\B_i.
\]
The identities for $\D_i,\E_i$ follow with $Z_i$ in place
of $V_i$. If $P_w(t)$ is the resulting matrix entry for a word
of length $g$, then
\[
 \deg P_w\le gp^2,\qquad P_w(-1)=W_w[S,T].
\]
Thus $P_w(0),\ldots,P_w(gp^2)$ determine the desired word entry
by interpolation at $-1$.

The empty word is handled directly. For nonempty words the displayed
bounds give polynomially many pairs in every sample. Values of
\eqref{eq:paired-value} have polynomial bit length, being perfect
matching counts on $4ph$ vertices, and exact interpolation preserves
polynomial bit complexity by Section~\ref{sec:prelims}.
\end{proof}

The negative evaluation point is used only in exact postprocessing:
all sampled values of $t$ are nonnegative and all paired layer graphs
are unweighted.

For each target graph class, it therefore suffices to recover
\eqref{eq:paired-value} using polynomially many queries of polynomial
size in $N+h$ and polynomial-time exact postprocessing. Indeed,
the lemma bounds both the number of paired instances and their
lengths by polynomials in $N+g$; composing such an evaluation
procedure with these instances remains a polynomial-time Turing
reduction. This separates the common algebraic argument from the
three geometric realizations.

The following five sections prove that evaluating
these implicitly represented word matrices is $\sharpP$-hard.

\section{One-way flow and restriction to a block}\label{sec:flow}

The transfer operations are defined globally and need not be local
in the usual tensor-product sense. We establish a weaker property
that will allow an encoding projection to restore locality.

A \emph{transition} of a matrix is an input-output pair with
nonzero entry. In this section a cut $j$ separates tracks
$0,\ldots,j-1$ from tracks $j,\ldots,n-1$; it is different from
a cut between layers in Section~\ref{sec:oracle}. For
$S\in\cS_{n,q}$ define its \emph{prefix count} by
\[
 c_j(S)=|S\cap\{0,\ldots,j-1\}|\qquad(0\le j\le n).
\]
A transition $S\to T$ is \emph{rightward} if
$c_j(T)\le c_j(S)$ for every $j$.

\begin{lemma}\label{lem:flow}
Let $0\le i\le n-2$, $S,T\in\cS_{n,q}$, and
$\mathsf O\in\{\R,\D,\B,\E\}$.
If $\mathsf O_i[S,T]\ne0$, then $S\cap[i]=T\cap[i]$ and
\[
 c_j(T)-c_j(S)\le
 \begin{cases}
  0,&\mathsf O\in\{\D,\E\},\\
  \ind{j=i+1},&\mathsf O\in\{\R,\B\},
 \end{cases}
 \qquad 0\le j\le n.
\]
\end{lemma}

\begin{proof}
Work in the square-free commutative algebra
\[
 \mathscr A_n=\mathbb Q[z_0,\ldots,z_{n-1}]/(z_0^2,\ldots,z_{n-1}^2),
 \qquad r_j=\sum_{v=j}^{n-1}z_v,\quad r_n=0.
\]
The quotient notation means that the variables commute and satisfy
$z_j^2=0$: terms containing a repeated $z_j$ vanish. Consequently
the square-free monomials $z_S=\prod_{j\in S}z_j$, for
$S\subseteq[n]$, form a basis. The \emph{homogeneous degree-$q$
part} is the span of those with $|S|=q$.
The products $\prod_{j\in S}r_j$, for $S\in\cS_{n,q}$, form a
basis of this degree-$q$ part: their coefficient matrix
in the monomial basis is $F_q(U)$, which is invertible.
We identify state $S$ with the row-form product
$\rho_S=\prod_{j\in S}r_j$ in this latter basis. More generally,
row $S$ of $F_q(M)$ is the coefficient row, in the $z$-monomial
basis, of the product of the row forms of $M$ indexed by $S$.
Right multiplication by $F_q(U)^{-1}$ converts these coefficients
to the $\rho$-basis. The matrix $Z_i$ changes only row $i$ of $U$,
from $r_i$ to $r_{i+1}$, whereas $V_i$ changes only row $i+1$,
from $r_{i+1}$ to $r_i$. Thus~\eqref{eq:RD} means that $\D_i$
replaces $r_i$ by $r_{i+1}$ when $i\in S$, and $\R_i$ replaces
$r_{i+1}$ by $r_i$ when $i+1\in S$; otherwise the respective
operation fixes $\rho_S$. Repeated factors
can be \emph{straightened}, meaning rewritten as a linear combination
of the distinct-index row-form products, using
\begin{equation}\label{eq:straightening}
 r_j^a=a r_jr_{j+1}^{a-1}-(a-1)r_{j+1}^a\qquad(a\ge2).
\end{equation}
Indeed, $r_j=r_{j+1}+z_j$ and $z_j^2=0$ give this identity.
Straightening moves factors only to larger indices; it terminates
at $r_n=0$ if a repeated last factor occurs. The replacement made
by $\D_i$ is itself rightward, and every affected factor has
index at least $i$. Thus $\D_i$ satisfies the stated prefix-count
bounds and fixes every position less than $i$.

Let $S_i$ be the permutation matrix exchanging positions $i,i+1$
in subset states. Put $P_S=\prod_{j\in S\setminus\{i,i+1\}}r_j$.
The following table lists the input and output factors before
straightening; every displayed factor is multiplied by $P_S$.
\[
\begin{array}{c|c|c|c|c}
 (\ind{i\in S},\ind{i+1\in S})
 &\text{input}&\R_i&\D_i&I+S_i\\\hline
 (0,0)&1&1&1&2\\
 (1,0)&r_i&r_i&r_{i+1}&r_i+r_{i+1}\\
 (0,1)&r_{i+1}&r_i&r_{i+1}&r_i+r_{i+1}\\
 (1,1)&r_ir_{i+1}&r_i^2&r_{i+1}^2&2r_ir_{i+1}
\end{array}
\]
The first three rows equate the sum of the $\R_i,\D_i$ outputs
with the last output. In the last row,
\[
 (r_i-r_{i+1})^2=z_i^2=0
 \quad\Longrightarrow\quad r_i^2+r_{i+1}^2=2r_ir_{i+1}.
\]
Multiplying these equalities by $P_S$ and taking their coordinates
in the $\rho$-basis proves, for every input row $S$,
\begin{equation}\label{eq:RplusD}
 \R_i+\D_i=I+S_i.
\end{equation}
The swap changes only positions $i,i+1$ and increases a prefix
count by at most one, only at cut $i+1$. Since
$\R_i=I+S_i-\D_i$, its support satisfies the claimed bounds.

Finally, complementing both states and reversing a transition
preserves its prefix-count inequalities, since
$c_j(\overline S)=j-c_j(S)$. Apply this observation to
\eqref{eq:BE} to obtain the assertions for $\B_i,\E_i$.
\end{proof}

Let $\mathcal B=\{a,\ldots,a+m-1\}\subseteq[n]$ be a consecutive
block with $m\ge1$, and let $w$ use only indices $i$ with
$i,i+1\in\mathcal B$. Write
$W_w$ for its product matrix on the $(n,q)$ space. Define
\[
 k_{\mathcal B}(S)=|S\cap\mathcal B|,\qquad
 s_{\mathcal B}(S)=\{u-a:u\in S\cap\mathcal B\}.
\]
The \emph{local word on $\mathcal B$} replaces each letter index
$i$ by $i-a$. Let $W_{\mathrm{loc}}^{(m,k)}$ be its product
on the isolated $(m,k)$ space; in particular, its $\B,\E$
operations use~\eqref{eq:BE} with complement taken inside
$\mathcal B$ and hence in the local $(m-k)$-particle space.

\begin{lemma}[Restriction without boundary flow]\label{lem:restriction}
For $S,T\in\cS_{n,q}$,
\[
 W_w[S,T]\ne0\quad\Longrightarrow\quad
 k_{\mathcal B}(T)\le k_{\mathcal B}(S).
\]
If $k_{\mathcal B}(S)=k_{\mathcal B}(T)=k$, then
\begin{equation}\label{eq:block-restriction}
 W_w[S,T]=\ind{S\setminus\mathcal B=T\setminus\mathcal B}\,
 W_{\mathrm{loc}}^{(m,k)}[s_{\mathcal B}(S),s_{\mathcal B}(T)].
\end{equation}
\end{lemma}

\begin{proof}
By Lemma~\ref{lem:flow}, every elementary operation fixes the
positions to the left of $\mathcal B$ and cannot increase the
prefix count at its right boundary. Since
\[
 k_{\mathcal B}(S)=c_{a+m}(S)-c_a(S),
\]
the block count cannot increase at any step. A nonzero product
entry has at least one contributing sequence of nonzero letter
entries, so the same inequality holds for $W_w$.

For the second assertion, partition the tracks into the part to the left of
$\mathcal B$, the block $\mathcal B$, and the part to its right. The matrices
$U,V_i,Z_i$ are block upper triangular for this partition.
Their permanental compounds are block upper triangular when
indexed by the numbers of particles in the three parts.
In a diagonal block with fixed particle counts, any nonzero
permanent term must match within each part: a rightward crossing
cannot be balanced by a leftward crossing. Therefore this
diagonal block is the tensor product of the three local compounds.
The diagonal blocks of $F_q(U)^{-1}$ are likewise tensor products
of the local inverses. In~\eqref{eq:RD}, the outside factors
cancel to identities, leaving the local operation on $\mathcal B$.
Complement and transpose give the same assertion for
\eqref{eq:BE}, using the complementary local particle number.

In a product, a decrease at the right boundary of $\mathcal B$
cannot be recovered. Thus a transition preserving its block count uses
only the just-described diagonal restrictions at every step.
Their product is the local word on $\mathcal B$ and the identity
on each outside state. Distinct exterior occupations therefore
give a zero entry; identical exterior occupations give exactly
the local coefficient in~\eqref{eq:block-restriction}.
\end{proof}

\section{A polynomial-length encoding projection}\label{sec:projection}

\subsection{A four-track filter}

Define the following words on four consecutive tracks, with indices
relative to that block:
\begin{equation}\label{eq:filter-words}
 \begin{split}
 \mathsf Z&=\R_1\D_2\R_2\B_1\R_2,\\
 \mathsf J&=\R_1\R_0\E_2\D_2\R_2,\\
 \mathsf T&=\mathsf Z\mathsf J,\qquad
 \Theta=\frac1{64}\mathsf T^2.
 \end{split}
\end{equation}
The filter $\Theta$ is a word of length twenty times a known scalar.
Use $\Theta^{[q]}$ to denote its restriction to the isolated four-track
$q$-particle space, that is, the displayed word evaluated with
$(n,q)=(4,q)$ and complements in the $4-q$ sector. We call this
word a \emph{filter} because it annihilates the sectors stated below;
a \emph{sector} means the span of states with the specified particle
number. Outputs in the next table are coordinate rows in the
bit-string notation defined in Section~\ref{sec:oracle}.

\begin{lemma}\label{lem:filter}
One has $\Theta^{[0]}=\Theta^{[1]}=0$. On the two-particle space the
complete action is
\begin{equation}\label{eq:filter-table}
\begin{array}{c|rrrrrr}
\text{input}&1100&1010&1001&0110&0101&0011\\\hline
\text{output}&1010&1010&-1010&0110&\tfrac12\,0110&\tfrac12\,0110.
\end{array}
\end{equation}
In particular,
\[
 (\Theta^{[2]})^2=\Theta^{[2]},\qquad
 \rank\Theta^{[2]}=2,\qquad
 s\Theta^{[2]}=s\quad(s\in\{1010,0110\}).
\]
\end{lemma}

\begin{proof}
For a fixed $n\times n$ matrix $M$, the entries needed here can be
computed by expansion along the least row:
\begin{equation}\label{eq:per-rec}
 \begin{split}
 f_M(\varnothing,\varnothing)&=1,\\
 f_M(S,T)&=\sum_{j\in T}M_{ij}
                 f_M(S\setminus\{i\},T\setminus\{j\}),
                 \qquad i=\min S\quad(S\ne\varnothing).
 \end{split}
\end{equation}
Thus $f_M(S,T)=F_q(M)[S,T]$ for $|S|=|T|=q$.
For every sector $0\le q\le n$, set $T_q=F_q(U)-I$.
As in Lemma~\ref{lem:word-bitlength}, a nonzero entry of $T_q$
strictly increases the subset weight $\sum_{i\in S}i$.
The range of this weight has width $q(n-q)$, so
$T_q^{q(n-q)+1}=0$ and
\begin{equation}\label{eq:finite-inverse}
 F_q(U)^{-1}=\sum_{a=0}^{q(n-q)}(-T_q)^a.
\end{equation}
Using these formulas with $n=4$ in
\eqref{eq:RD}--\eqref{eq:BE} and~\eqref{eq:filter-words},
we obtain $\Theta^{[0]}=\Theta^{[1]}=0$ and
\eqref{eq:filter-table} by direct calculation.
The displayed outputs all lie in the span of the two fixed rows,
proving idempotence and rank.
\end{proof}

We encode logical zero by $1010$ and logical one by $0110$.
We will not need the values of the filter in particle numbers
three and four. Its vanishing below two particles, together with
a fixed total particle number, is enough to isolate the encoding.

\subsection{Stabilization of the global filter}

For $k\ge1$, use $N=4k$ tracks and $p=2k$ particles, split into
blocks $b=0,\ldots,k-1$, with block $b$ consisting of tracks
$\{4b,\ldots,4b+3\}$. Let $\Theta_b$ be the word in
\eqref{eq:filter-words} translated to block $b$, now acting in the
global $p$-particle space: replace every subscript $i$ by $4b+i$
and evaluate the resulting letters at $(N,p)=(4k,2k)$. Define
\[
 \cF_k=\Theta_0\Theta_1\cdots\Theta_{k-1},\qquad
 K_k=2k(k-1)+2.
\]
Although $\Theta_b$ can move particles beyond block $b$, the following
power is a genuine projection.

\begin{proposition}\label{prop:projection}
The matrix
\begin{equation}\label{eq:global-projection}
 P_k=\cF_k^{K_k}
\end{equation}
satisfies
\[
 P_k^2=P_k,\qquad \rank P_k=2^k.
\]
It is a word of length $20kK_k=O(k^3)$ multiplied by the known
scalar $64^{-kK_k}$.
\end{proposition}

\begin{proof}
Partition the state space by the vector of block particle counts
$q=(q_0,\ldots,q_{k-1})$. Every count satisfies $0\le q_b\le4$
and $\sum_bq_b=2k$. Define the order
\[
 q\preceq q'\ \Longleftrightarrow\
 \sum_{b<j}q'_b\le\sum_{b<j}q_b\quad(1\le j\le k).
\]
Write $\cF_k[q,q']$ for the submatrix from states with count
vector $q$ to states with count vector $q'$.
Lemma~\ref{lem:flow} gives
\[
 \cF_k[q,q']\ne0\ \Longrightarrow\ q\preceq q',
\]
so $\cF_k$ is block upper triangular in any compatible order.
By Lemma~\ref{lem:restriction}, its diagonal block is
\begin{equation}\label{eq:diagonal-filter}
 \cF_k[q,q]=\bigotimes_{b=0}^{k-1}\Theta^{[q_b]}.
\end{equation}
If $q\ne q_*=(2,\ldots,2)$, then at least one $q_b$ is zero or
one; otherwise the fixed sum forces $q=q_*$. The corresponding
factor in~\eqref{eq:diagonal-filter} is zero. Thus the only
nonzero diagonal block is
\[
 \cF_k[q_*,q_*]=(\Theta^{[2]})^{\otimes k},
\]
an idempotent of rank $2^k$.

For $\phi(q)=\sum_b bq_b$ and $q\preceq q'$ with $q\ne q'$,
\[
 \phi(q')-\phi(q)
 =\sum_{j=1}^{k-1}\left(\sum_{b<j}q_b-\sum_{b<j}q'_b\right)
 \ge1.
\]
Thus every strict block transition increases $\phi$.
Its range has width at most $2k(k-1)$,
so a path contributing to a matrix power has at most this many
strict block transitions. It cannot stay in any block except
$q_*$, since all other diagonal blocks are zero. Once it leaves
$q_*$ it cannot return. For a power greater than $2k(k-1)$,
every nonzero path must consequently spend at least one step
on the diagonal block $q_*$. All positive powers of this block
are equal. Fixing the strict transitions fixes the one place
where the remaining steps can be inserted, so increasing the
power introduces neither a new contribution nor a multiplicity
factor. Hence
\[
 \cF_k^m=\cF_k^{K_k}=P_k\qquad(m\ge K_k).
\]

It follows that $P_k^2=\cF_k^{2K_k}=P_k$. Its only nonzero
diagonal block has trace $2^k$, so $\rank P_k=\tr P_k=2^k$. The word length
and scalar follow directly from the definition.
\end{proof}

Figure~\ref{fig:encoding-projection} illustrates both the local
four-track encoding and the stabilization of the global filter.
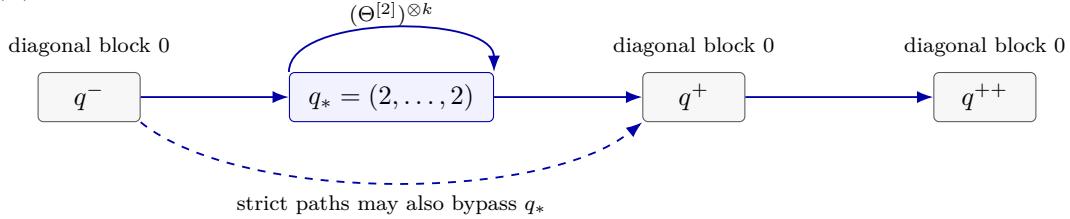
\begin{figure}[htbp]
\centering
\begin{tikzpicture}[x=1cm,y=1cm,>=Latex,
  every node/.style={font=\small},
  input/.style={inner sep=2pt},
  count/.style={draw=black!65,rounded corners=2pt,fill=black!3,
    minimum width=1.35cm,minimum height=.65cm},
  flow/.style={->,line width=.7pt,draw=blue!65!black},
  occupied/.style={circle,draw=blue!65!black,fill=blue!65!black,
    inner sep=0pt,minimum size=5pt},
  vacant/.style={circle,draw=blue!65!black,fill=white,
    inner sep=0pt,minimum size=5pt}]
\path[use as bounding box] (0,0) rectangle (15.1,8.7);
\node[anchor=west,font=\small\bfseries] at (0,8.45)
  {(a) A local filter: six two-particle states collapse to two fixed rows};
\node[input] (a0) at (.8,7.8) {$1100$};
\node[input] (b0) at (2.2,7.8) {$1010$};
\node[input] (c0) at (3.6,7.8) {$1001$};
\node[input] (a1) at (5.7,7.8) {$0110$};
\node[input] (b1) at (7.1,7.8) {$0101$};
\node[input] (c1) at (8.5,7.8) {$0011$};
\draw[rounded corners=3pt,draw=blue!65!black,fill=blue!3]
  (.55,6.05) rectangle (3.85,6.85);
\draw[rounded corners=3pt,draw=blue!65!black,fill=blue!3]
  (5.45,6.05) rectangle (8.75,6.85);
\coordinate (zero) at (2.2,6.85);
\coordinate (one) at (7.1,6.85);
\draw[flow] (a0.south) -- (zero)
  node[pos=.45,left,font=\scriptsize] {$1$};
\draw[flow] (b0.south) -- (zero)
  node[pos=.43,right,font=\scriptsize] {$1$};
\draw[flow] (c0.south) -- (zero)
  node[pos=.45,right,font=\scriptsize] {$-1$};
\draw[flow] (a1.south) -- (one)
  node[pos=.45,left,font=\scriptsize] {$1$};
\draw[flow] (b1.south) -- (one)
  node[pos=.43,right,font=\scriptsize] {$\tfrac12$};
\draw[flow] (c1.south) -- (one)
  node[pos=.45,right,font=\scriptsize] {$\tfrac12$};
\foreach \i in {0,1,2,3} {
  \node[font=\scriptsize] at ({1.3+.6*\i},6.25) {$\i$};
  \node[font=\scriptsize] at ({6.2+.6*\i},6.25) {$\i$};
}
\foreach \i in {0,2} \node[occupied] at ({1.3+.6*\i},6.57) {};
\foreach \i in {1,3} \node[vacant] at ({1.3+.6*\i},6.57) {};
\foreach \i in {1,2} \node[occupied] at ({6.2+.6*\i},6.57) {};
\foreach \i in {0,3} \node[vacant] at ({6.2+.6*\i},6.57) {};
\node at (2.2,5.7) {logical $0$: $1010$};
\node at (7.1,5.7) {logical $1$: $0110$};
\node[align=left,anchor=west] at (10,6.8) {
  $(\Theta^{[2]})^2=\Theta^{[2]}$\\[5pt]
  $\rank \Theta^{[2]}=2$\\[5pt]
  $\Theta^{[0]}=\Theta^{[1]}=0$};
\draw[black!20] (0,5.2) -- (15.1,5.2);
\node[anchor=west,font=\small\bfseries] at (0,4.9)
  {(b) The global filter: rightward flow and one surviving diagonal block};
\node[count] (qm) at (1.4,3.45) {$q^{-}$};
\node[count,draw=blue!65!black,fill=blue!5,minimum width=2.7cm]
  (qs) at (5.4,3.45) {$q_*=(2,\ldots,2)$};
\node[count] (qp) at (9.4,3.45) {$q^{+}$};
\node[count] (qpp) at (13.25,3.45) {$q^{++}$};
\draw[flow] (qm) -- (qs);
\draw[flow] (qs) -- (qp);
\draw[flow] (qp) -- (qpp);
\draw[flow] (qs.north west) .. controls (4.05,4.55) and (6.75,4.55)
  .. (qs.north east);
\node[font=\scriptsize] at (5.4,4.55) {$(\Theta^{[2]})^{\otimes k}$};
\node[font=\scriptsize] at (1.4,4.12) {diagonal block $0$};
\node[font=\scriptsize] at (9.4,4.12) {diagonal block $0$};
\node[font=\scriptsize] at (13.25,4.12) {diagonal block $0$};
\draw[flow,dashed] (qm.south east) .. controls (3.4,2.05) and (7.4,2.05)
  .. (qp.south west);
\node[font=\scriptsize,fill=white,inner sep=2pt] at (5.4,2.02)
  {strict paths may also bypass $q_*$};
\node[align=center] at (7.55,1.5) {
  Every strict step increases $\phi(q)=\sum_b bq_b$;
  at most $2k(k-1)$ strict steps are possible.};
\node[align=center] at (7.55,.87) {
  $\cF_k^m=P_k\quad(m\ge K_k),\qquad
  P_k^2=P_k,\qquad \rank P_k=2^k.$};
\node[font=\scriptsize,align=center] at (7.55,.3) {
  Off-diagonal contributions can survive in $P_k$:
  its row space need not be the span of the raw code states.};
\end{tikzpicture}
\caption{Local encoding and global stabilization.
In (a), filled dots are occupied tracks; arrow labels are coefficients
in the action of $\Theta^{[2]}$ on coordinate rows.
In (b), arrows illustrate possible directions of nonzero block
transitions, not a claim that every displayed block is nonzero.
Only the block $q_*$ can contribute a diagonal step. A sufficiently
long contributing path therefore passes through $q_*$, where repeating
the diagonal step changes neither its value nor its multiplicity.
$P_k$ may retain contributions from paths entering or leaving $q_*$.}
\label{fig:encoding-projection}
\end{figure}

\subsection{The encoded action and its locality}

Let $d=2^k$ and $D_k=\binom{4k}{2k}$. For a logical bit string
$x\in\{0,1\}^k$, let $\operatorname{code}(x)$ replace each zero
by $1010$ and each one by $0110$. These are the \emph{raw code
states}. Define the $D_k\times d$ matrix
$C_k[S,x]=\ind{S=\operatorname{code}(x)}$, ordering logical strings
lexicographically. Its columns are the corresponding standard
coordinate columns. Define
\begin{equation}\label{eq:encoding}
 \cE_k=C_k^{\mathsf T}P_k,\qquad \cL_k=P_kC_k.
\end{equation}
The matrix $\cE_k$ has dimensions $d\times D_k$, and $\cL_k$
has dimensions $D_k\times d$.
For a matrix $M$ on the global $(4k,2k)$ space, define its
\emph{encoded operation} by $\enc M=\cE_kM\cL_k$; the hat
notation suppresses $k$.

\begin{lemma}\label{lem:encoding}
One has
\[
 C_k^{\mathsf T}P_kC_k=I_d,\quad
 \cE_k\cL_k=I_d,\quad P_k=\cL_k\cE_k.
\]
For $a\ge1$ and any matrices $M_1,\ldots,M_a$ on this space,
\begin{equation}\label{eq:composition}
 \enc M_1\cdots\enc M_a
 =C_k^{\mathsf T}P_kM_1P_kM_2P_k\cdots P_kM_aP_kC_k.
\end{equation}
\end{lemma}

\begin{proof}
A path starting and ending at block counts $q_*$ cannot leave
that block. Its restriction in $P_k$ is
$(\Theta^{[2]})^{\otimes k}$, which fixes every raw code row.
This proves the first identity. Idempotence gives
$\cE_k\cL_k=I_d$. The $d$ independent columns of $\cL_k$ lie
in the rank-$d$ column space of $P_k$, so they span it.
Write $P_k=\cL_kV$. Since $C_k^{\mathsf T}\cL_k=I_d$,
multiplication on the left by $C_k^{\mathsf T}$ gives
$V=C_k^{\mathsf T}P_k=\cE_k$.
The composition identity follows by multiplication.
\end{proof}

The encoded row space, spanned by the rows of $\cE_k$, is the
row space of $P_k$ and need not be the span of the raw code states.
A logical row $v\in\mathbb Q^{1\times d}$ is encoded and recovered by
\[
 v\longmapsto v\cE_k,\qquad (v\cE_k)\cL_k=v.
\]

For $1\le r\le k$ and $0\le b\le k-r$, suppose every letter index
$i$ of a word $M$ satisfies $4b\le i\le4(b+r)-2$.
Thus both tracks $i,i+1$ lie in the $r$ blocks beginning at $b$.
Let $M^{\mathrm{loc}}$ be the word obtained
by subtracting $4b$ from every letter index and evaluating on the
isolated $4r$-track, $2r$-particle space.

\begin{lemma}[Projected locality]\label{lem:locality}
With $b,r,M$ as above,
\[
 \enc M=I_{2^b}\otimes
 \bigl(C_r^{\mathsf T}P_rM^{\mathrm{loc}}P_rC_r\bigr)
 \otimes I_{2^{k-b-r}}.
\]
\end{lemma}

\begin{proof}
Consider a path from a raw code state to another through
$P_kMP_k$. At every block boundary other than those internal
to the active region, each of $P_k,M,P_k$ can only decrease
the prefix count. If $S_0,\ldots,S_\ell$ is a contributing state
path and $j$ is such a boundary, then
\[
 c_j(S_0)\ge c_j(S_1)\ge\cdots\ge c_j(S_\ell)=c_j(S_0).
\]
Every inequality is therefore an equality: the path preserves
each of these counts throughout.

Restrict to these preserved counts. By
Lemma~\ref{lem:restriction}, each pass of $\cF_k$ is the tensor
product of $\cF_r$ on the active region and $\Theta^{[2]}$ on each
outside block. Its $K_k$th power restricts to $P_r$ on the active
region, since $K_k\ge K_r$ and the powers have stabilized, and
to $\Theta^{[2]}$ on each outside block. The word $M$ is the identity
outside and equals $M^{\mathrm{loc}}$ inside. Finally, extracting
the raw code entries of each outside factor gives
$C_1^{\mathsf T}\Theta^{[2]}C_1=I_2$.
The claimed tensor product follows.
\end{proof}

With $b,r,M$ as in Lemma~\ref{lem:locality}, put
\[
 g=C_r^{\mathsf T}P_rM^{\mathrm{loc}}P_rC_r,
 \qquad g_{\mathrm{ext}}=I_{2^b}\otimes g\otimes I_{2^{k-b-r}}.
\]
Together with $P_k=\cL_k\cE_k$ from Lemma~\ref{lem:encoding},
projected locality gives
\[
 \cE_kMP_k=g_{\mathrm{ext}}\cE_k,
 \qquad P_kMP_k=\cL_kg_{\mathrm{ext}}\cE_k.
\]
Thus the global word followed by $P_k$ sends an encoded row
$v\cE_k$ to $(vg_{\mathrm{ext}})\cE_k$.
Equation~\eqref{eq:composition} composes these logical operations
by retaining the intervening projections.

\section{Constant-size operations on the encoding}\label{sec:gates}

The logical space on $k$ bits is $\mathbb Q^{2^k}$, with coordinate
rows indexed lexicographically by $\{0,1\}^k$. A \emph{one-bit
operation} at bit $j\in\{0,\ldots,k-1\}$ has the form
\[
 I_{2^j}\otimes M\otimes I_{2^{k-j-1}},\qquad M\in\mathbb Q^{2\times2}.
\]
A \emph{two-bit operation} is defined analogously with a
$4\times4$ matrix acting on two specified bits.
A \emph{logical circuit} is an explicitly listed product of these
operations, also called \emph{gates}. Its size is the number of gates,
and its evaluation is a specified matrix entry of the product.
These are rational linear circuits: gates may be noninvertible
and need not map a bit string to a single bit string. The one-bit
basis order is $0,1$, and the two-bit order is $00,01,10,11$.

Lemma~\ref{lem:locality} reduces the following computations to
four or eight tracks, independently of the number of logical bits.
Let
\[
 X=\begin{pmatrix}0&1\\1&0\end{pmatrix},\qquad
 H=\begin{pmatrix}1&1\\1&-1\end{pmatrix},\qquad
 Q=\diag(2,1).
\]
The matrix $H$ is unnormalized; all factors of two will be retained
explicitly.

\begin{lemma}\label{lem:gates}
On one logical bit the following encoded actions hold:
\begin{equation}\label{eq:one-bit}
\begin{split}
 \enc{\R_0}&=\begin{pmatrix}1&0\\1&0\end{pmatrix},\qquad
 \enc{\E_0}=\begin{pmatrix}1&1\\0&0\end{pmatrix},\\
 \enc{\B_1}&=Q,\qquad
 \enc{\D_2}=A:=\diag(-1,1/2),\qquad
 \enc{\mathsf T}=8X.
\end{split}
\end{equation}
For the word $\mathsf M=\D_1\B_1\B_0\B_1\R_1\R_2$,
\begin{equation}\label{eq:mix}
 \enc{\mathsf M}=-2\begin{pmatrix}1&-2\\-2&-4\end{pmatrix},
 \qquad A\enc{\mathsf M}A=-2H.
\end{equation}
On two adjacent bits, the operation at their shared track
boundary satisfies
\begin{equation}\label{eq:G}
 \enc{\R_3}=G=
 \begin{pmatrix}
 -2&2&-7/2&-6\\
 0&-3&0&-2\\
 0&0&2&-4\\
 0&0&0&1
 \end{pmatrix},
\end{equation}
in the order $00,01,10,11$.
\end{lemma}

\begin{proof}
For one bit, the encoded matrices are
\[
 C_1^{\mathsf T}P_1wP_1C_1,
 \qquad P_1=\Theta^{[2]},\qquad
 w\in\{\R_0,\E_0,\B_1,\D_2,\mathsf T,\mathsf M\}.
\]
For two bits, they are computed on eight tracks using
\[
 C_2^{\mathsf T}P_2\R_3P_2C_2,
 \qquad P_2=(\Theta_0\Theta_1)^6,
\]
where $C_2$ selects $10101010,10100110,01101010,01100110$
in that order. Using~\eqref{eq:per-rec}--\eqref{eq:finite-inverse},
we obtain~\eqref{eq:one-bit}--\eqref{eq:G} by direct calculation.
Lemma~\ref{lem:locality} places the same matrices at every block
or adjacent pair of blocks.
\end{proof}

In particular,
\[
 X=\tfrac18\enc{\mathsf T},\qquad
 H=-\tfrac12\enc{\D_2}\enc{\mathsf M}\enc{\D_2},\qquad
 Q=\enc{\B_1}.
\]
Equation~\eqref{eq:composition} implements these products by
a constant number of raw letters and occurrences of $P_k$,
with the displayed nonzero scalar normalizations.
Below, the \emph{available one-bit operations} mean
the fixed matrices in~\eqref{eq:one-bit}, together with $X,H,Q$;
their products are given as explicit circuits. In particular,
a power $Q^u$ means a product of
$u$ \emph{projected} copies. For a general word $M$, one must not
replace $(\enc M)^u$ by $\enc{M^u}$ without justification.

\section{Shared-parameter interpolation of constraints}\label{sec:interpolation}

\subsection{Removing the off-diagonal entries}

For a fixed finite set $\mathcal G$ of rational $4\times4$ matrices,
let $\mathrm{CircuitEval}(\mathcal G)$ be the following problem.
The input specifies a number of logical bits, an explicitly listed
circuit, and two boundary bit strings. Its one-bit gates are the
available operations of Section~\ref{sec:gates}; its two-bit gates
belong to $\mathcal G$ and act on adjacent bits. The output is the
specified entry of the circuit product as an exact rational number.
Thus the input records the
gate locations and order without listing the full circuit matrix.
Figure~\ref{fig:circuit-eval} illustrates the fixed boundaries and
the internal states summed over in this evaluation.

\begin{figure}[htbp]
\centering
\begin{tikzpicture}[x=1cm,y=1cm,>=Latex,
  every node/.style={font=\small},
  gate/.style={draw=blue!65!black,fill=blue!4,line width=.7pt,
    minimum width=.8cm,minimum height=.65cm,inner sep=3pt},
  two/.style={gate,fill=blue!9,minimum height=1.65cm},
  state/.style={font=\small,text=blue!65!black}]
\node[align=center] at (1.05,3.9)
  {fixed input\\$s=(s_0,s_1,s_2)$};
\node[align=center] at (11.55,3.9)
  {fixed output\\$t=(t_0,t_1,t_2)$};
\draw[->,line width=.7pt,draw=blue!65!black]
  (3,4.05) -- (9.6,4.05)
  node[midway,above=3pt] {gate order / matrix multiplication};
\foreach \j/\yy in {0/2.6,1/1.6,2/.6} {
  \node[anchor=east] at (1.5,\yy) {$s_{\j}$};
  \draw[draw=black!65,line width=.7pt] (1.6,\yy) -- (11,\yy);
  \node[anchor=west] at (11.1,\yy) {$t_{\j}$};
}
\foreach \xx/\cutstate in {4.1/a,6.3/b,8.5/c} {
  \draw[dashed,draw=blue!40] (\xx,.2) -- (\xx,3.1);
  \node[state] at (\xx,3.35) {$\cutstate$};
}
\node[gate] at (3,2.6) {$H$};
\node[two] at (5.2,2.1) {$G_*$};
\node[gate] at (7.4,1.6) {$Q$};
\node[two] at (9.6,1.1) {$G_*$};
\node[state] at (6.3,-.2)
  {sum over internal states $a,b,c\in\{0,1\}^3$};
\end{tikzpicture}
\[
 C=(H\otimes I_2\otimes I_2)(G_*\otimes I_2)
   (I_2\otimes Q\otimes I_2)(I_2\otimes G_*).
\]
\caption{A three-bit instance of $\mathrm{CircuitEval}(\mathcal G)$,
with $G_*\in\mathcal G$. Each box acts on the wires it meets;
the factors of $C$ follow the boxes from left to right.
The specified boundary strings $s,t$ are fixed, while the states
at the dashed cuts are summed over. The value to be computed is
the rational number $C[s,t]=e_sCe_t^{\mathsf T}$.}
\label{fig:circuit-eval}
\end{figure}

\begin{remark*}[Tensor networks and Holant]
Each CircuitEval instance is a Boolean tensor-network contraction
with fixed boundary indices; see~\cite{MarkovShi08} for the circuit
viewpoint. A gate $M$ on $d\in\{1,2\}$ bits supplies the tensor
$f_M(x,y)=M[x,y]$ with $2d$ Boolean indices, ordered as inputs
then outputs. Wire segments carry Boolean variables, and unary
pins $\delta_b(z)=\ind{z=b}$ impose the boundary bits. Summing
the product of these local entries over the wire variables gives
exactly the specified circuit entry. This is a rational-weighted
Boolean Holant evaluation~\cite{CaiLuXia18}, restricted to the
fixed gate signatures together with the boundary pins, on planar
circuit layouts with a temporal order and adjacent two-bit gates.
\end{remark*}

Set
\[
 \Delta=\diag(-2,-3,2,1).
\]
\begin{lemma}\label{lem:diagonal}
$\mathrm{CircuitEval}(\{\Delta\})\le_T\Word$.
\end{lemma}

\begin{proof}
For a nonnegative integer $u$, set $q=2^u$ and
$K(q)=\diag(1,q)$. Since
\[
 K(q)=XQ^uX,\qquad K(q)^{-1}=q^{-1}Q^u,
\]
both matrices are implementable by words of length $O(u)$
in the logical operations, with known scalar factors; for $u=0$,
use the empty word for both identity matrices.
Replace every occurrence of $\Delta$ by
\[
 G(q)=(K(q)^{-1}\otimes K(q)^{-1})\,G\,
                (K(q)\otimes K(q)).
\]
For two-bit strings $x,y$, the relation $x\le y$ means that each
bit of $x$ is at most the corresponding bit of $y$; let
$\wt(x)$ be the \emph{Hamming weight}, the number of one bits in $x$.
Equation~\eqref{eq:G} gives
\[
 G[x,y]\ne0\ \Longrightarrow\ x\le y
 \ \Longrightarrow\ 0\le\wt(y)-\wt(x)\le2.
\]
Consequently the entrywise polynomial extension is
\begin{equation}\label{eq:Gq}
 G(q)[x,y]=
 \begin{cases}
 G[x,y]q^{\wt(y)-\wt(x)},&x\le y,\\
 0,&x\not\le y.
 \end{cases}
\end{equation}
Each entry has degree at most two. Moreover,
$x\le y$ and $\wt(x)=\wt(y)$ imply $x=y$, so $G(0)=\Delta$.

If the circuit has $h$ such occurrences, sharing $q$ among all
of them gives a chosen-entry polynomial $\Phi$ satisfying
\[
 \deg\Phi\le2h,\qquad
 \Phi(0)=\text{the entry of the original $\Delta$-circuit}.
\]
Evaluate it at $q=2^u$ for $u=0,\ldots,2h$, and
interpolate at zero. Other gates, including $H$, can occur
anywhere in the circuit: they do not affect the degree argument.

A circuit with no gates is evaluated directly. Otherwise, on
$k\ge1$ logical bits, implement each sampled gate
using Section~\ref{sec:gates}, retaining its internal projections
(including those in $H$ and the repeated copies of $Q$).
Then use~\eqref{eq:composition} to insert $P_k$ between successive
gates. Each $P_k$ expands into $O(k^3)$ letters and the scalar
$64^{-kK_k}$. Since $u\le2h$, this produces one polynomial-length
word on $N=4k$ tracks, multiplied by a known rational scalar.
For logical boundaries $s,t$, its selected physical entry has
boundaries $\operatorname{code}(s),\operatorname{code}(t)$;
in particular, $0^k$ corresponds to $(1010)^k$.

The sample nodes $2^u$ have polynomial bit length. The scalar is a product
of polynomially many fixed gate normalizations, projection
normalizations, and powers of $q^{-1}$, so it has polynomial bit
length. Lemma~\ref{lem:word-bitlength} gives the same bound for
each $\Word$ answer. Multiplying by the scalar recovers the sampled
entry, and the exact interpolation bounds following
\eqref{eq:lagrange} recover $\Phi(0)$ in polynomial time.
\end{proof}

\subsection{Two diagonal constraints from powers of \texorpdfstring{$\Delta$}{Delta}}

We require the matrices
\[
 \mathsf N=\diag(1,1,1,0),\qquad
 \mathrm{CZ}=\diag(1,1,1,-1).
\]
Here $\mathsf N$ is the binary constraint that forbids $11$, and
$\mathrm{CZ}$ is the controlled-sign operation: it multiplies
the row $11$ by $-1$ and fixes the other three coordinate rows.
The first is a diagonal constraint on assignments, not a Boolean
gate with a separate output wire.

\begin{lemma}\label{lem:spectral}
$\mathrm{CircuitEval}(\{\mathsf N,\mathrm{CZ}\})
\le_T\mathrm{CircuitEval}(\{\Delta\})$.
\end{lemma}

\begin{proof}
If neither target type occurs, the circuit is already an instance
of $\mathrm{CircuitEval}(\{\Delta\})$. Otherwise fix a target type
occurring $g\ge1$ times, and replace all of these occurrences by
\[
 \Delta^{2r}=\diag(4^r,9^r,4^r,1)
\]
with a common nonnegative integer $r$. In the expansion of any
circuit entry over intermediate Boolean states, these positions
contribute $(4^a9^b)^r$: $a$ counts their states $00$ or $10$,
$b$ counts $01$, and $g-a-b$ counts $11$.
The numbers
\[
 \alpha_{ab}=4^a9^b,\qquad a,b\in\mathbb Z_{\ge0},\quad a+b\le g,
\]
are distinct by unique prime factorization. There are
$J_g=(g+1)(g+2)/2$ of them.

For the target $\mathsf N$, choose the unique polynomial $L_{g,N}$
of degree less than $J_g$ satisfying
\[
 L_{g,N}(\alpha_{ab})=\ind{a+b=g}.
\]
For the target $\mathrm{CZ}$, instead require
\[
 L_{g,\mathrm{CZ}}(\alpha_{ab})=(-1)^{g-a-b}.
\]
Write the chosen polynomial as
$L_{g,*}(z)=\sum_{r=0}^{J_g-1}\lambda_rz^r$, where $*$ denotes
the chosen target type, $N$ or $\mathrm{CZ}$.
Group the intermediate-state paths by $(a,b)$, and let $c_{ab}$
be the sum of their weights from all other gate positions.
Writing $V(r)$ for the entry with the $g$ positions replaced
by $\Delta^{2r}$ gives
\[
 V(r)=\sum_{a+b\le g}c_{ab}\alpha_{ab}^r,
 \qquad
 \sum_{r=0}^{J_g-1}\lambda_rV(r)
 =\sum_{a+b\le g}c_{ab}L_{g,*}(\alpha_{ab}).
\]
The prescribed values of $L_{g,*}$ are exactly the products of
the $g$ target entries on these paths. The last sum therefore
restores all occurrences of that type simultaneously.

When both types occur, with multiplicities $g_N,g_{\mathrm{CZ}}$, use
two parameters $r_N,r_{\mathrm{CZ}}$. Let $V(r,s)$ be the sampled
entry and $\lambda_r^N,\lambda_s^{\mathrm{CZ}}$ the coefficients
of the corresponding interpolation polynomials. The target entry is
\[
 \sum_{r=0}^{J_{g_N}-1}\sum_{s=0}^{J_{g_{\mathrm{CZ}}}-1}
 \lambda_r^N\lambda_s^{\mathrm{CZ}}V(r,s).
\]
The same pathwise argument proves this identity on a grid of size
$J_{g_N}J_{g_{\mathrm{CZ}}}$, even when the two types and the
one-bit operations are interleaved. Each power uses only
$2r_*$ copies of $\Delta$, and $r_*<J_{g_*}$ is polynomial.
Each grid point is an instance of
$\mathrm{CircuitEval}(\{\Delta\})$, so its value is supplied by that oracle.

The nodes $\alpha_{ab}$ have $O(g)$ bits, so the interpolation
coefficients have polynomial bit length and are computable in
polynomial time by the bounds following~\eqref{eq:lagrange}.
The oracle answers also have polynomial bit length: choose fixed
integers $d_0,b_0\ge1$ such that every entry of an available gate
or of $\Delta$ is $a/d_0$ with $|a|\le b_0$.
A circuit of length $\ell$ has at most $4^\ell$ state paths from
a fixed input, since each gate acts on at most two bits.
Thus each specified entry is $\nu/d_0^\ell$ with
$|\nu|\le(4b_0)^\ell$. The polynomial grid and the displayed
linear combination therefore give polynomial-time exact
postprocessing.
\end{proof}

\section{The algebraic counting reduction}\label{sec:reduction}

\subsection{Counting independent sets}

\begin{proposition}\label{prop:source}
$\IS\le_T\mathrm{CircuitEval}(\{\mathsf N,\mathrm{CZ}\})$.
\end{proposition}

\begin{proof}
Let the source graph $\Gamma=(V,E)$ have $k$ vertices
and $m$ edges, with $V=[k]$.
For $m=0$, return $2^k$ directly. Otherwise assign one
logical bit to each vertex, initially zero. At each bit apply
$\enc{\E_0}$ from~\eqref{eq:one-bit}. With $e_x$ denoting the
coordinate row of assignment $x\in\{0,1\}^k$,
\[
 e_{0^k}(\enc{\E_0})^{\otimes k}
 =\sum_{x\in\{0,1\}^k}e_x.
\]
For each source edge, apply $\mathsf N$ to its two endpoint bits.
The resulting coefficient of an assignment, indexed by source
vertices, is
\[
 a(x)=\prod_{\{u,v\}\in E}(1-x_ux_v)
 =\ind{\{v:x_v=1\}\text{ is independent in }\Gamma}.
\]

To apply such an operation to nonadjacent bits, move them
together by adjacent exchanges, updating the correspondence
between source vertices and bit positions. The required
exchange is the permutation operation
$\mathrm{SWAP}:(a,b)\mapsto(b,a)$. Define
$\mathrm{CNOT}_{1\to2}:(a,b)\mapsto(a,a\oplus b)$ and
$\mathrm{CNOT}_{2\to1}:(a,b)\mapsto(a\oplus b,b)$, where
$\oplus$ is addition modulo two. Both are permutation matrices
under the row-vector convention. They are available because
\[
 (I\otimes H)\mathrm{CZ}(I\otimes H)=2\,\mathrm{CNOT}_{1\to2},
\]
and
\[
 (H\otimes I)\mathrm{CZ}(H\otimes I)=2\,\mathrm{CNOT}_{2\to1},
 \qquad
 \mathrm{SWAP}=\mathrm{CNOT}_{1\to2}\mathrm{CNOT}_{2\to1}
                   \mathrm{CNOT}_{1\to2}.
\]
Thus an exchange is implemented by
three $\mathrm{CZ}$ operations and six $H$ operations, with scalar
normalization $1/8$. Each edge requires at most $k-1$ adjacent
exchanges. No final restoration of the variable order is needed.

Finally apply $\enc{\R_0}$ at every bit. Since
$e_x(\enc{\R_0})^{\otimes k}=e_{0^k}$ for every $x$,
this sums the coefficients independently of the current variable
order. After the explicit scalar normalizations, the all-zero to
all-zero entry is
\[
 \sum_{x\in\{0,1\}^k}a(x)=\#\{I\subseteq V:I\text{ is independent in }\Gamma\}.
\]
The circuit has $O(km)$ gates, and the product of its
normalization factors has polynomial bit length.
\end{proof}

Figure~\ref{fig:independent-set-reduction} gives a three-vertex
instance, including the vertex order after an adjacent exchange.

\begin{figure}[htbp]
\centering
\begin{tikzpicture}[x=1cm,y=1cm,>=Latex,
  every node/.style={font=\small},
  gate/.style={draw=blue!65!black,fill=blue!4,line width=.7pt,
    minimum width=.72cm,minimum height=.62cm,inner sep=3pt},
  two/.style={gate,fill=blue!9,minimum height=1.62cm},
  vertex/.style={circle,draw=blue!65!black,fill=blue!4,
    minimum size=.46cm,inner sep=0pt},
  mapping/.style={font=\small,text=blue!65!black}]
\node at (1.05,3.4) {source graph $\Gamma$};
\node at (8.7,3.4) {logical circuit $C$: fixed boundary $000\to000$};
\node[vertex] (v0) at (.4,1.4) {$0$};
\node[vertex] (v1) at (1.7,2.1) {$1$};
\node[vertex] (v2) at (1.7,.7) {$2$};
\draw[line width=.8pt] (v0) -- (v1) (v0) -- (v2);
\draw[->,line width=.8pt,draw=blue!65!black] (2.35,1.4) -- (3.2,1.4);
\foreach \yy in {2.4,1.4,.4} {
  \node at (3.65,\yy) {$0$};
  \draw[draw=black!65,line width=.7pt] (3.9,\yy) -- (13.25,\yy);
  \node at (13.55,\yy) {$0$};
  \node[gate] at (4.8,\yy) {$\enc{\E_0}$};
  \node[gate] at (12.35,\yy) {$\enc{\R_0}$};
}
\node[two] at (6.55,1.9) {$\mathsf N$};
\node[mapping] at (6.55,2.96) {edge $01$};
\draw[draw=blue!65!black,line width=.8pt] (8.4,.4) -- (8.4,1.4);
\foreach \yy in {.4,1.4} {
  \draw[draw=blue!65!black,line width=.9pt]
    (8.27,\yy-.13) -- (8.53,\yy+.13)
    (8.27,\yy+.13) -- (8.53,\yy-.13);
}
\node[fill=white,inner sep=2pt] at (8.4,.9) {$\mathrm{SWAP}$};
\node[two] at (10.55,1.9) {$\mathsf N$};
\node[mapping] at (10.55,2.96) {edge $02$};
\foreach \v/\yy in {0/2.4,1/1.4,2/.4} {
  \node[mapping,above=2pt] at (5.55,\yy) {$v_{\v}$};
}
\foreach \v/\yy in {0/2.4,2/1.4,1/.4} {
  \node[mapping,above=2pt] at (9.4,\yy) {$v_{\v}$};
}
\node[mapping] at (5.55,-.18) {order $(0,1,2)$};
\node[mapping] at (9.4,-.18) {order $(0,2,1)$};
\node[align=center] at (6.95,-1.0)
  {$\mathcal I(\Gamma)=\{\varnothing,\{0\},\{1\},\{2\},\{1,2\}\}$
   \qquad $C[000,000]=|\mathcal I(\Gamma)|=5$};
\end{tikzpicture}
\caption{The reduction of Proposition~\ref{prop:source} for
$E(\Gamma)=\{01,02\}$. The blue labels identify the source vertex
carried by each wire. The first constraint forbids $x_0x_1=1$;
after exchanging the lower two wires, the second forbids $x_0x_2=1$.
The final one-bit gates sum the five surviving assignments.
Here SWAP is the normalized permutation operation: its realization
by six $H$ and three $\mathrm{CZ}$ gates has the scalar factor $1/8$.
Without this factor, the corresponding circuit entry is $40$.}
\label{fig:independent-set-reduction}
\end{figure}
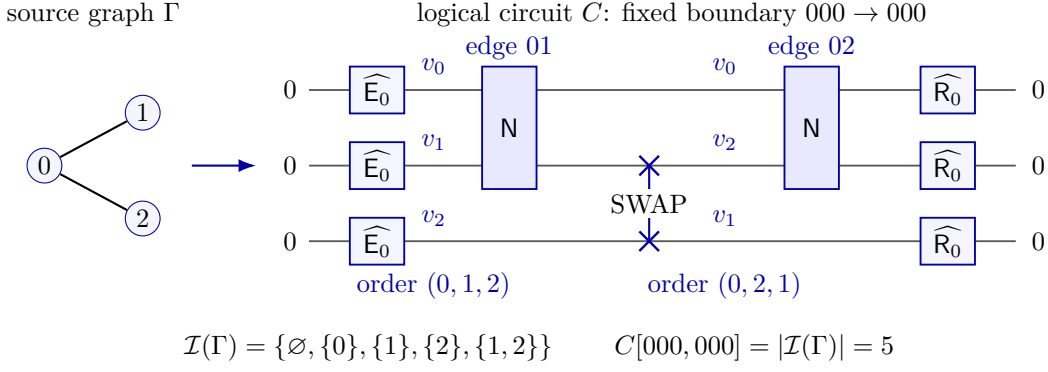

\subsection{Algebraic hardness}

\begin{theorem}\label{thm:word-hard}
$\IS\le_T\Word\le_T\mathrm{PairEval}$.
\end{theorem}

\begin{proof}
The first reduction is the composition of Proposition~\ref{prop:source},
Lemma~\ref{lem:spectral}, and Lemma~\ref{lem:diagonal}.
The second is Lemma~\ref{lem:paired-transfer-interface}.
\end{proof}

Since $\IS$ is $\sharpP$-complete~\cite{PB83}, both $\Word$
and $\mathrm{PairEval}$ are $\sharpP$-hard. Each of the next three
sections proves $\mathrm{PairEval}\le_T\#\mathrm{PM}(\mathcal C)$
for one target graph class $\mathcal C$, using a shared probe parameter.

\section{Monotone graphs}\label{sec:monotone}

Our monotone construction keeps every layer independent. Paired bipartite
probes enforce the required cancellation within a bipartite graph. We produce
both a permutation diagram and a monotone ordering in the sense of
Section~\ref{sec:prelims}.

\begin{theorem}\label{thm:monotone}
$\#\mathrm{PM}(\mathrm{Monotone})$ is $\sharpP$-complete.
\end{theorem}

We construct an oracle for $\mathrm{PairEval}$. Fix $N=2p$, $h\ge1$,
$S,T\in\cS_{N,p}$, and pairs $(M_r,H_r)$ of the allowed types, and write
\begin{equation}\label{eq:monotone-target}
 Z=(\cA_{M_1}\cA_{H_1}\cdots\cA_{M_h}\cA_{H_h})[S,T].
\end{equation}
Let $C_a=\{v_{a,0},\ldots,v_{a,N-1}\}$, for $0\le a\le2h$.
Keep only $S$ in the first layer and $[N]\setminus T$ in the last.
The target graph has independent layers and consecutive cut matrices
$M_1,H_1,\ldots,M_h,H_h$, and its perfect-matching count
is~\eqref{eq:monotone-target}. Its $n_0=4ph$ vertices are evenly divided
between the even and odd layers.

\subsection{A bipartite probe identity}

We first give the cancellation identity used by the construction.
Let a bipartite graph $B$ on $n_0$ original vertices have color classes
$X,Y$. For each $r\in\{1,\ldots,h\}$ choose $A_r\subseteq X$ and
$B_r\subseteq Y$. Add pairwise disjoint new independent sets
$P_r,Q_r$, disjoint from $X\cup Y$, with $|P_r|=|Q_r|=s$.
The enlarged color classes are
\[
 X\mathbin{\dot\cup}\bigcup_{r=1}^h Q_r,\qquad
 Y\mathbin{\dot\cup}\bigcup_{r=1}^h P_r.
\]
Add all edges in
\[
 P_r\times Q_r,\qquad P_r\times A_r,\qquad Q_r\times B_r.
\]
No other edges are incident to probe vertices. In particular,
different probe pairs have no edges between them. Denote this graph
by $K_s$. For $x\in X$ and $y\in Y$, define the weight
\[
 w_B(xy)=\ind{xy\in E(B)}
              -\sum_{r=1}^{h}\ind{x\in A_r,\ y\in B_r}.
\]

\begin{lemma}\label{lem:bipartite-probe}
There exists $\Phi\in\mathbb Q[s]$ satisfying
\[
 \begin{aligned}
 \deg\Phi&\le\lfloor n_0/2\rfloor,\\
 \Phi(s)&=\frac{\PM(K_s)}{(s!)^h} &&(s\in\mathbb Z_{\ge0}),\\
 \Phi(-1)&=\PM(K_{X,Y};w_B),
 \end{aligned}
\]
where $K_{X,Y}$ is the complete bipartite graph with parts $X,Y$.
\end{lemma}

\begin{proof}
Fix the original vertices matched into each probe class and the
matching on the remaining original vertices. If $a$ original vertices
are matched into $P_r$, exactly $a$ must be matched into $Q_r$:
the remaining probe vertices can only be matched to one another
within that pair. For prescribed sets of $a$ original vertices on
each side, the number of extensions through the probe pair is
$(s)_a^2(s-a)!$. After division by $s!$, its contribution is
\begin{equation}\label{eq:bipartite-probe-factor}
 \frac{(s)_a^2(s-a)!}{s!}=(s)_a,
 \qquad (s)_a=s(s-1)\cdots(s-a+1).
\end{equation}
The calculation initially applies when $s\ge a$. For nonnegative
$s<a$, the right-hand side is zero, as required. A product of these
factors has degree $\sum_r a_r\le n_0/2$, because it uses
$2\sum_r a_r$ distinct original vertices. Summing over the choices
proves the polynomial assertion.

At $s=-1$, the factor in~\eqref{eq:bipartite-probe-factor} becomes
$(-1)^a a!$. This is the total weight of all bijections between the
two prescribed original sets when every added edge has weight $-1$.
Expanding the sums of original and added edge weights therefore
gives exactly the stated weighted perfect matching count. This argument also
allows the sets attached to different probe pairs to overlap, since
the matching expansion uses each original vertex at most once.
\end{proof}

\subsection{The query graph and its endpoint orders}

For each pair $r$, keep the first cut equal to $M_r$ and change the
second cut to $\mathbf{1}-H_r$, where $\mathbf{1}$ is the all-ones
$N\times N$ matrix. All original layers remain independent, and there
are no edges between nonconsecutive layers. Add independent probe
classes $P_r,Q_r$, each of size $s$, with
\begin{equation}\label{eq:monotone-probe-neighborhoods}
 N(P_r)=C_{2r}\cup Q_r,\qquad
 N(Q_r)=C_{2r-1}\cup P_r.
\end{equation}
Here $N(P_r)$ denotes the common neighborhood of a vertex of $P_r$,
and likewise for $Q_r$. The layers in this display are restricted
at the boundaries as above. There are no edges between different
probe pairs. Write $G_s$ for the resulting graph. Its two color
classes are
\begin{equation}\label{eq:monotone-bipartition}
 \left(\bigcup_{a\text{ even}}C_a\right)\cup\bigcup_r Q_r,
 \qquad
 \left(\bigcup_{a\text{ odd}}C_a\right)\cup\bigcup_r P_r,
\end{equation}
and each has $2ph+hs$ vertices.

We specify two endpoint orders for $G_s$. Temporarily, within pair
$r$, abbreviate
\[
 a_i=v_{2r-2,i},\qquad b_i=v_{2r-1,i},\qquad c_i=v_{2r,i}.
\]
For ordinary cuts $(U,L)$, define the upper and lower blocks
\begin{align}
 T_r&=(b_{N-1},a_{N-1},\ b_{N-2},a_{N-2},\ldots,b_0,a_0),
       \label{eq:monotone-top-block}\\
 D_r&=(b_{N-1},c_{N-1},\ b_{N-2},c_{N-2},\ldots,b_0,c_0).
       \label{eq:monotone-bottom-block}
\end{align}
The first block realizes $U$; the second realizes
$\mathbf{1}-L$, the strict upper triangular matrix. Special cuts
require the following adjacent transpositions:
\begin{equation}\label{eq:monotone-swaps}
\begin{array}{c|c|c|c}
 \text{specified matrix}&\text{block}&\text{labels to interchange}
       &\text{change in the realized cut}\\\hline
 M_r=V_i&T_r&a_{i+1},\ b_i&\text{add }(i+1,i)\\
 M_r=Z_i&T_r&b_i,\ a_i&\text{delete }(i,i)\\
 H_r=V_i^{\mathsf T}&D_r&c_{i+1},\ b_i&\text{delete }(i,i+1)\\
 H_r=Z_i^{\mathsf T}&D_r&b_i,\ c_i&\text{add }(i,i).
\end{array}
\end{equation}
The last two changes implement $\mathbf{1}-H_r$, rather than $H_r$.
Each transposition preserves the decreasing track order within every
original layer.

Regard $P_r,Q_r$ as blocks of $s$ labeled vertices, using the same
internal order for each such block on both lines. The complete upper
order is
\begin{equation}\label{eq:monotone-top-order}
 T_1,P_1,Q_1,\ T_2,P_2,Q_2,\ldots,T_h,P_h,Q_h,\
 v_{2h,N-1},\ldots,v_{2h,0}.
\end{equation}
The complete lower order is
\begin{equation}\label{eq:monotone-bottom-order}
 v_{0,N-1},\ldots,v_{0,0},\
 Q_1,D_1,P_1,\ Q_2,D_2,P_2,\ldots,Q_h,D_h,P_h.
\end{equation}
Finally delete the unwanted boundary vertices from both orders.
When $s=0$, all probe blocks are empty.

\begin{lemma}\label{lem:monotone-representation}
The orders~\eqref{eq:monotone-top-order} and
\eqref{eq:monotone-bottom-order} give a bipartite permutation
representation of $G_s$. Ordering each part by its upper endpoints
gives a monotone ordering.
\end{lemma}

\begin{proof}
Every original layer has decreasing track order on both lines, so
it is independent. Each probe class is also independent, since its
internal orders agree. Nonconsecutive original layers occur in the
same relative order on both lines.

Within pair $r$, all lower endpoints of the $a$-layer precede those
of the $b$-layer. The order in~\eqref{eq:monotone-top-block} makes
$a_u,b_v$ intersect exactly when $u\le v$. Thus the first cut is
$U$ before its possible adjacent transposition. All upper endpoints
of the $b$-layer precede those of the $c$-layer. In
\eqref{eq:monotone-bottom-block}, $b_u$ follows $c_v$ exactly when
$u<v$, giving $\mathbf{1}-L$ on the second cut. The four adjacent
transpositions change only their indicated pairwise orders, so
\eqref{eq:monotone-swaps} yields the prescribed matrices
$M_r,\mathbf{1}-H_r$ and no other changes.

In the upper order, both $P_r$ and $Q_r$ lie after the $b$-layer
and before the $c$-layer. In the lower order, $P_r$ follows the
entire $b,c$ block while $Q_r$ precedes it. Hence $P_r$ intersects
all of the $c$-layer, $Q_r$ intersects all of the $b$-layer, and
neither intersects its own color class. Moreover, $P_r$ precedes
$Q_r$ on the upper line and follows it on the lower line, producing
all edges of $P_r\times Q_r$. Earlier and later blocks have consistent
orders on both lines. Thus there are no other probe edges, including
edges between different probe pairs. This proves the claimed
representation and the bipartition~\eqref{eq:monotone-bipartition}.

By~\eqref{eq:monotone-bipartite-permutation} and the explicit sorting
rule in Section~\ref{sec:prelims}, ordering each color class by its
upper endpoints gives the claimed monotone ordering.
\end{proof}

Figure~\ref{fig:monotone-construction} displays the bipartite
permutation representation and the effect of its paired probes.
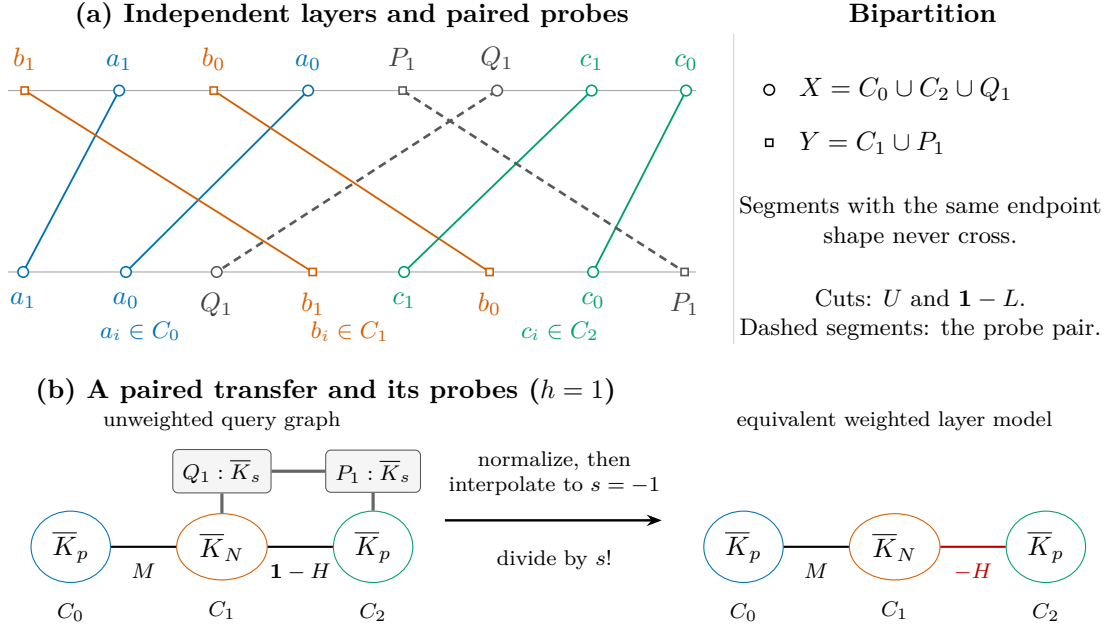
\begin{figure}[htbp]
\centering
\begin{tikzpicture}[x=1cm,y=1cm,font=\small,
  xend/.style={circle,inner sep=1.40pt,fill=white,line width=.6pt},
  yend/.style={rectangle,inner sep=1.35pt,fill=white,line width=.6pt},
  track/.style={line width=.7pt},
  probe/.style={draw=black!65,line width=1pt,densely dashed}]
\definecolor{geomCzero}{HTML}{0072B2}
\definecolor{geomCone}{HTML}{D55E00}
\definecolor{geomCtwo}{HTML}{009E73}
\node[font=\small\bfseries] at (4.60,3.65)
 {(a) Independent layers and paired probes};
\draw[black!35] (.05,2.65)--(9.15,2.65);
\draw[black!35] (.05,.25)--(9.15,.25);
% Upper: b1,a1,b0,a0,P1,Q1,c1,c0.
% Lower: a1,a0,Q1,b1,c1,b0,c0,P1.
% This is the pre-boundary-deletion representation for N=2,h=1,s=1.
\foreach \name/\xt/\xb/\col/\shape in {
 a_0/4.02/1.61/geomCzero/xend,a_1/1.52/.25/geomCzero/xend,
 b_0/2.77/6.42/geomCone/yend,b_1/.27/4.08/geomCone/yend,
 c_0/9.02/7.79/geomCtwo/xend,c_1/7.77/5.29/geomCtwo/xend}{
  \draw[track,draw=\col] (\xt,2.65)--(\xb,.25);
  \node[\shape,draw=\col] at (\xt,2.65) {};
  \node[\shape,draw=\col] at (\xb,.25) {};
  \node[above=4pt,text=\col] at (\xt,2.65) {$\name$};
  \node[below=4pt,text=\col] at (\xb,.25) {$\name$};
}
\foreach \name/\xt/\xb/\shape in {P_1/5.27/9.00/yend,Q_1/6.52/2.81/xend}{
  \draw[probe] (\xt,2.65)--(\xb,.25);
  \node[\shape,draw=black!65] at (\xt,2.65) {};
  \node[\shape,draw=black!65] at (\xb,.25) {};
  \node[above=4pt,text=black!75] at (\xt,2.65) {$\name$};
  \node[below=4pt,text=black!75] at (\xb,.25) {$\name$};
}
\node[font=\footnotesize,text=geomCzero] at (1.80,-.53) {$a_i\in C_0$};
\node[font=\footnotesize,text=geomCone] at (4.57,-.53) {$b_i\in C_1$};
\node[font=\footnotesize,text=geomCtwo] at (7.35,-.53) {$c_i\in C_2$};
\draw[black!18] (9.65,-.70)--(9.65,3.40);
\node[font=\small\bfseries] at (12.15,3.65) {Bipartition};
\node[xend,draw=black] at (10.12,2.66) {};
\node[anchor=west] at (10.38,2.66) {$X=C_0\cup C_2\cup Q_1$};
\node[yend,draw=black] at (10.12,1.96) {};
\node[anchor=west] at (10.38,1.96) {$Y=C_1\cup P_1$};
\node[align=center,font=\footnotesize] at (12.12,.91)
 {Segments with the same endpoint\\shape never cross.};
\node[align=center,font=\footnotesize] at (12.12,-.31)
 {Cuts: $U$ and $\mathbf 1-L$.\\Dashed segments: the probe pair.};
\end{tikzpicture}

\par\medskip
\begin{tikzpicture}[x=1cm,y=1cm,font=\small,
      layer/.style={ellipse,draw,minimum width=1.03cm,minimum height=.65cm,fill=white},
      probe/.style={rectangle,rounded corners=2pt,draw=black!65,fill=black!4,
                    minimum width=.77cm,minimum height=.5cm,font=\scriptsize},
      cut/.style={line width=.8pt},
      complete/.style={draw=black!60,line width=1.2pt},
      lab/.style={font=\scriptsize,below=5pt,inner sep=1pt}]
\definecolor{geomCzero}{HTML}{0072B2}
\definecolor{geomCone}{HTML}{D55E00}
\definecolor{geomCtwo}{HTML}{009E73}
\node[anchor=west,font=\small\bfseries] at (-.6,2.2)
  {(b) A paired transfer and its probes ($h=1$)};
\node[font=\scriptsize] at (2,1.82) {unweighted query graph};
\node[font=\scriptsize] at (10.9,1.82) {equivalent weighted layer model};

\node[layer,draw=geomCzero] (l0) at (0,.15) {$\overline K_p$};
\node[font=\scriptsize,below=4pt] at (l0.south) {$C_0$};
\node[layer,draw=geomCone] (l1) at (2,.15) {$\overline K_N$};
\node[font=\scriptsize,below=4pt] at (l1.south) {$C_1$};
\node[layer,draw=geomCtwo] (l2) at (4,.15) {$\overline K_p$};
\node[font=\scriptsize,below=4pt] at (l2.south) {$C_2$};
\node[layer,draw=geomCzero] (r0) at (8.9,.15) {$\overline K_p$};
\node[font=\scriptsize,below=4pt] at (r0.south) {$C_0$};
\node[layer,draw=geomCone] (r1) at (10.9,.15) {$\overline K_N$};
\node[font=\scriptsize,below=4pt] at (r1.south) {$C_1$};
\node[layer,draw=geomCtwo] (r2) at (12.9,.15) {$\overline K_p$};
\node[font=\scriptsize,below=4pt] at (r2.south) {$C_2$};
\draw[cut] (l0)--node[lab]{$M$}(l1);
\draw[cut] (l1)--node[lab]{$\mathbf1-H$}(l2);
\draw[cut,black] (r0)--node[lab]{$M$}(r1);
\draw[cut,red!75!black] (r1)--node[lab]{$-H$}(r2);
\node[probe] (q) at (2,1.15) {$Q_1:\overline K_s$};
\node[probe] (p) at (4,1.15) {$P_1:\overline K_s$};
\draw[complete] (q)--(p); \draw[complete] (q)--(l1);
\draw[complete] (p)--(l2);
\draw[-{Stealth[length=2mm]},line width=.9pt] (4.95,.5)--(7.82,.5);
\node[font=\scriptsize,align=center] at (6.4,1.12) {normalize, then\\interpolate to $s=-1$};
\node[font=\scriptsize,align=center] at (6.4,-.02) {divide by $s!$};
\end{tikzpicture}

\caption{Independent layers and paired bipartite probes.
Panel (a) is the ordinary $N=2,h=1,s=1$ permutation diagram before
boundary deletion. Circles and squares mark the two independent color
classes; layer colors distinguish $C_0,C_1,C_2$, and dashed segments
are the probes. Equal internal endpoint orders expand each probe to
an independent set of size $s$. In (b), the probe pair forms $K_{s,s}$,
with complete connections $Q_1$--$C_1$ and $P_1$--$C_2$;
the boundary layers have been restricted to $p$ vertices, with the
corresponding cut restrictions suppressed in the labels.
Normalization by $s!$ and interpolation to $s=-1$ subtract the
all-ones matrix from the second cut, giving $-H$, with the first cut
unchanged. The right-hand side represents a weighted-count identity;
every oracle query remains unweighted.}
\label{fig:monotone-construction}
\end{figure}

\subsection{Interpolation and the oracle reduction}

For every nonnegative integer $s$, let
\begin{equation}\label{eq:monotone-probe-polynomial}
 \Psi(s)=\frac{\PM(G_s)}{(s!)^h}.
\end{equation}
Apply Lemma~\ref{lem:bipartite-probe} with
$A_r=C_{2r}$ and $B_r=C_{2r-1}$. It gives a polynomial of degree
at most $n_0/2=2ph$. Evaluating at $s=-1$ subtracts the all-ones
matrix from each second cut, so its weight matrix becomes
\[
 (\mathbf{1}-H_r)-\mathbf{1}=-H_r.
\]
No first-cut edge changes, and no intralayer edge is introduced.
The remaining graph is consequently the target independent-layer
graph, with weight $-1$ on every edge of a second cut and weight
$1$ on every edge of a first cut.

Every target perfect matching uses exactly $p$ edges at each cut.
There are $h$ second cuts, so every surviving matching has the same
sign $(-1)^{ph}$. Therefore
\begin{equation}\label{eq:monotone-cancellation}
 (-1)^{ph}\Psi(-1)=Z.
\end{equation}
All oracle queries remain
simple and unweighted; the negative weights only describe the
polynomial identity evaluated in the reduction.

The $2ph+1$ samples $s=0,1,\ldots,2ph$ suffice for interpolation.
The same parameter is used for all probe pairs, so the number of
samples does not grow exponentially with the number of pairs. Each
query has at most
\begin{equation}\label{eq:monotone-query-size}
 n_0+2h(2ph)=4ph(h+1)
\end{equation}
vertices. Its endpoint orders and its monotone ordering are explicit;
the endpoint coordinates can be their integer ranks.

\begin{proof}[Proof of Theorem~\ref{thm:monotone}]
The preceding construction computes the paired-transfer value in
time polynomial in $N+h$ from perfect-matching counts on monotone
graphs. For a query with $v$ vertices, its answer and the normalizing
factor $(s!)^h$ have $O(v\log v)$ bits. The interpolation nodes
have polynomial magnitude, and the rational interpolation
coefficients and intermediate values also have polynomial bit
length. Thus the normalization, interpolation, and multiplication
by $(-1)^{ph}$ are polynomial-time exact computations.

We have proved $\mathrm{PairEval}\le_T
\#\mathrm{PM}(\mathrm{Monotone})$.
Theorem~\ref{thm:word-hard} therefore proves
$\sharpP$-hardness on monotone graphs. Membership in $\sharpP$
follows by guessing a canonically encoded perfect matching and
checking its edges and vertex coverage. Both a monotone ordering
and a permutation diagram are produced by the construction.
\end{proof}

\section{Unit interval graphs}\label{sec:unitinterval}

We now realize the paired-transfer interface geometrically. In this
section the original layers are cliques, and auxiliary clique probes
remove their unwanted edges by a single interpolation parameter.

\begin{proposition}\label{prop:oracle}
$\mathrm{PairEval}\le_T\UIPM$.
\end{proposition}

\subsection{An equal-length realization of the pairs}

Fix $N=2p$ and a paired layered instance with $h\ge1$ pairs from $\mathcal P_N$,
and write the two matrices of a pair as $(M,H)$. We first build an unweighted graph whose
layers are cliques, whose first cut in a pair is the entrywise
complement $\mathbf 1-M$, and whose second cut is $H$.
Here $\mathbf 1$ is the $N\times N$ all-ones matrix. A
\emph{profile} is the increasing vector of offsets assigned to
the tracks of a layer; the layer's global position is added to
these offsets in~\eqref{eq:intervals}.

Set
\begin{equation}\label{eq:geometric-scale}
 \delta=1000(h+1),\qquad \Lambda=(N+2)\delta,
 \qquad a_j=(j+1)\delta\quad(0\le j<N)
\end{equation}
and take $a=(a_0,\ldots,a_{N-1})$ as the initial profile.
A pair has input, middle, and output profiles $a,b,c$. Normally put $c=a$ and
\[
 b_j=\left\lfloor\frac{a_j+a_{j+1}}2\right\rfloor\quad(j<N-1),
 \qquad b_{N-1}=a_{N-1}+\delta/2.
\]
For a special pair make only the changes in
Table~\ref{tab:geometry}. The output $c$ becomes the next pair's
input. Associate profile $a^{(j)}$ at layer $j$ with intervals
\begin{equation}\label{eq:intervals}
 I_{j,v}=[j\Lambda+a^{(j)}_v,(j+1)\Lambda+a^{(j)}_v].
\end{equation}

\begin{table}[htbp]
\centering
\caption{Changes from the ordinary middle and output profiles.
Unlisted coordinates retain their ordinary values.}
\label{tab:geometry}
\begin{tabular}{llll}
\toprule
Type & Pair $(M,H)$ & Middle profile & Output profile\\
\midrule
$\R_i$ & $(V_i,L)$ & $b_i=a_{i+1}+1$ & $c_{i+1}=a_{i+1}+2$\\
$\D_i$ & $(Z_i,L)$ & $b_i=a_i-1$ & $c_i=a_i-2$\\
$\B_i$ & $(U,V_i^{\mathsf T})$ & $b_i=a_{i+1}-1$ & $c_{i+1}=a_{i+1}-2$\\
$\E_i$ & $(U,Z_i^{\mathsf T})$ & $b_i=a_i+1$ & $c_i=a_i+2$\\
\bottomrule
\end{tabular}
\end{table}

Every profile is strictly increasing, and all its coordinates lie in $(0,\Lambda)$.
Indeed, any input coordinate differs from its initial value by at
most $2h$, while adjacent initial coordinates differ by $\delta$.
These inequalities also show that all entries of
Table~\ref{tab:geometry} have the claimed comparisons with their
neighbors. Two intervals in one layer intersect, and nonconsecutive
layers do not meet. The consecutive cuts are determined by
\[
 \ind{a_u\ge b_v}=1-M_{uv},\qquad
 \ind{b_u\ge c_v}=H_{uv}.
\]
For ordinary profiles these are the strict lower and weak lower
triangular relations. Moving the single middle coordinate in each
row of the table changes exactly the prescribed comparison; the
output change makes the second cut equal to the listed $H$.
This verifies the entire interval representation. Delete first-
and last-layer vertices according to the boundary convention in
\eqref{eq:transfer}.

\subsection{Clique cancellation using equal-length probes}

There are $2h+1$ layers. For $j=0,\ldots,h$, introduce a probe
interval
\begin{equation}\label{eq:probes}
 Q_j=[(2j+1)\Lambda,(2j+2)\Lambda].
\end{equation}
Its original neighbors are exactly layers $2j$ and $2j+1$,
with the latter absent for $j=h$. Different probe intervals are
disjoint. Replace each probe by $s$ labeled identical copies,
where $s$ is a common nonnegative even integer. Copies of a probe
form a clique; different probe classes have no edges between them.
Denote the resulting unit interval graph by $G_s$.
For $a\ge0$, the odd double factorial
$(2a-1)!!=\prod_{j=1}^{a}(2j-1)$, with $(-1)!!=1$, counts
perfect matchings of a clique on $2a$ labeled vertices.

Let $B$ be any graph on $n_0$ original vertices. For each of
$b$ probe classes add $s$ vertices forming a clique, with one common
original neighborhood $A_r\subseteq V(B)$ for class $r$, and no
edges between different probe classes. Write $B_s$ for the resulting
graph and $K_{V(B)}$ for the complete graph on the original vertices.
For distinct original vertices $u,v$, define
\[
 w_B(uv)=\ind{uv\in E(B)}-\sum_{r=1}^{b}\ind{u,v\in A_r}.
\]

\begin{lemma}\label{lem:clones}\label{lem:clique-probe}
There exists $\Psi\in\mathbb Q[s]$ satisfying
\[
 \begin{aligned}
 \deg\Psi&\le\lfloor n_0/2\rfloor,\\
 \Psi(s)&=\frac{\PM(B_s)}{((s-1)!!)^{b}}
                  &&(s\in2\mathbb Z_{\ge0}),\\
 \Psi(-1)&=\PM(K_{V(B)};w_B).
 \end{aligned}
\]
\end{lemma}

\begin{proof}
Fix a nonnegative even integer $s$. Every perfect matching of $B_s$
uniquely determines the sets $X_r\subseteq A_r$ of original vertices
matched to probe class $r$. These sets are pairwise disjoint, and
$|X_r|=2a_r$ is even because the remaining vertices of each probe
class must be matched internally. For a fixed tuple $(X_1,\ldots,X_b)$
with $s\ge2a_r$ for every $r$, the number of such matchings is
\[
 \PM\!\left(B-\bigcup_{r=1}^{b}X_r\right)
 \prod_{r=1}^{b}(s)_{2a_r}(s-2a_r-1)!!.
\]
Here $(s)_{2a}=s(s-1)\cdots(s-2a+1)$ counts injections of the
$2a$ distinct original vertices into the $s$ labeled probe vertices;
the double factorial counts the internal matching on the remaining
probe vertices. The other factor counts the matching on the remaining
original vertices.

Define $f_a(s)=\prod_{j=0}^{a-1}(s-2j)$, with $f_0(s)=1$.
Dividing the contribution of a probe class by $(s-1)!!$ gives
\begin{equation}\label{eq:clone-factor}
 \frac{(s)_{2a}(s-2a-1)!!}{(s-1)!!}
 =f_a(s) \qquad(s\ge2a,\ s\text{ even}).
\end{equation}
If $s<2a$, there are no such injections, and $f_a(s)=0$ for every
nonnegative even $s$ in this range. Consequently the polynomial
\[
 \Psi(s)=
 \sum_{\substack{X_r\subseteq A_r\ (1\le r\le b)\\
                  X_r\cap X_t=\varnothing\ (r\ne t)\\
                  |X_r|\ \mathrm{even}}}
 \PM\!\left(B-\bigcup_{r=1}^{b}X_r\right)
 \prod_{r=1}^{b}f_{|X_r|/2}(s)
\]
equals $\PM(B_s)/((s-1)!!)^{b}$ at every nonnegative even integer.
Each summand has degree at most $\sum_r a_r\le\lfloor n_0/2\rfloor$.

At $s=-1$, we have $f_a(-1)=(-1)^a(2a-1)!!$. This is the total
weight of perfect matchings on $2a$ vertices with every edge assigned
weight $-1$. Thus $\Psi(-1)$ counts weighted perfect matchings obtained
by retaining the edges of $B$ with weight $1$ and, for each $r$, adding
a separately labeled edge of weight $-1$ between every pair in $A_r$.
Grouping such matchings by the vertices using edges labeled $r$
gives exactly the displayed sum. Expanding the sum of the edge weights
at each pair of endpoints yields $w_B(uv)$, including when several
probe neighborhoods overlap or an original edge has the same
endpoints. Hence $\Psi(-1)=\PM(K_{V(B)};w_B)$.
\end{proof}

Figure~\ref{fig:unit-interval-construction} combines a small interval
realization with the effect of the probe identity.
\begin{figure}[htbp]
\centering
\begin{tikzpicture}[x=.52cm,y=1cm,font=\small,
  c0/.style={draw=geomCzero},
  c1/.style={draw=geomCone},
  c2/.style={draw=geomCtwo},
  interval/.style={line width=1.4pt},
  probe/.style={draw=black!65,line width=1.5pt,dashed}]
\definecolor{geomCzero}{HTML}{0072B2}
\definecolor{geomCone}{HTML}{D55E00}
\definecolor{geomCtwo}{HTML}{009E73}
\node[anchor=west,font=\small\bfseries] at (-1.3,.65)
  {(a) Equal-length intervals for an ordinary pair};
\foreach \x/\lab in {0/0,6/\Lambda,12/2\Lambda,18/3\Lambda,24/4\Lambda}{
  \draw[gray!25] (\x,.12)--(\x,-3.22);
  \node[font=\scriptsize,above] at (\x,.12) {$\lab$};
}
\foreach \a/\b/\y/\lab/\col in {
  1/7/-.35/a_0/c0,3/9/-.68/a_1/c0,
  8/14/-1.15/b_0/c1,10/16/-1.48/b_1/c1,
  13/19/-1.95/c_0/c2,15/21/-2.28/c_1/c2}{
  \draw[interval,\col] (\a,\y)--(\b,\y);
  \draw[\col] (\a,{\y-.065})--(\a,{\y+.065});
  \draw[\col] (\b,{\y-.065})--(\b,{\y+.065});
  \node[anchor=east,font=\scriptsize] at (-.4,\y) {$\lab$};
}
\foreach \a/\b/\y/\lab in {6/12/-2.75/Q_0,18/24/-3.12/Q_1}{
  \draw[probe] (\a,\y)--(\b,\y);
  \draw[black!65] (\a,{\y-.065})--(\a,{\y+.065});
  \draw[black!65] (\b,{\y-.065})--(\b,{\y+.065});
  \node[anchor=east,font=\scriptsize] at (-.4,\y) {$\lab$};
}
\end{tikzpicture}

\par\medskip
\begin{tikzpicture}[x=1cm,y=1cm,font=\small,
      layer/.style={ellipse,draw,minimum width=1.03cm,minimum height=.65cm,fill=white},
      probe/.style={rectangle,rounded corners=2pt,draw=black!65,fill=black!4,
                    minimum width=.77cm,minimum height=.5cm,font=\scriptsize},
      cut/.style={line width=.8pt},
      complete/.style={draw=black!60,line width=1.2pt},
      lab/.style={font=\scriptsize,below=5pt,inner sep=1pt}]
\definecolor{geomCzero}{HTML}{0072B2}
\definecolor{geomCone}{HTML}{D55E00}
\definecolor{geomCtwo}{HTML}{009E73}
\node[anchor=west,font=\small\bfseries] at (-.6,2.2)
  {(b) A paired transfer and its probes ($h=1$)};
\node[font=\scriptsize] at (2,1.82) {unweighted query graph};
\node[font=\scriptsize] at (10.9,1.82) {equivalent weighted layer model};

\node[layer,draw=geomCzero] (l0) at (0,.15) {$K_p$};
\node[font=\scriptsize,below=4pt] at (l0.south) {$C_0$};
\node[layer,draw=geomCone] (l1) at (2,.15) {$K_N$};
\node[font=\scriptsize,below=4pt] at (l1.south) {$C_1$};
\node[layer,draw=geomCtwo] (l2) at (4,.15) {$K_p$};
\node[font=\scriptsize,below=4pt] at (l2.south) {$C_2$};
\node[layer,draw=geomCzero] (r0) at (8.9,.15) {$\overline K_p$};
\node[font=\scriptsize,below=4pt] at (r0.south) {$C_0$};
\node[layer,draw=geomCone] (r1) at (10.9,.15) {$\overline K_N$};
\node[font=\scriptsize,below=4pt] at (r1.south) {$C_1$};
\node[layer,draw=geomCtwo] (r2) at (12.9,.15) {$\overline K_p$};
\node[font=\scriptsize,below=4pt] at (r2.south) {$C_2$};
\draw[cut] (l0)--node[lab]{$\mathbf1-M$}(l1);
\draw[cut] (l1)--node[lab]{$H$}(l2);
\draw[cut,red!75!black] (r0)--node[lab]{$-M$}(r1);
\draw[cut,black] (r1)--node[lab]{$H$}(r2);
\node[probe] (q0) at (1,1.15) {$Q_0:K_s$};
\node[probe] (q1) at (4,1.15) {$Q_1:K_s$};
\draw[complete] (q0)--(l0); \draw[complete] (q0)--(l1);
\draw[complete] (q1)--(l2);
\draw[-{Stealth[length=2mm]},line width=.9pt] (4.95,.5)--(7.82,.5);
\node[font=\scriptsize,align=center] at (6.4,1.12) {normalize, then\\interpolate to $s=-1$};
\node[font=\scriptsize,align=center] at (6.4,-.02) {divide by $((s-1)!!)^{2}$};
\end{tikzpicture}

\caption{Equal-length realization and clique-probe cancellation.
In (a), $a_i,b_i,c_i$ are the two tracks of $C_0,C_1,C_2$.
The ordinary pair uses offsets $(1,3),(2,4),(1,3)$ and interval length
$\Lambda=6$; only the relative endpoint comparisons matter in this
illustration. Each dashed probe interval is replaced by $s$ identical
copies. In (b), a line incident to a probe denotes all edges to the
indicated layer; the other connections have their displayed cut matrices.
After normalization and interpolation, the probes cancel every
intralayer edge, and $\mathbf1-M$ becomes $-M$.
Panel (a) precedes boundary deletion; (b) uses the retained boundary
layers of size $p$, with the boundary restrictions implicit in the cut labels.
The right-hand picture expresses the weighted-count identity at $s=-1$.}
\label{fig:unit-interval-construction}
\end{figure}
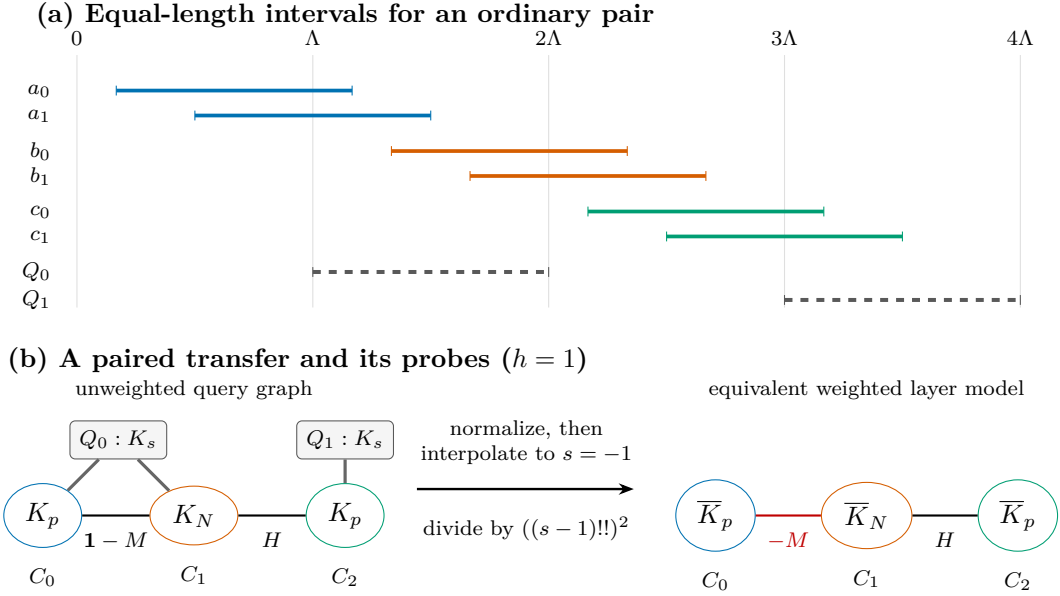

Apply the lemma with $b=h+1$ and
$\Psi(s)=\PM(G_s)/((s-1)!!)^{h+1}$.
In our construction, each layer belongs to exactly one probe
neighborhood. Thus all intralayer edges cancel at $s=-1$.
At the first cut of a pair, the matrix $\mathbf 1-M$ becomes
$-M$, while the second cut $H$ is unchanged. Every remaining
perfect matching has $p$ edges at each cut, so it has the same
sign $(-1)^{ph}$. Consequently
\begin{equation}\label{eq:sign}
 (-1)^{ph}\Psi(-1)
 = (\cA_{M_1}\cA_{H_1}\cdots\cA_{M_h}\cA_{H_h})[S,T],
\end{equation}
where the matrices on the right are the desired uncomplemented
paired cuts in~\eqref{eq:paired-value}.

\begin{proof}[Completion of the proof of Proposition~\ref{prop:oracle}]
For an instance with $h$ pairs, the original number of vertices is
$n_0=4ph$: the first and last layers each have $p$ vertices and
all $2h-1$ internal layers have $2p$ vertices. The $n_0/2+1$ queries
$s=0,2,\ldots,n_0$ determine $\Psi(-1)$ by interpolation, and
\eqref{eq:sign} recovers the desired paired-transfer value.
Each query graph has at most
\[
 n_0+(h+1)n_0=4ph(h+2)
\]
vertices, and the absolute values of all interval endpoints are
$O(Nh^2)$. Thus both the number and sizes of the queries are
polynomial in $N+h$. A query on $v$ vertices has at most $(v-1)!!$
perfect matchings, whose bit length is $O(v\log v)$; the
normalizing double factorials also have polynomial bit length.
The exact interpolation facts in Section~\ref{sec:prelims} therefore
give polynomial-time normalization and recovery of~\eqref{eq:sign}.
This proves the proposition.
\end{proof}

Every query in Proposition~\ref{prop:oracle} comes with the
equal-length interval representation constructed above, whose
endpoints are integers.

\begin{theorem}\label{thm:unitinterval}
$\UIPM$ is $\sharpP$-complete.
\end{theorem}
\begin{proof}
Theorem~\ref{thm:word-hard} and Proposition~\ref{prop:oracle} give
\[
 \IS\le_T\mathrm{PairEval}\le_T\UIPM.
\]
A canonical list of the pairs in a perfect matching is a polynomial-time
verifiable certificate, giving membership in $\sharpP$.
\end{proof}
\begin{corollary}\label{cor:integer}
Given integers $x_v$ and a positive integer $r$ as part of the input,
computing
the number of perfect matchings in the graph
\[
 uv\in E\quad\Longleftrightarrow\quad |x_u-x_v|<r
 \qquad(u\ne v)
\]
is $\sharpP$-complete.
\end{corollary}

\begin{proof}
Every oracle graph in Proposition~\ref{prop:oracle}, and hence
in the reduction proving Theorem~\ref{thm:unitinterval}, has an explicitly
constructed integer equal-length representation. In that model,
closed-interval intersection means $|x_u-x_v|\le\Lambda$.
This is equivalent to $|x_u-x_v|<\Lambda+1$. Set
$r=\Lambda+1$. Membership in $\sharpP$ is immediate.
\end{proof}

\section{Chordal permutation graphs}\label{sec:permutation}

We use the two-order convention for permutation diagrams from
Section~\ref{sec:prelims}. A private probe for each layer permits clique
cancellation while preserving both the permutation representation and
chordality.

\begin{theorem}\label{thm:chordal-permutation}
$\#\mathrm{PM}(\mathrm{ChordalPermutation})$ is $\sharpP$-complete.
\end{theorem}

We construct an oracle for $\mathrm{PairEval}$. Fix $N=2p$, $h\ge1$, boundary
states $S,T\in\cS_{N,p}$, and pairs
$(M_1,H_1),\ldots,(M_h,H_h)$, where
\[
 M_r\in\{U,V_i,Z_i:0\le i\le N-2\},\qquad
 H_r\in\{L,V_i^{\mathsf T},Z_i^{\mathsf T}:0\le i\le N-2\}.
\]
The required value is
\begin{equation}\label{eq:permutation-target}
 Z=(\cA_{M_1}\cA_{H_1}\cdots\cA_{M_h}\cA_{H_h})[S,T].
\end{equation}
Use layers $C_a=\{v_{a,0},\ldots,v_{a,N-1}\}$, for $0\le a\le2h$.
The cuts from $C_{2r-2}$ to $C_{2r-1}$ and from $C_{2r-1}$ to
$C_{2r}$ have matrices $M_r$ and $H_r$, respectively, with no edges
between nonconsecutive layers. Retain only the tracks in $S$ in $C_0$
and only the tracks in $[N]\setminus T$ in $C_{2h}$. There are
$n_0=4ph$ retained vertices. When all layers are independent, the
perfect-matching count is~\eqref{eq:permutation-target}.

\subsection{A representation with one private probe per layer}

For the oracle construction, make every original layer a clique and
add a vertex $q_a$ whose original neighborhood is exactly $C_a$.
The vertices $q_0,\ldots,q_{2h}$ are pairwise nonadjacent. We first
give the orders before deleting the unwanted boundary vertices.

For ordinary cuts $(U,L)$, the upper order is obtained by concatenating,
for $a=0,2,\ldots,2h-2$, the blocks
\begin{equation}\label{eq:permutation-top-block}
 q_a,\ v_{a+1,N-1},v_{a,N-1},\
 v_{a+1,N-2},v_{a,N-2},\ldots,
 v_{a+1,0},v_{a,0},\ q_{a+1},
\end{equation}
and then appending
\begin{equation}\label{eq:permutation-top-end}
 q_{2h},\ v_{2h,N-1},\ldots,v_{2h,0}.
\end{equation}
The lower order starts with
\begin{equation}\label{eq:permutation-bottom-start}
 v_{0,0},\ldots,v_{0,N-1},\ q_0,
\end{equation}
and then concatenates, for $a=1,3,\ldots,2h-1$, the blocks
\begin{equation}\label{eq:permutation-bottom-block}
 q_a,\ v_{a+1,0},v_{a,0},\
 v_{a+1,1},v_{a,1},\ldots,
 v_{a+1,N-1},v_{a,N-1},\ q_{a+1}.
\end{equation}

Each special cut is realized by one adjacent transposition in its
corresponding block. The complete list is
\begin{equation}\label{eq:permutation-swaps}
\begin{array}{c|c|c}
 \text{cut and its left-layer index}&\text{order}&
       \text{labels to interchange}\\\hline
 V_i\quad(a\text{ even})&\text{upper}&v_{a,i+1},\ v_{a+1,i}\\
 Z_i\quad(a\text{ even})&\text{upper}&v_{a+1,i},\ v_{a,i}\\
 V_i^{\mathsf T}\quad(a\text{ odd})&\text{lower}&
       v_{a,i},\ v_{a+1,i+1}\\
 Z_i^{\mathsf T}\quad(a\text{ odd})&\text{lower}&
       v_{a+1,i},\ v_{a,i}.
\end{array}
\end{equation}
Different cuts use different blocks, so these modifications can be
performed independently.

\begin{lemma}\label{lem:permutation-private-probes}
The two orders above represent precisely the graph with clique layers,
cut matrices $M_r,H_r$, and
\[
 N(q_a)=C_a\quad(0\le a\le2h),\qquad
 E(C_a,C_b)=\varnothing\quad(|a-b|>1).
\]
\end{lemma}

\begin{proof}
Within each layer, the upper track order is decreasing and the lower
track order is increasing. Hence each layer is a clique. Nonconsecutive
layers appear in the same relative order on both lines and have no
edges between them.

At an even cut, all lower endpoints of the earlier layer precede
those of the later layer. In~\eqref{eq:permutation-top-block}, the
upper endpoint of $v_{a,u}$ follows that of $v_{a+1,v}$ exactly when
$u\le v$. Thus this cut is $U$. At an odd cut, the earlier layer
precedes the later layer in the upper order, and
\eqref{eq:permutation-bottom-block} reverses the order of
$v_{a,u},v_{a+1,v}$ exactly when $u\ge v$. Thus this cut is $L$.
The adjacent transpositions in~\eqref{eq:permutation-swaps} change
only the specified pairwise order: the two $V_i$ modifications add
the required edge, and the two $Z_i$ modifications delete the required
diagonal edge. No within-layer order changes.

For even $a$, the probe $q_a$ precedes its upper two-layer block
and follows its lower two-layer block; these blocks share only layer
$C_a$. It therefore intersects all of $C_a$ and none of the adjacent
layers. For odd $a$, it follows the upper block and precedes the
lower block, with the same conclusion. The singleton blocks
\eqref{eq:permutation-top-end} and
\eqref{eq:permutation-bottom-start} give the assertion at the two
boundaries. Earlier and later blocks cannot create any additional
probe edges. Finally, the probes occur in the order
$q_0,q_1,\ldots,q_{2h}$ on both lines, so no two probes are adjacent.
Deleting a vertex from both orders takes an induced subgraph, proving
the boundary assertion.
\end{proof}

For a nonnegative even integer $s$, replace every probe $q_a$ by a
clique $Q_a$ of $s$ labeled vertices, each with the same original
neighborhood $C_a$. Different probe cliques have no edges between
them. In the representation, replace the probe label by $s$
consecutive labels, in opposite internal orders on the two lines.
All external relative orders are preserved, so the resulting graph
$G_s$ is a permutation graph. When $s=0$, the probe label is simply
deleted.

Figure~\ref{fig:permutation-construction} illustrates the endpoint
orders, a single local modification, and the cancellation used below.
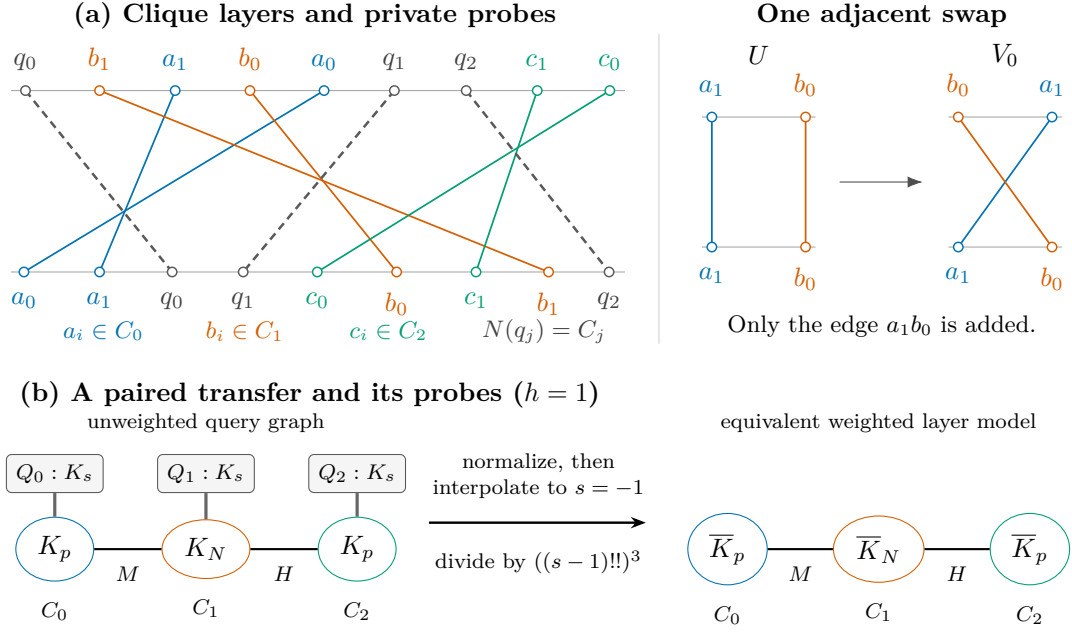
\begin{figure}[htbp]
\centering
\begin{tikzpicture}[x=1cm,y=1cm,font=\small,
  endpoint/.style={circle,inner sep=1.25pt,fill=white,line width=.55pt},
  track/.style={line width=.65pt},
  probe/.style={draw=black!65,line width=.95pt,densely dashed}]
\definecolor{geomCzero}{HTML}{0072B2}
\definecolor{geomCone}{HTML}{D55E00}
\definecolor{geomCtwo}{HTML}{009E73}
\node[font=\small\bfseries] at (4.1,3.65) {(a) Clique layers and private probes};
\node[font=\small\bfseries] at (11.6,3.65) {One adjacent swap};
\draw[black!35] (.08,2.65)--(8.20,2.65);
\draw[black!35] (.08,.25)--(8.20,.25);
% Upper: q0,b1,a1,b0,a0,q1,q2,c1,c0.
% Lower: a0,a1,q0,q1,c0,b0,c1,b1,q2.
% Nonuniform lower coordinates preserve these orders and avoid triple crossings.
\foreach \name/\xt/\xb/\col in {
 a_0/4.22/.25/geomCzero,a_1/2.25/1.25/geomCzero,
 b_0/3.24/5.18/geomCone,b_1/1.25/7.19/geomCone,
 c_0/8.00/4.13/geomCtwo,c_1/7.04/6.23/geomCtwo}{
  \draw[track,draw=\col] (\xt,2.65)--(\xb,.25);
  \node[endpoint,draw=\col] at (\xt,2.65) {};
  \node[endpoint,draw=\col] at (\xb,.25) {};
  \node[above=4pt,text=\col] at (\xt,2.65) {$\name$};
  \node[below=4pt,text=\col] at (\xb,.25) {$\name$};
}
\foreach \name/\xt/\xb in {q_0/.27/2.21,q_1/5.15/3.15,q_2/6.10/7.99}{
  \draw[probe] (\xt,2.65)--(\xb,.25);
  \node[endpoint,draw=black!65] at (\xt,2.65) {};
  \node[endpoint,draw=black!65] at (\xb,.25) {};
  \node[above=4pt,text=black!75] at (\xt,2.65) {$\name$};
  \node[below=4pt,text=black!75] at (\xb,.25) {$\name$};
}
\node[font=\footnotesize,text=geomCzero] at (1.30,-.53) {$a_i\in C_0$};
\node[font=\footnotesize,text=geomCone] at (3.17,-.53) {$b_i\in C_1$};
\node[font=\footnotesize,text=geomCtwo] at (5.07,-.53) {$c_i\in C_2$};
\node[font=\footnotesize,text=black!75] at (7.13,-.53) {$N(q_j)=C_j$};
\draw[black!18] (8.65,-.70)--(8.65,3.40);
% The zoom retains the order of the only two transposed upper endpoints.
\node at (9.97,3.12) {$U$};
\node at (13.23,3.12) {$V_0$};
\foreach \l/\r in {9.22/10.72,12.48/13.98}{
  \draw[black!35] (\l,2.30)--(\r,2.30);
  \draw[black!35] (\l,.58)--(\r,.58);
}
\draw[track,draw=geomCzero] (9.35,2.30)--(9.35,.58);
\draw[track,draw=geomCone] (10.59,2.30)--(10.59,.58);
\draw[track,draw=geomCzero] (13.85,2.30)--(12.61,.58);
\draw[track,draw=geomCone] (12.61,2.30)--(13.85,.58);
\foreach \name/\xt/\xb/\col in {
 a_1/9.35/9.35/geomCzero,b_0/10.59/10.59/geomCone,
 a_1/13.85/12.61/geomCzero,b_0/12.61/13.85/geomCone}{
  \node[endpoint,draw=\col] at (\xt,2.30) {};
  \node[endpoint,draw=\col] at (\xb,.58) {};
  \node[above=4pt,text=\col] at (\xt,2.30) {$\name$};
  \node[below=4pt,text=\col] at (\xb,.58) {$\name$};
}
\draw[-{Latex[length=2mm]},black!70] (11.04,1.44)--(12.14,1.44);
\node[font=\footnotesize,align=center] at (11.60,-.47)
 {Only the edge $a_1b_0$ is added.};
\end{tikzpicture}

\par\medskip
\begin{tikzpicture}[x=1cm,y=1cm,font=\small,
      layer/.style={ellipse,draw,minimum width=1.03cm,minimum height=.65cm,fill=white},
      probe/.style={rectangle,rounded corners=2pt,draw=black!65,fill=black!4,
                    minimum width=.77cm,minimum height=.5cm,font=\scriptsize},
      cut/.style={line width=.8pt},
      complete/.style={draw=black!60,line width=1.2pt},
      lab/.style={font=\scriptsize,below=5pt,inner sep=1pt}]
\definecolor{geomCzero}{HTML}{0072B2}
\definecolor{geomCone}{HTML}{D55E00}
\definecolor{geomCtwo}{HTML}{009E73}
\node[anchor=west,font=\small\bfseries] at (-.6,2.2)
  {(b) A paired transfer and its probes ($h=1$)};
\node[font=\scriptsize] at (2,1.82) {unweighted query graph};
\node[font=\scriptsize] at (10.9,1.82) {equivalent weighted layer model};

\node[layer,draw=geomCzero] (l0) at (0,.15) {$K_p$};
\node[font=\scriptsize,below=4pt] at (l0.south) {$C_0$};
\node[layer,draw=geomCone] (l1) at (2,.15) {$K_N$};
\node[font=\scriptsize,below=4pt] at (l1.south) {$C_1$};
\node[layer,draw=geomCtwo] (l2) at (4,.15) {$K_p$};
\node[font=\scriptsize,below=4pt] at (l2.south) {$C_2$};
\node[layer,draw=geomCzero] (r0) at (8.9,.15) {$\overline K_p$};
\node[font=\scriptsize,below=4pt] at (r0.south) {$C_0$};
\node[layer,draw=geomCone] (r1) at (10.9,.15) {$\overline K_N$};
\node[font=\scriptsize,below=4pt] at (r1.south) {$C_1$};
\node[layer,draw=geomCtwo] (r2) at (12.9,.15) {$\overline K_p$};
\node[font=\scriptsize,below=4pt] at (r2.south) {$C_2$};
\draw[cut] (l0)--node[lab]{$M$}(l1);
\draw[cut] (l1)--node[lab]{$H$}(l2);
\draw[cut,black] (r0)--node[lab]{$M$}(r1);
\draw[cut,black] (r1)--node[lab]{$H$}(r2);
\node[probe] (q0) at (0,1.15) {$Q_0:K_s$};
\draw[complete] (q0)--(l0);
\node[probe] (q1) at (2,1.15) {$Q_1:K_s$};
\draw[complete] (q1)--(l1);
\node[probe] (q2) at (4,1.15) {$Q_2:K_s$};
\draw[complete] (q2)--(l2);
\draw[-{Stealth[length=2mm]},line width=.9pt] (4.95,.5)--(7.82,.5);
\node[font=\scriptsize,align=center] at (6.4,1.12) {normalize, then\\interpolate to $s=-1$};
\node[font=\scriptsize,align=center] at (6.4,-.02) {divide by $((s-1)!!)^{3}$};
\end{tikzpicture}

\caption{A private clique probe for each layer.
Panel (a) shows the ordinary $N=2,h=1$ permutation representation
before boundary deletion, with $a_i,b_i,c_i$ in $C_0,C_1,C_2$.
Dashed segments are probes. The inset interchanges the adjacent upper
endpoints $a_1,b_0$, changing $U$ to $V_0$ by adding just the edge
$a_1b_0$. Opposite internal endpoint orders expand each probe to $K_s$.
Panel (b) uses the retained boundary layers of size $p$, suppressing
the corresponding cut restrictions. It shows the normalized polynomial identity: at $s=-1$,
each probe cancels the edges inside its own layer, while both cuts
remain unchanged. Probe-to-layer connections denote complete adjacency;
the right-hand side is an equivalent weighted-count expression.}
\label{fig:permutation-construction}
\end{figure}

\begin{lemma}\label{lem:permutation-chordality}
Every graph $G_s$ is chordal.
\end{lemma}

\begin{proof}
By~\eqref{eq:chordal-peo}, it suffices to exhibit a perfect
elimination ordering. First eliminate all probe
vertices. At any such elimination, the remaining neighbors of a vertex
in $Q_a$ consist of the remaining vertices of $Q_a$ and the retained
vertices of $C_a$. They form a clique, because both parts are cliques
and all edges between them are present.

It remains to eliminate the original layers. Every allowed cut matrix
has row neighborhoods ordered by inclusion. For $U$ these are suffixes;
$V_i$ and $Z_i$ merely make two consecutive row thresholds equal.
The transposed matrices likewise have nested row neighborhoods. These
inclusions are preserved by deleting boundary rows or columns.

In the leftmost remaining layer, eliminate vertices in increasing
inclusion order of their neighborhoods in the next layer. When a
vertex $v$ is eliminated, its remaining neighbors are the other
vertices of its own layer and its neighbors in the next layer. Each
of these two sets is a clique. Every remaining vertex $u$ of the same
layer has a next-layer neighborhood containing that of $v$, so all
edges between the two sets are also present. Thus $v$ is simplicial.
In formulas, if $C$ is the current layer, $D$ the next layer, and
$C'$ the vertices of $C$ remaining after $v$, then
\[
 N_D(v)\subseteq N_D(u)\quad(u\in C'),\qquad
 N_{\rm rem}(v)=C'\cup N_D(v),
\]
so every pair in $N_{\rm rem}(v)$ is adjacent.
Repeating this operation through all layers gives the required
ordering.
\end{proof}

\subsection{Cancellation and the oracle reduction}

For nonnegative even $s$, put
\begin{equation}\label{eq:permutation-probe-polynomial}
 \Psi(s)=\frac{\PM(G_s)}{((s-1)!!)^{2h+1}},
 \qquad (-1)!!=1.
\end{equation}
Lemma~\ref{lem:clique-probe} shows that these are evaluations of a
polynomial of degree at most $n_0/2$. In particular, a probe clique
matched to $2a$ prescribed original vertices contributes, after
normalization,
\[
 \frac{(s)_{2a}(s-2a-1)!!}{(s-1)!!}
 =\prod_{b=0}^{a-1}(s-2b)\qquad(s\ge2a,\ s\text{ even}).
\]
Here $(s)_{2a}$ is a falling factorial. At $s=-1$, the polynomial
on the right has value $(-1)^a(2a-1)!!$, which adds edges of weight $-1$ on the probe's
original neighborhood.

In the present construction that neighborhood is one entire layer.
Consequently every intralayer edge receives an additional weight
$-1$ and cancels, while no cut edge changes. The remaining weighted
graph is exactly the independent-layer transfer graph, with all
remaining edge weights equal to one. Therefore
\begin{equation}\label{eq:permutation-cancellation}
 \Psi(-1)=Z.
\end{equation}
There is no sign factor in this identity.

The values $s=0,2,\ldots,n_0$ give $n_0/2+1=2ph+1$ samples, sufficient
to recover~\eqref{eq:permutation-cancellation} by interpolation.
All $2h+1$ probe classes use the same parameter. Each query has at most
\begin{equation}\label{eq:permutation-query-size}
 n_0+(2h+1)n_0=8ph(h+1)
\end{equation}
vertices, and its permutation representation is constructed explicitly
by the orders above. Integer ranks in the two orders supply endpoint
coordinates of magnitude at most the query size.

\begin{proof}[Proof of Theorem~\ref{thm:chordal-permutation}]
The construction and interpolation evaluate the paired-transfer value
in time polynomial in $N+h$, using only simple unweighted chordal
permutation graphs. To check the bit complexity, a graph on $v$ vertices
has at most $(v-1)!!$ perfect matchings, whose bit length is
$O(v\log v)$. The normalizing double factorials and the rational
coefficients for interpolation at $-1$ have polynomial bit length
as well. Thus all arithmetic can be performed exactly in polynomial
time.

We have proved $\mathrm{PairEval}\le_T
\#\mathrm{PM}(\mathrm{ChordalPermutation})$.
Together with Theorem~\ref{thm:word-hard}, this proves
$\sharpP$-hardness. A perfect matching has a canonical edge-set encoding
that can be verified in polynomial time, giving membership in
$\sharpP$. The permutation representation is supplied explicitly
by the construction.
\end{proof}

Since every chordal permutation graph is a permutation graph,
$\#\mathrm{PM}(\mathrm{Permutation})$ is also $\sharpP$-complete.

\section{Exact counting in quasi-chain graphs}\label{sec:qchains}

We use the definition of quasi-chain graphs introduced by
Dyer and M\"uller~\cite{DM19}. A \emph{chain graph} is a bipartite
graph whose neighborhoods in either part are linearly ordered by
inclusion. Let $\mathsf{Chains}$ denote the class of graphs whose
connected components are chain graphs. For a vertex bipartition
$V(G)=L\mathbin{\dot\cup}R$, write $G[L:R]$ for the spanning
bipartite graph containing exactly the edges of $G$ between $L$ and $R$.
Then
\[
 G\in\mathsf{QChains}
 \quad\Longleftrightarrow\quad
 G[L:R]\in\mathsf{Chains}
 \text{ for every vertex bipartition }L,R.
\]
Thus the underlying bipartite class is closed under disjoint union;
it is not the class consisting of single chain graphs.

A graph on five vertices is a \emph{pre-$P_5$} if some bipartition
of its vertices has cut graph $P_5$. Equivalently, after labeling
the path $1,2,3,4,5$ in order, any subset of the edges
$13,15,35,24$ may be added, while $14$ and $25$ must remain absent.
Dyer and M\"uller~\cite[Section~3.4]{DM19} show that
$\mathsf{QChains}$ is precisely the class of graphs containing no
induced pre-$P_5$.

\begin{lemma}\label{lem:qchains-dh}
$\mathsf{QChains}\subsetneq\mathsf{DH}$.
\end{lemma}

\begin{proof}
The forbidden induced subgraphs for distance-hereditary graphs,
as defined in~\eqref{eq:distance-hereditary-definition}, are the
house, gem, domino, and chordless cycles of length at least
five~\cite{BandeltMulder86}.
We verify that a graph in $\mathsf{QChains}$ excludes each of them.

Start with the path $1-2-3-4-5$. Adding the edge $15$ gives $C_5$.
Adding $13,15$ gives a house: its quadrangle is $1-3-4-5-1$ and
its roof is vertex $2$. Adding $13,15,35$ gives a gem: vertex $3$
is universal and the other vertices induce the path $2-1-5-4$.
All the added edges lie within $\{1,3,5\}$, so each of these
graphs, as well as $P_5$ itself, is a pre-$P_5$.

Every chordless cycle of length at least six contains an induced
$P_5$. To handle the domino, write its two quadrangles as
$a-b-c-d-a$ and $c-e-f-d-c$, sharing the edge $cd$.
The vertices $b,a,d,f,e$, in that order, induce a $P_5$.
Thus a graph containing no induced pre-$P_5$ excludes every
distance-hereditary obstruction. This proves the inclusion.
Finally, $P_5$ is a tree and hence is distance-hereditary, but it
does not belong to $\mathsf{QChains}$.
\end{proof}

The common pre-$P_5$ obstruction behind this argument is shown
in Figure~\ref{fig:qchains-obstructions}.
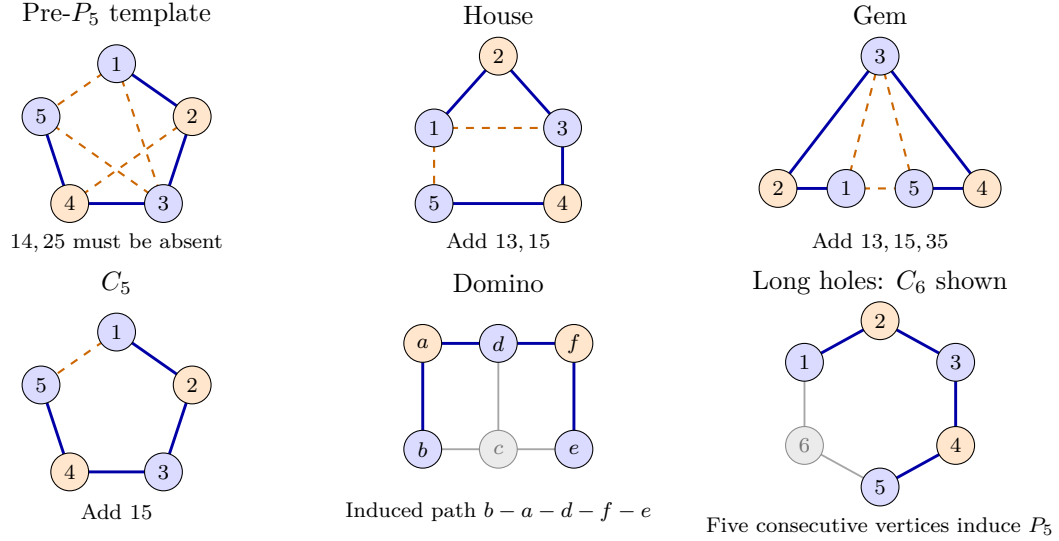
\begin{figure}[htbp]
\centering
\begin{tikzpicture}[x=1cm,y=1cm,font=\small,
  qp/.style={circle,draw=black,fill=blue!14,inner sep=0pt,minimum size=5mm,font=\scriptsize},
  qe/.style={qp,fill=orange!20},
  qg/.style={qp,fill=black!8,draw=black!45,text=black!55},
  path edge/.style={draw=blue!65!black,line width=1.1pt},
  added edge/.style={draw=orange!80!black,dashed,line width=.8pt},
  other edge/.style={draw=black!35,line width=.7pt}]
\path[use as bounding box] (-2.3,-5.55) rectangle (12.55,1.8);
% The template records precisely the six forced and four optional pairs.
\begin{scope}
  \node at (0,1.5) {Pre-$P_5$ template};
  \coordinate (a1) at (0,.85); \coordinate (a2) at (1,.15);
  \coordinate (a3) at (.62,-1); \coordinate (a4) at (-.62,-1);
  \coordinate (a5) at (-1,.15);
  \draw[path edge] (a1)--(a2)--(a3)--(a4)--(a5);
  \draw[added edge] (a1)--(a3) (a1)--(a5) (a3)--(a5) (a2)--(a4);
  \foreach \i in {1,3,5} \node[qp] at (a\i) {\i};
  \foreach \i in {2,4} \node[qe] at (a\i) {\i};
  \node[font=\scriptsize] at (0,-1.52) {$14,25$ must be absent};
\end{scope}
\begin{scope}[xshift=5.05cm]
  \node at (0,1.5) {House};
  \coordinate (b1) at (-.85,0); \coordinate (b2) at (0,.95);
  \coordinate (b3) at (.85,0); \coordinate (b4) at (.85,-1);
  \coordinate (b5) at (-.85,-1);
  \draw[path edge] (b1)--(b2)--(b3)--(b4)--(b5);
  \draw[added edge] (b1)--(b3) (b1)--(b5);
  \foreach \i in {1,3,5} \node[qp] at (b\i) {\i};
  \foreach \i in {2,4} \node[qe] at (b\i) {\i};
  \node[font=\scriptsize] at (0,-1.52) {Add $13,15$};
\end{scope}
\begin{scope}[xshift=10.1cm]
  \node at (0,1.5) {Gem};
  \coordinate (c1) at (-.45,-.8); \coordinate (c2) at (-1.35,-.8);
  \coordinate (c3) at (0,.95); \coordinate (c4) at (1.35,-.8);
  \coordinate (c5) at (.45,-.8);
  \draw[path edge] (c1)--(c2)--(c3)--(c4)--(c5);
  \draw[added edge] (c1)--(c3) (c1)--(c5) (c3)--(c5);
  \foreach \i in {1,3,5} \node[qp] at (c\i) {\i};
  \foreach \i in {2,4} \node[qe] at (c\i) {\i};
  \node[font=\scriptsize] at (0,-1.52) {Add $13,15,35$};
\end{scope}
\begin{scope}[yshift=-3.5cm]
  \node at (0,1.45) {$C_5$};
  \coordinate (d1) at (0,.8); \coordinate (d2) at (1,.1);
  \coordinate (d3) at (.62,-1.05); \coordinate (d4) at (-.62,-1.05);
  \coordinate (d5) at (-1,.1);
  \draw[path edge] (d1)--(d2)--(d3)--(d4)--(d5);
  \draw[added edge] (d1)--(d5);
  \foreach \i in {1,3,5} \node[qp] at (d\i) {\i};
  \foreach \i in {2,4} \node[qe] at (d\i) {\i};
  \node[font=\scriptsize] at (0,-1.57) {Add $15$};
\end{scope}
\begin{scope}[xshift=5.05cm,yshift=-3.5cm]
  \node at (0,1.45) {Domino};
  \coordinate (ea) at (-1,.65); \coordinate (eb) at (-1,-.75);
  \coordinate (ec) at (0,-.75); \coordinate (ed) at (0,.65);
  \coordinate (ee) at (1,-.75); \coordinate (ef) at (1,.65);
  \draw[other edge] (eb)--(ec)--(ed) (ec)--(ee);
  \draw[path edge] (eb)--(ea)--(ed)--(ef)--(ee);
  \node[qp] at (eb) {$b$}; \node[qe] at (ea) {$a$};
  \node[qp] at (ed) {$d$}; \node[qe] at (ef) {$f$};
  \node[qp] at (ee) {$e$}; \node[qg] at (ec) {$c$};
  \node[font=\scriptsize] at (0,-1.57) {Induced path $b-a-d-f-e$};
\end{scope}
\begin{scope}[xshift=10.1cm,yshift=-3.5cm]
  \node at (0,1.45) {Long holes: $C_6$ shown};
  \coordinate (f1) at (-1,.4); \coordinate (f2) at (0,.95);
  \coordinate (f3) at (1,.4); \coordinate (f4) at (1,-.7);
  \coordinate (f5) at (0,-1.25); \coordinate (f6) at (-1,-.7);
  \draw[other edge] (f5)--(f6)--(f1);
  \draw[path edge] (f1)--(f2)--(f3)--(f4)--(f5);
  \foreach \i in {1,3,5} \node[qp] at (f\i) {\i};
  \foreach \i in {2,4} \node[qe] at (f\i) {\i};
  \node[qg] at (f6) {6};
  \node[font=\scriptsize] at (0,-1.77) {Five consecutive vertices induce $P_5$};
\end{scope}
\end{tikzpicture}
\caption{Why quasi-chain graphs are distance-hereditary.
In the template, blue solid edges are mandatory and orange dashed
edges are optional; the vertex fills give the cut with parts
$\{1,3,5\}$ and $\{2,4\}$. In the house, gem, and $C_5$ panels,
every displayed edge is present, including the dashed ones, and
the same cut is $P_5$. The blue paths in the last two panels are
induced $P_5$'s; every chordless cycle of length at least six has
such a path. Thus every distance-hereditary obstruction contains
an induced pre-$P_5$.}
\label{fig:qchains-obstructions}
\end{figure}

Lemma~\ref{lem:qchains-dh} already gives a polynomial-time counting
algorithm through established clique-width results.
Golumbic and Rotics~\cite{GolumbicRotics00} show that every
distance-hereditary graph has clique-width at most three and that
a corresponding $3$-expression can be constructed in linear time.
Makowsky, Rotics, Averbouch, and Godlin~\cite{MRAG06} give a
polynomial-time algorithm for the matching polynomial on graphs
of any fixed clique-width. In its generating form
$\sum_k m_k(G)z^k$, where $m_k(G)$ counts matchings of size $k$,
the coefficient of $z^{|V(G)|/2}$ is $\PM(G)$ when $|V(G)|$ is even.
Curticapean and Marx~\cite[Theorem~1.3]{CurticapeanMarx16} give the
more specific bound $O(n^{k+1})$ for counting perfect matchings
from a $k$-expression. With $k=3$, their theorem already yields
an $O(n^4)$ bound for distance-hereditary graphs.
The direct algorithm below uses $O(n^2)$ arithmetic operations
on integers of $O(n\log n)$ bits.

\begin{theorem}\label{thm:qchains-fp}
For $G\in\mathsf{DH}$ with $n$ vertices and $m$ edges,
$\PM(G)$ can be computed using $O(n^2)$ arithmetic operations on
integers of $O(n\log n)$ bits, after $O(n+m)$ graph processing.
\end{theorem}

\subsection{Pruning and the boundary of a bag}

Two vertices are \emph{twins} if they have the same neighborhood
outside the pair. They are \emph{true twins} when adjacent and
\emph{false twins} otherwise. A pendant vertex has degree one.
The pruning characterization of distance-hereditary graphs states
that every connected induced subgraph with at least two vertices
has a pendant vertex or a pair of twins~\cite{BandeltMulder86}.
Consequently, repeated pendant or twin deletions reduce each
connected component to one vertex.
A pruning sequence records each deleted vertex, its surviving
pendant neighbor or twin, and the type of the deletion.
Such a sequence can be constructed in $O(n+m)$ time and space
by Uehara and Uno~\cite[Theorem~21]{UeharaUno06}.
We process the sequence in deletion order and finalize the last
representative of each component as an isolated vertex.

The algorithm maintains an induced subgraph $H$ of $G$ on the
surviving representatives. The bag $X_x$ is the set of all original
vertices currently represented by $x$: it consists of $x$ and the
vertices in all bags previously absorbed into it.
Initially $H=G$ and $X_x=\{x\}$.
For a pendant or twin deletion of $v$, let $u$ be its surviving
neighbor or twin, respectively. Update
\begin{equation}\label{eq:qchains-bag-merge}
 X_u^{\mathrm{new}}=X_u^{\mathrm{old}}\cup X_v^{\mathrm{old}},
 \qquad H^{\mathrm{new}}=H^{\mathrm{old}}-v,
\end{equation}
leaving all other bags unchanged. The deleted representative's
entire bag is now represented by $u$, and its matching information
is incorporated into the state array of $u$. The original graph
$G$ remains the graph whose perfect matchings are counted.
The bags partition the original vertices not yet finalized;
an isolated representative's whole bag is finalized by the rule
below.

Each bag has a nonempty active set $T_x\subseteq X_x$, initially
$T_x=\{x\}$.
The structural invariant is that, for distinct representatives
$x,y$, the original edges between their bags are exactly
\begin{equation}\label{eq:qchains-structure}
 E_G(X_x,X_y)=
 \begin{cases}
   \{ab:a\in T_x,\ b\in T_y\},&xy\in E(H),\\
   \varnothing,&xy\notin E(H).
 \end{cases}
\end{equation}
In particular, vertices in $X_x\setminus T_x$ have no neighbors
outside $X_x$.

For a matching $M$ in $G[X_x]$, let $U_x(M)=X_x\setminus V(M)$
be its uncovered vertices. Store
\begin{equation}\label{eq:qchains-state}
 f_x(k)=\#\{M\text{ a matching in }G[X_x]:
           U_x(M)\subseteq T_x,\ |U_x(M)|=k\}.
\end{equation}
Uncovered active vertices may be matched outside the bag later.
Only their number must be recorded, because they all have the
same neighbors outside the bag. Initially $f_x(k)=\mathbf1_{k=1}$.

Figure~\ref{fig:qchains-merges} illustrates the boundary state and
the choices counted in the three updates below.
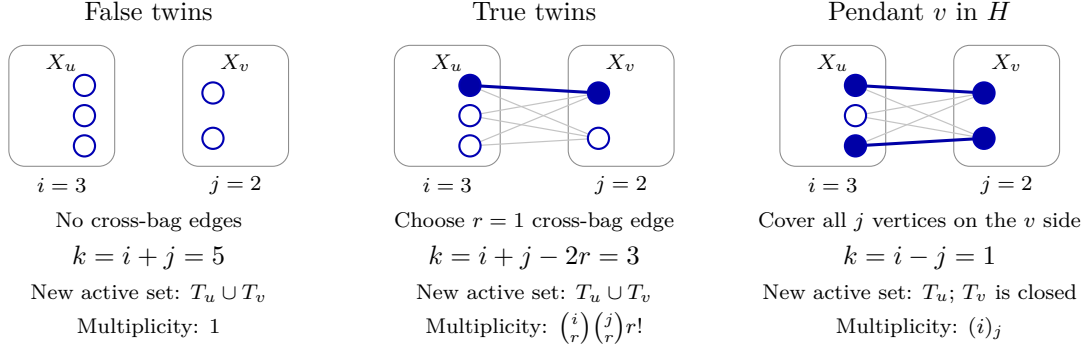
\begin{figure}[htbp]
\centering
\begin{tikzpicture}[x=1cm,y=1cm,font=\small,
  free/.style={circle,draw=blue!70!black,fill=white,line width=.8pt,inner sep=0pt,minimum size=2.8mm},
  used/.style={free,fill=blue!65!black},
  interior/.style={circle,draw=black!55,fill=black!35,inner sep=0pt,minimum size=2.8mm},
  bag/.style={draw=black!45,rounded corners=6pt},
  active/.style={draw=blue!65!black,dashed,rounded corners=5pt},
  possible/.style={draw=black!23,line width=.45pt},
  matching/.style={draw=blue!65!black,line width=1.15pt}]
\path[use as bounding box] (-.3,-5.55) rectangle (14.8,2.0);
% One bag explains the state before the three merges are illustrated.
\draw[bag] (0,0) rectangle (4,1.8);
\node[anchor=west] at (.1,1.53) {$X_x$};
\draw[active] (2.75,.2) rectangle (3.7,1.6);
\node[font=\scriptsize] at (3.24,1.38) {$T_x$};
\coordinate (a1) at (3.22,1.07); \coordinate (a2) at (3.22,.73);
\coordinate (a3) at (3.22,.39);
\coordinate (n1) at (.6,.52); \coordinate (n2) at (1.35,.52);
\coordinate (n3) at (2.15,.39);
\coordinate (out) at (5.25,.9);
\draw[possible] (a1)--(out) (a2)--(out) (a3)--(out);
\draw[matching] (n1)--(n2) (n3)--(a3);
\node[free] at (a1) {}; \node[free] at (a2) {};
\node[used] at (a3) {};
\foreach \i in {1,2,3} \node[interior] at (n\i) {};
\node[circle,draw=black!50,fill=black!8,inner sep=1.7pt,font=\scriptsize] at (out) {$z$};
\node[font=\scriptsize,align=center] at (5.23,.23) {outside\\the bag};
\node[font=\scriptsize,anchor=north] at (2,-.15) {Example: a matching counted by $f_x(2)$};
\node[anchor=west,align=left,text width=7.75cm,font=\small] at (6.5,.97)
 {The state $f_x(k)$ records $k$ uncovered active vertices.\\
 All vertices in $X_x\setminus T_x$ are covered.\\
 Any outside vertex sees all of $T_x$ or none of it.};
\draw[black!18] (0,-.65)--(14.55,-.65);
% The remaining panels show only the uncovered active vertices.
\foreach \off/\kind in {0/{False twins},5.1/{True twins},10.2/{Pendant $v$ in $H$}} {
  \begin{scope}[xshift=\off cm,yshift=-2.4cm]
    \node at (2.1,1.3) {\kind};
    \draw[bag] (.25,-.75) rectangle (1.65,.85);
    \draw[bag] (2.55,-.75) rectangle (3.95,.85);
    \node[font=\scriptsize] at (.95,.62) {$X_u$};
    \node[font=\scriptsize] at (3.25,.62) {$X_v$};
    \node[font=\scriptsize] at (.95,-1) {$i=3$};
    \node[font=\scriptsize] at (3.25,-1) {$j=2$};
  \end{scope}
}
\begin{scope}[yshift=-2.4cm]
  \foreach \y in {.32,-.08,-.48} \node[free] at (1.25,\y) {};
  \foreach \y in {.22,-.38} \node[free] at (2.95,\y) {};
  \node[font=\scriptsize,align=center] at (2.1,-1.48) {No cross-bag edges};
  \node at (2.1,-1.95) {$k=i+j=5$};
  \node[font=\scriptsize] at (2.1,-2.43) {New active set: $T_u\cup T_v$};
  \node[font=\scriptsize] at (2.1,-2.89) {Multiplicity: $1$};
\end{scope}
\begin{scope}[xshift=5.1cm,yshift=-2.4cm]
  \foreach \y in {.32,-.08,-.48} {
    \foreach \z in {.22,-.38} \draw[possible] (1.25,\y)--(2.95,\z);
  }
  \draw[matching] (1.25,.32)--(2.95,.22);
  \node[used] at (1.25,.32) {}; \node[used] at (2.95,.22) {};
  \foreach \y in {-.08,-.48} \node[free] at (1.25,\y) {};
  \node[free] at (2.95,-.38) {};
  \node[font=\scriptsize] at (2.1,-1.48) {Choose $r=1$ cross-bag edge};
  \node at (2.1,-1.95) {$k=i+j-2r=3$};
  \node[font=\scriptsize] at (2.1,-2.43) {New active set: $T_u\cup T_v$};
  \node[font=\scriptsize] at (2.1,-2.89) {Multiplicity: $\binom{i}{r}\binom{j}{r}r!$};
\end{scope}
\begin{scope}[xshift=10.2cm,yshift=-2.4cm]
  \foreach \y in {.32,-.08,-.48} {
    \foreach \z in {.22,-.38} \draw[possible] (1.25,\y)--(2.95,\z);
  }
  \draw[matching] (1.25,.32)--(2.95,.22) (1.25,-.48)--(2.95,-.38);
  \foreach \y in {.32,-.48} \node[used] at (1.25,\y) {};
  \foreach \y in {.22,-.38} \node[used] at (2.95,\y) {};
  \node[free] at (1.25,-.08) {};
  \node[font=\scriptsize] at (2.1,-1.48) {Cover all $j$ vertices on the $v$ side};
  \node at (2.1,-1.95) {$k=i-j=1$};
  \node[font=\scriptsize] at (2.1,-2.43) {New active set: $T_u$; $T_v$ is closed};
  \node[font=\scriptsize] at (2.1,-2.89) {Multiplicity: $(i)_j$};
\end{scope}
\end{tikzpicture}
\caption{The boundary state and the three bag merges.
Hollow circles are uncovered active vertices; blue solid edges
are chosen matching edges. The merge panels show only vertices
uncovered in the two old states, and gray lines show all possible
cross-bag edges. Active vertices share their neighborhood
outside their bag; they need not be twins inside it.
The pendant relation is in the current representative graph $H$.
In that case the whole $v$ side loses its external boundary and
must be covered during the merge.}
\label{fig:qchains-merges}
\end{figure}

\subsection{The three merge operations}

Each merge uses~\eqref{eq:qchains-bag-merge}; the three types
differ in their active-set and state-array updates. The formulas
below use the two old arrays; entries outside their ranges are zero.

If $u,v$ are false twins, their bags have no edges between them.
The new bag is $X_u\cup X_v$, its active set is $T_u\cup T_v$,
and its array is
\begin{equation}\label{eq:qchains-false}
 f_{\mathrm{new}}(k)
 =\sum_{i+j=k}f_u(i)f_v(j).
\end{equation}

If $u,v$ are true twins, the new active set is again
$T_u\cup T_v$. All edges between these two sets are present.
For old states $i,j$, adding $r$ cross-bag matching edges requires
choosing $r$ uncovered vertices on each side and a bijection
between the chosen sets. Therefore
\begin{equation}\label{eq:qchains-true}
 f_{\mathrm{new}}(k)
 =\sum_{\substack{i,j\ge0,\ 0\le r\le\min(i,j)\\
                   i+j-2r=k}}
   f_u(i)f_v(j)\binom{i}{r}\binom{j}{r}r!.
\end{equation}

Finally, suppose $v$ is pendant in $H$, with unique neighbor $u$.
The cross-bag edges form the same complete bipartite block, but
the new active set is only $T_u$. Vertices in $T_v$ have no
neighbors outside the merged bag. Hence every one of the $j$
uncovered vertices on the $v$ side must be matched to a distinct
uncovered vertex on the $u$ side. Writing
$(i)_j=i(i-1)\cdots(i-j+1)$, with $(i)_0=1$, gives
\begin{equation}\label{eq:qchains-pendant}
 f_{\mathrm{new}}(k)
 =\sum_{\substack{i\ge j\ge0\\i-j=k}}
   f_u(i)f_v(j)\binom{i}{j}j!
 =\sum_{\substack{i\ge j\ge0\\i-j=k}}
   f_u(i)f_v(j)(i)_j.
\end{equation}

If a representative $x$ is isolated in $H$, its bag has no
external edges. Multiply the running answer by $f_x(0)$ and
remove $x$. The running answer starts at one, so this convention
also handles the empty graph.

\begin{lemma}\label{lem:qchains-bag-invariant}
The updates~\eqref{eq:qchains-false}--\eqref{eq:qchains-pendant}
preserve~\eqref{eq:qchains-structure} and~\eqref{eq:qchains-state}.
At termination, let $\mathcal F$ be the set of finalized bags. Then
\[
 \PM(G)=\prod_{X_x\in\mathcal F}f_x(0).
\]
\end{lemma}

\begin{proof}
Both invariants hold initially. For a twin merge, the two
representatives have identical adjacency to every other
representative. Thus replacing their active sets by their union
preserves all cross-bag blocks. The block between the merged
bags is empty for false twins and complete between their active
sets for true twins.

For a pendant merge, $v$ has no neighbor other than $u$ in $H$.
Consequently $T_v$ has no neighbors outside $X_u\cup X_v$,
whereas $T_u$ retains exactly the external neighbors of the
merged bag. Taking $T_u$ as its active set therefore preserves
the structural invariant. In all three cases, the new
representative graph is exactly $H-v$.

A matching in the merged bag decomposes uniquely into its
restrictions to the two old bags and its cross-bag edges.
The nonactive vertices of either old bag cannot participate in
cross-bag edges, so both restrictions satisfy the old state
conditions. For false twins there are no cross-bag edges, giving
\eqref{eq:qchains-false}. For true twins the choice of endpoints
and their bijection gives precisely the factor in
\eqref{eq:qchains-true}. In the pendant case, all uncovered
vertices of $T_v$ must be covered by cross-bag edges, giving
\eqref{eq:qchains-pendant}. These decompositions are reversible
and unique, establishing the counting invariant.

An isolated representative has no edges between its bag and any
other bag, so $f_x(0)$ is the number of perfect matchings on this
closed set of vertices. Its contribution is independent of all
remaining bags. Multiplying these contributions completes the
count for $G$.
\end{proof}

\subsection{A quadratic arithmetic implementation}

Write $a=|X_u|$, $b=|X_v|$, and encode the old arrays by ordinary
generating polynomials
\[
 F(x)=\sum_i f_u(i)x^i,\qquad
 G(x)=\sum_j f_v(j)x^j,
 \qquad \deg F\le a,\quad \deg G\le b.
\]
For false twins, the new polynomial is $FG$, computable by direct
convolution in $O(ab)$ arithmetic operations. For a pendant
merge, enumerate the pairs $i\ge j$ in~\eqref{eq:qchains-pendant},
generating the factors within each row by
\[
 (i)_0=1,\qquad (i)_{j+1}=(i-j)(i)_j.
\]
This also takes $O(ab)$ arithmetic operations.

The true-twin sum~\eqref{eq:qchains-true} is the polynomial
\begin{equation}\label{eq:qchains-star}
 F\star G:=\sum_{r\ge0}\frac{F^{(r)}G^{(r)}}{r!},
 \qquad
 \frac{(i)_r(j)_r}{r!}=\binom ir\binom jr r!.
\end{equation}
The displayed identity identifies each summand with its
cross-bag matching multiplicity. The implementation below
evaluates this operation using integer arithmetic only.

For $j\ge0$, define
\begin{equation}\label{eq:qchains-Q-definition}
 Q_j(x)=F\star x^j
       =\sum_{r=0}^j\binom jr x^{j-r}F^{(r)}(x).
\end{equation}
Then $F\star G=\sum_j f_v(j)Q_j$.
\begin{lemma}\label{lem:qchains-Q-recurrence}
The polynomials $Q_j$ satisfy
\begin{equation}\label{eq:qchains-Q-recurrence}
 Q_0=F,\qquad Q_1=xF+F',\qquad
 Q_{j+1}=xQ_j+Q_j'-jQ_{j-1}\quad(j\ge1).
\end{equation}
For two bags of sizes $a,b\ge1$, $F\star G$ can be evaluated
in $O(ab)$ integer arithmetic operations.
\end{lemma}

\begin{proof}
Using formal power series, equation~\eqref{eq:qchains-Q-definition}
gives
\[
 \mathcal Q(x,t):=\sum_{j\ge0}Q_j(x)\frac{t^j}{j!}
   =e^{xt}F(x+t),\qquad
 \partial_t\mathcal Q=(x+\partial_x-t)\mathcal Q.
\]
Indeed, $\partial_t(e^{xt}F(x+t))=x\mathcal Q+e^{xt}F'(x+t)$,
whereas $\partial_x\mathcal Q=t\mathcal Q+e^{xt}F'(x+t)$.
Comparing coefficients of $t^j/j!$ proves the recurrence,
including $Q_1=xF+F'$ from the constant coefficient.

Since $F\star G=G\star F$, choose the polynomial operands so
that $a\ge b$. This exchanges the two input arrays when needed,
without changing which graph representative survives the merge.
After this exchange, write $G=\sum_j g_jx^j$.
Compute $Q_0,\ldots,Q_b$ successively and accumulate
$\sum_{j=0}^b g_jQ_j$, with $g_j=0$ beyond its range.
The bounds
\[
 \deg Q_j\le a+j\le a+b,\qquad
 (b+1)(a+b+1)=O(ab)\quad(a\ge b\ge1)
\]
show that the recurrence and accumulation take $O(ab)$ operations.
Only the two most recent $Q$ arrays and the output accumulator
are needed. Differentiation multiplies coefficients by integers;
the recurrence uses no divisions.
\end{proof}

\begin{proof}[Proof of Theorem~\ref{thm:qchains-fp}]
For $n\le1$, return one if $n=0$ and zero otherwise. Assume $n\ge2$.
Construct a pruning sequence in $O(n+m)$ graph processing and
apply the bag updates in its order. Store each bag size and its
coefficient array; explicit active vertex sets are needed only
for the invariant, not for executing the updates. Each deletion
performs one merge or finalizes one isolated bag. The pruning
characterization guarantees termination, and
Lemma~\ref{lem:qchains-bag-invariant} gives the final value.

The merges form a binary forest with the original vertices as
leaves. At a merge with child bags of sizes $a,b$, charge its
$ab$ pairs having one endpoint in each child. Each pair is
charged at most once, at the first merge that places its
endpoints in the same bag. Consequently,
\begin{equation}\label{eq:qchains-pair-charge}
 \sum_{\text{merges}}ab
 =\sum_{X\in\mathcal F}\binom{|X|}{2}
 \le\binom n2.
\end{equation}
Every merge costs $O(ab)$ arithmetic operations by the direct
false-twin and pendant updates and
Lemma~\ref{lem:qchains-Q-recurrence}. Initialization and
finalization cost $O(n)$ operations. Thus the total is $O(n^2)$.

It remains to bound all intermediate integers, including those
in the recurrence with subtraction. Any graph on at most $n$
vertices has at most $(n+1)^n$ matchings: encode a matching by
assigning to each vertex its partner or an unmatched marker.
Hence this bounds every state entry. The coefficient of $x^k$
in $Q_j$ has the same type of counting interpretation: add $j$
independent new active vertices to the first bag, join each to
all its active vertices, require the old nonactive vertices to
be covered, and count matchings with $k$ uncovered active vertices.
Choosing $r$ cross-bag edges gives exactly
\eqref{eq:qchains-Q-definition}. This auxiliary graph has
$a+j\le a+b\le n$ vertices, so
\[
 0\le [x^k]Q_j\le(n+1)^n.
\]
The derivative and scalar multiplication in
\eqref{eq:qchains-Q-recurrence} multiply coefficients by at most
$n$. Thus the magnitudes of its intermediate sums and differences
are at most $(2n+1)(n+1)^n$. Multiplication by $g_j$ and
accumulation over at most $n+1$ indices still give
$n^{O(n)}$ bounds. The other updates use falling factorials
at most $n!$, products of state entries, and at most $(n+1)^2$
summands. The product of finalized contributions counts
matchings on disjoint sets of at most $n$ vertices. All these
integers therefore have $O(n\log n)$ bits, proving the stated
arithmetic bound and polynomial bit complexity.
\end{proof}

By Lemma~\ref{lem:qchains-dh}, this algorithm applies to every
graph in $\mathsf{QChains}$. Hence
$\#\mathrm{PM}(\mathsf{QChains})\in\mathrm{FP}$.

\FloatBarrier


\begin{thebibliography}{99}
\small
\setlength{\itemsep}{3pt}
\setlength{\parskip}{0pt}
\interlinepenalty=10000

\bibitem{AASV21}
Yeganeh Alimohammadi, Nima Anari, Kirankumar Shiragur, and Thuy-Duong Vuong.
\newblock Fractionally log-concave and sector-stable polynomials:
counting planar matchings and more.
\newblock In \emph{Proceedings of the 53rd Annual ACM SIGACT Symposium
on Theory of Computing (STOC 2021)}, pages 433--446. ACM, 2021.
\newblock \href{https://doi.org/10.1145/3406325.3451123}{doi:10.1145/3406325.3451123}.

\bibitem{ADRS23}
Yeganeh Alimohammadi, Persi Diaconis, Mohammad Roghani, and Amin Saberi.
\newblock Sequential importance sampling for estimating expectations over the space
of perfect matchings.
\newblock \emph{The Annals of Applied Probability}, 33(2):999--1033, 2023.
\newblock \href{https://doi.org/10.1214/22-AAP1834}{doi:10.1214/22-AAP1834}.

\bibitem{Backens21}
Miriam Backens.
\newblock A full dichotomy for Holant$^c$, inspired by quantum computation.
\newblock \emph{SIAM Journal on Computing}, 50(6):1739--1799, 2021.
\newblock \href{https://doi.org/10.1137/20M1311557}{doi:10.1137/20M1311557}.

\bibitem{BandeltMulder86}
Hans-J\"urgen Bandelt and Henry Martyn Mulder.
\newblock Distance-hereditary graphs.
\newblock \emph{Journal of Combinatorial Theory, Series B},
41(2):182--208, 1986.
\newblock \href{https://doi.org/10.1016/0095-8956(86)90043-2}{doi:10.1016/0095-8956(86)90043-2}.

\bibitem{Barvinok17}
Alexander Barvinok.
\newblock Approximating permanents and hafnians.
\newblock \emph{Discrete Analysis}, 2017:2, 34 pp., 2017.
\newblock \href{https://arxiv.org/abs/1601.07518}{arXiv:1601.07518}.

\bibitem{BSVV08}
Ivona Bez\'akov\'a, Daniel \v{S}tefankovi\v{c}, Vijay~V. Vazirani,
and Eric Vigoda.
\newblock Accelerating simulated annealing for the permanent and combinatorial
counting problems.
\newblock \emph{SIAM Journal on Computing}, 37(5):1429--1454, 2008.
\newblock \href{https://doi.org/10.1137/050644033}{doi:10.1137/050644033}.

\bibitem{CaiLuXia18}
Jin-Yi Cai, Pinyan Lu, and Mingji Xia.
\newblock Dichotomy for real Holant$^c$ problems.
\newblock In \emph{Proceedings of the Twenty-Ninth Annual ACM-SIAM
Symposium on Discrete Algorithms (SODA 2018)}, pages 1802--1821.
SIAM, 2018.
\newblock \href{https://doi.org/10.1137/1.9781611975031.118}{doi:10.1137/1.9781611975031.118}.

\bibitem{CorneilEtAl95}
Derek~G. Corneil, Hiryoung Kim, Sridhar Natarajan, Stephan Olariu,
and Alan~P. Sprague.
\newblock Simple linear time recognition of unit interval graphs.
\newblock \emph{Information Processing Letters}, 55(2):99--104, 1995.
\newblock \href{https://doi.org/10.1016/0020-0190(95)00046-F}{doi:10.1016/0020-0190(95)00046-F}.

\bibitem{CurticapeanMarx16}
Radu Curticapean and D\'aniel Marx.
\newblock Tight conditional lower bounds for counting perfect matchings
on graphs of bounded treewidth, cliquewidth, and genus.
\newblock In \emph{Proceedings of the Twenty-Seventh Annual ACM-SIAM
Symposium on Discrete Algorithms (SODA 2016)}, pages 1650--1669.
SIAM, 2016.
\newblock \href{https://doi.org/10.1137/1.9781611974331.ch113}{doi:10.1137/1.9781611974331.ch113}.

\bibitem{DGH01}
Persi Diaconis, Ronald Graham, and Susan~P. Holmes.
\newblock Statistical problems involving permutations with restricted positions.
\newblock In \emph{State of the Art in Probability and Statistics}, volume~36 of
\newblock \emph{IMS Lecture Notes--Monograph Series}, pages 195--222.
Institute of Mathematical Statistics, 2001.
\newblock \href{https://doi.org/10.1214/lnms/1215090070}{doi:10.1214/lnms/1215090070}.

\bibitem{DK21}
Persi Diaconis and Brett Kolesnik.
\newblock Randomized sequential importance sampling for estimating the number of
perfect matchings in bipartite graphs.
\newblock \emph{Advances in Applied Mathematics}, 131, Article~102247, 2021.
\newblock \href{https://doi.org/10.1016/j.aam.2021.102247}{doi:10.1016/j.aam.2021.102247}.

\bibitem{DJM17}
Martin Dyer, Mark Jerrum, and Haiko M\"uller.
\newblock On the switch Markov chain for perfect matchings.
\newblock \emph{Journal of the ACM}, 64(2), Article~12, 2017.
\newblock \href{https://doi.org/10.1145/2822322}{doi:10.1145/2822322}.

\bibitem{DM15}
Martin Dyer and Haiko M\"uller.
\newblock Graph classes and the switch Markov chain for matchings.
\newblock \emph{Annales de la Facult\'e des Sciences de Toulouse.
Math\'ematiques}, 24(4):885--933, 2015.
\newblock \href{https://doi.org/10.5802/afst.1469}{doi:10.5802/afst.1469}.

\bibitem{DM19}
Martin Dyer and Haiko M\"uller.
\newblock Counting perfect matchings and the switch chain.
\newblock \emph{SIAM Journal on Discrete Mathematics},
33(3):1146--1174, 2019.
\newblock \href{https://doi.org/10.1137/18M1172910}{doi:10.1137/18M1172910}.

\bibitem{DMQ19}
Martin Dyer and Haiko M\"uller.
\newblock Quasimonotone graphs.
\newblock \emph{Discrete Applied Mathematics}, 271:25--48, 2019.
\newblock \href{https://doi.org/10.1016/j.dam.2019.08.006}{doi:10.1016/j.dam.2019.08.006}.

\bibitem{ENO22}
Farzam Ebrahimnejad, Ansh Nagda, and Shayan Oveis Gharan.
\newblock Counting and sampling perfect matchings in regular expanding
non-bipartite graphs.
\newblock In \emph{13th Innovations in Theoretical Computer Science
Conference (ITCS 2022)}, volume 215 of \emph{LIPIcs},
pages 61:1--61:12. Schloss Dagstuhl, 2022.
\newblock \href{https://doi.org/10.4230/LIPIcs.ITCS.2022.61}{doi:10.4230/LIPIcs.ITCS.2022.61}.

\bibitem{ElMaaloulyWang22}
Nicolas El Maalouly and Yanheng Wang.
\newblock Counting perfect matchings in dense graphs is hard.
\newblock Preprint, arXiv:2210.15014, 2022.
\newblock \href{https://doi.org/10.48550/arXiv.2210.15014}{doi:10.48550/arXiv.2210.15014}.

\bibitem{FulkersonGross65}
D.~R. Fulkerson and O.~A. Gross.
\newblock Incidence matrices and interval graphs.
\newblock \emph{Pacific Journal of Mathematics}, 15(3):835--855, 1965.
\newblock \href{https://doi.org/10.2140/pjm.1965.15.835}{doi:10.2140/pjm.1965.15.835}.

\bibitem{GalluccioLoebl99}
Anna Galluccio and Martin Loebl.
\newblock On the theory of Pfaffian orientations. I. Perfect matchings
and permanents.
\newblock \emph{The Electronic Journal of Combinatorics},
6(1):R6, 1999.
\newblock \href{https://doi.org/10.37236/1438}{doi:10.37236/1438}.

\bibitem{GolumbicRotics00}
Martin Charles Golumbic and Udi Rotics.
\newblock On the clique-width of some perfect graph classes.
\newblock \emph{International Journal of Foundations of Computer Science},
11(3):423--443, 2000.
\newblock \href{https://doi.org/10.1142/S0129054100000260}{doi:10.1142/S0129054100000260}.

\bibitem{HuangLu16}
Sangxia Huang and Pinyan Lu.
\newblock A dichotomy for real weighted Holant problems.
\newblock \emph{Computational Complexity}, 25(1):255--304, 2016.
\newblock \href{https://doi.org/10.1007/s00037-015-0118-3}{doi:10.1007/s00037-015-0118-3}.

\bibitem{JS89}
Mark Jerrum and Alistair Sinclair.
\newblock Approximating the permanent.
\newblock \emph{SIAM Journal on Computing}, 18(6):1149--1178, 1989.
\newblock \href{https://doi.org/10.1137/0218077}{doi:10.1137/0218077}.

\bibitem{JSV04}
Mark Jerrum, Alistair Sinclair, and Eric Vigoda.
\newblock A polynomial-time approximation algorithm for the permanent of a matrix
with nonnegative entries.
\newblock \emph{Journal of the ACM}, 51(4):671--697, 2004.
\newblock \href{https://doi.org/10.1145/1008731.1008738}{doi:10.1145/1008731.1008738}.

\bibitem{JVV86}
Mark~R. Jerrum, Leslie~G. Valiant, and Vijay~V. Vazirani.
\newblock Random generation of combinatorial structures from a uniform distribution.
\newblock \emph{Theoretical Computer Science}, 43:169--188, 1986.
\newblock \href{https://doi.org/10.1016/0304-3975(86)90174-X}{doi:10.1016/0304-3975(86)90174-X}.

\bibitem{JozsaMiyake08}
Richard Jozsa and Akimasa Miyake.
\newblock Matchgates and classical simulation of quantum circuits.
\newblock \emph{Proceedings of the Royal Society A}, 464:3089--3106, 2008.
\newblock \href{https://doi.org/10.1098/rspa.2008.0189}{doi:10.1098/rspa.2008.0189}.

\bibitem{Jumadildayev25}
Medet Jumadildayev.
\newblock Duality relations of graph polynomials.
\newblock Preprint, arXiv:2512.15351, 2025.
\newblock \url{https://arxiv.org/abs/2512.15351}.

\bibitem{Kasteleyn61}
P.~W. Kasteleyn.
\newblock The statistics of dimers on a lattice: I. The number of
dimer arrangements on a quadratic lattice.
\newblock \emph{Physica}, 27(12):1209--1225, 1961.
\newblock \href{https://doi.org/10.1016/0031-8914(61)90063-5}{doi:10.1016/0031-8914(61)90063-5}.

\bibitem{KOU11}
Shuji Kijima, Yoshio Okamoto, and Takeaki Uno.
\newblock Dominating set counting in graph classes.
\newblock In Bin Fu and Ding-Zhu Du, editors,
\emph{Computing and Combinatorics (COCOON 2011)},
volume 6842 of \emph{Lecture Notes in Computer Science},
pages 13--24. Springer, 2011.
\newblock \href{https://doi.org/10.1007/978-3-642-22685-4_2}{doi:10.1007/978-3-642-22685-4\_2}.

\bibitem{Li26}
Baitian Li.
\newblock Counting perfect matchings and Hamiltonian cycles faster.
\newblock In \emph{53rd International Colloquium on Automata, Languages,
and Programming (ICALP 2026)}, volume 374 of
\emph{Leibniz International Proceedings in Informatics}, pages 138:1--138:16.
Schloss Dagstuhl--Leibniz-Zentrum f\"ur Informatik, 2026.
\newblock \href{https://doi.org/10.4230/LIPIcs.ICALP.2026.138}{doi:10.4230/LIPIcs.ICALP.2026.138}.

\bibitem{Mader67}
W. Mader.
\newblock Homomorphieeigenschaften und mittlere Kantendichte von Graphen.
\newblock \emph{Mathematische Annalen}, 174:265--268, 1967.
\newblock \href{https://doi.org/10.1007/BF01364272}{doi:10.1007/BF01364272}.

\bibitem{MRAG06}
J.~A. Makowsky, Udi Rotics, Ilya Averbouch, and Benny Godlin.
\newblock Computing graph polynomials on graphs of bounded clique-width.
\newblock In Fedor V. Fomin, editor,
\emph{Graph-Theoretic Concepts in Computer Science (WG 2006)},
volume 4271 of \emph{Lecture Notes in Computer Science},
pages 191--204. Springer, 2006.
\newblock \href{https://doi.org/10.1007/11917496_18}{doi:10.1007/11917496\_18}.

\bibitem{MarkovShi08}
Igor~L. Markov and Yaoyun Shi.
\newblock Simulating quantum computation by contracting tensor networks.
\newblock \emph{SIAM Journal on Computing}, 38(3):963--981, 2008.
\newblock \href{https://doi.org/10.1137/050644756}{doi:10.1137/050644756}.

\bibitem{OUU10}
Yoshio Okamoto, Ryuhei Uehara, and Takeaki Uno.
\newblock Counting the number of matchings in chordal and chordal
bipartite graph classes.
\newblock In Christophe Paul and Michel Habib, editors,
\emph{Graph-Theoretic Concepts in Computer Science (WG 2009)},
volume 5911 of \emph{Lecture Notes in Computer Science},
pages 296--307. Springer, 2010.
\newblock \href{https://doi.org/10.1007/978-3-642-11409-0_26}{doi:10.1007/978-3-642-11409-0\_26}.

\bibitem{PB83}
J.~Scott Provan and Michael~O. Ball.
\newblock The complexity of counting cuts and of computing the
probability that a graph is connected.
\newblock \emph{SIAM Journal on Computing}, 12(4):777--788, 1983.
\newblock \href{https://doi.org/10.1137/0212053}{doi:10.1137/0212053}.

\bibitem{SJ89}
Alistair Sinclair and Mark Jerrum.
\newblock Approximate counting, uniform generation and rapidly mixing
Markov chains.
\newblock \emph{Information and Computation}, 82(1):93--133, 1989.
\newblock \href{https://doi.org/10.1016/0890-5401(89)90067-9}{doi:10.1016/0890-5401(89)90067-9}.

\bibitem{SVW18}
Daniel \v{S}tefankovi\v{c}, Eric Vigoda, and John Wilmes.
\newblock On counting perfect matchings in general graphs.
\newblock In \emph{LATIN 2018: Theoretical Informatics}, volume 10807
of \emph{Lecture Notes in Computer Science}, pages 873--885.
Springer, 2018.
\newblock \href{https://arxiv.org/abs/1712.07504}{arXiv:1712.07504}.

\bibitem{StraubThieraufWagner16}
Simon Straub, Thomas Thierauf, and Fabian Wagner.
\newblock Counting the number of perfect matchings in $K_5$-free graphs.
\newblock \emph{Theory of Computing Systems}, 59(3):416--439, 2016.
\newblock \href{https://doi.org/10.1007/s00224-015-9645-1}{doi:10.1007/s00224-015-9645-1}.

\bibitem{TemperleyFisher61}
H.~N.~V. Temperley and Michael~E. Fisher.
\newblock Dimer problem in statistical mechanics---an exact result.
\newblock \emph{Philosophical Magazine}, 6(68):1061--1063, 1961.
\newblock \href{https://doi.org/10.1080/14786436108243366}{doi:10.1080/14786436108243366}.

\bibitem{Tesler00}
Glenn Tesler.
\newblock Matchings in graphs on non-orientable surfaces.
\newblock \emph{Journal of Combinatorial Theory, Series B},
78(2):198--231, 2000.
\newblock \href{https://doi.org/10.1006/jctb.1999.1941}{doi:10.1006/jctb.1999.1941}.

\bibitem{ThilikosWiederrecht24}
Dimitrios~M. Thilikos and Sebastian Wiederrecht.
\newblock Killing a vortex.
\newblock \emph{Journal of the ACM}, 71(4), Article~27, 2024.
\newblock \href{https://doi.org/10.1145/3664648}{doi:10.1145/3664648}.

\bibitem{UeharaUno06}
Ryuhei Uehara and Takeaki Uno.
\newblock Canonical tree representation of distance hereditary graphs
with applications.
\newblock \emph{IEICE Technical Report}, COMP2005-61, pages 31--36, 2006.
\newblock Full version, March 6, 2006:
\url{https://www.jaist.ac.jp/~uehara/pdf/dh.pdf}.

\bibitem{Valiant79}
Leslie~G. Valiant.
\newblock The complexity of computing the permanent.
\newblock \emph{Theoretical Computer Science}, 8(2):189--201, 1979.
\newblock \href{https://doi.org/10.1016/0304-3975(79)90044-6}{doi:10.1016/0304-3975(79)90044-6}.

\bibitem{Valiant02}
Leslie~G. Valiant.
\newblock Quantum circuits that can be simulated classically in polynomial time.
\newblock \emph{SIAM Journal on Computing}, 31(4):1229--1254, 2002.
\newblock \href{https://doi.org/10.1137/S0097539700377025}{doi:10.1137/S0097539700377025}.

\bibitem{Vazirani89}
Vijay~V. Vazirani.
\newblock NC algorithms for computing the number of perfect matchings
in $K_{3,3}$-free graphs and related problems.
\newblock \emph{Information and Computation}, 80(2):152--164, 1989.
\newblock \href{https://doi.org/10.1016/0890-5401(89)90017-5}{doi:10.1016/0890-5401(89)90017-5}.

\end{thebibliography}
\end{document}